\documentclass[12pt]{article}

\usepackage{preamble_paper}

\title{Identity-Compatible Auctions\thanks{
    This paper was previously titled ``Identity-Proof Auctions.'' I thank my advisor, Marek Pycia, for continuous guidance and support, and Samuel Häfner for many useful discussions. For their comments, I thank Christoph Carnehl, Christian Ewerhart, Piotr Dworczak, Péter Esö, William Fuchs, Wei He, Ian Jewitt, Navin Kartik, Bettina Klaus, Paul Klemperer, Qingmin Liu, Margaret Meyer, Nick Netzer, Alessandro Pavan, Andrea Prat, Tim Roughgarden, Armin Schmutzler, Shigehiro Serizawa, Ludvig Sinander, Vasiliki Skreta, Roland Strausz, Balázs Szentes, Éva Tardos, Alex Teytelboym, Tong Li, Fernando Vega-Redondo, Jiawei Zhang, and seminar and conference audiences at the University of Zurich, the 2024 Conference on Mechanism and Institution Design, the 8th Swiss Theory Day, University of Oxford, the 14th Conference on Economic Design, the 1st Berlin MT \& BE PhD Conference, the 2025 World Congress of the Econometric Society, Hebrew University, Tel Aviv University, Peking University, Shanghai Jiao Tong University, Tsinghua University, Mannheim University, The Chinese University of Hong Kong, Shanghai University of Finance and Economics, Universitat Pompeu Fabra, Vanderbilt University, UC Berkeley, The University of Hong Kong, Columbia Business School, Chicago Booth Market Design Conference, and Yale University. First presentation: March 2022. This paper is based on the first chapter of my Ph.D. dissertation at the University of Zurich. All errors are my own.
  }
}

\author{Haoyuan Zeng\thanks{Yale University. Email: haoyuan.zeng@yale.edu.}
}

\date{\today \\ \href{https://haoyuanzeng.github.io/papers/ICA-Zeng.pdf}{Click here for the latest version.}}

\newcommand{\E}{\mathbb E}
\newcommand{\supp}{\operatorname{supp}}
\newcommand{\Rev}{\operatorname{Rev}}
\newcommand{\R}{\mathbb R}
\newcommand{\N}{\mathbb N}
\newcommand{\cond}{\;\middle|\;}
\newcommand{\one}{\mathbf 1}

\newcommand{\dd}{\,d}

\newcommand{\onlineappendix}{%
  \clearpage
  \setcounter{section}{0}%
  \renewcommand{\thesection}{OA\arabic{section}}%
  \renewcommand{\theHsection}{OA.\arabic{section}}
  \renewcommand{\thetable}{OA\arabic{table}}%
  \renewcommand{\thefigure}{OA\arabic{figure}}%
}

\begin{document}

\maketitle

\begin{abstract}
  This paper studies sellers' and buyers' incentives to engage in shill bidding in single-object auctions. An auction is seller identity-compatible if the seller cannot profit by posing as one or more bidders through fictitious identities, and buyer identity-compatible if no buyer can profit by posing as multiple bidders. We distinguish lit auctions, which reveal the number of bidders, from dark auctions, which conceal it. Our results highlight the importance of concealing the number of bidders in deterring shill bidding: no optimal lit auction can achieve our design goal, whereas the dark first-price auction uniquely does.
\end{abstract}


\JELcodes{D44, D47, D82, D83}

\newpage
\section{Introduction}

Shill bidding has a long history. Even when auctions required physical attendance, buyers or sellers could participate through agents who communicated using subtle signals, ranging from ``a wink, a nod, scratching an ear'' to ``tugging the coat of the auctioneer'' \citep{Cassady-AuctionsAuctioneering-1967}. Modern auction markets make the identities behind bids even harder to verify. At major auction houses such as Christie's and Sotheby's, most bids are placed online or by telephone rather than in person.\footnote{\label{fn:empty}\href{https://www.wsj.com/articles/with-absentee-bidding-on-the-rise-auction-rooms-seem-empty-these-days-1402683887}{``Why Auction Rooms Seem Empty These Days,'' \emph{The Wall Street Journal}, June 15, 2014}.} A long-standing practice in auction houses is ``chandelier bidding,'' in which auctioneers announce fictitious bids to stimulate competition. New York City previously permitted such bids only below the seller's undisclosed reserve price but repealed this restriction in 2022, raising concerns about transparency and trust.\footnote{\href{https://www.nytimes.com/2022/05/03/arts/design/nyc-auction-rules-sothebys-christies.html}{``New York City Eliminates the Rules That Govern Art and Other Auctions,'' \emph{The New York Times}, May 3, 2022}.}

Internet auctions are particularly susceptible to shill bidding because participants trade anonymously, making their identities difficult to verify.\footnote{\href{https://www.nytimes.com/2001/03/09/business/3-men-are-charged-with-fraud-in-1100-art-auctions-on-ebay.html}{``3 Men Are Charged With Fraud In 1,100 Art Auctions on EBay,'' \emph{The New York Times}, March 9, 2001}.} Although shill bidding is typically associated with sellers, buyers can also profit by posing as multiple bidders. Example~\ref{ex:two-stage-auction} below shows how buyers can use shills to crowd out competitors in multistage auctions.

Shill bidding also threatens the reputations of auction marketplaces. Despite deregulation in New York City, major auction houses have expressed their intention to continue following the previous rules \citep{Herrera-NYCAuctionHouses-2023}. Similarly, eBay states that it prohibits shill bidding to ``maintain a fair marketplace for all our users'' and monitors bidding patterns to detect manipulation; see \href{https://www.ebay.com/help/policies/selling-policies/selling-practices-policy/shill-bidding-policy?id=4353}{eBay, ``Shill Bidding Policy.''} Nevertheless, shill bidding remains a concern.\footnote{\citet{Chen-Liang-Chang-Liu-Yin-EtAl-HowSeriousShill-2020} estimate that approximately 9\% of bidders in their sample of eBay Motors auctions are shill bidders and that 22\% of listings contain shill bids.} Moreover, marketplaces themselves may face commitment and credibility problems \citep{Akbarpour-Li-CredibleAuctionsTrilemma-2020}. For example, Google, a major operator in online display advertising, has faced litigation alleging manipulation of advertising auctions through practices resembling shill bidding.\footnote{\href{https://www.wsj.com/articles/google-misled-publishers-and-advertisers-unredacted-lawsuit-alleges-11642176036}{``Google Misled Publishers and Advertisers, Unredacted Lawsuit Alleges,'' \emph{The Wall Street Journal}, January 14, 2022}.\label{fn:google}}

These concerns motivate our central question: Which auction formats deter shill bidding by design? We address this question through the notion of \textbf{identity compatibility}, which requires that no participant can profit by bidding through fictitious identities. As \citet{Hurwicz-InformationallyDecentralizedSystems-1972} puts it, standard incentive compatibility ``requires that no one should find it profitable to `cheat,' where cheating is defined as behavior that can be made to look `legal' by a misrepresentation of a participant's preferences or endowment.'' Identity compatibility extends this principle to identities: no participant should profit by misrepresenting who is bidding. We distinguish between \textbf{buyers}, who seek to purchase the object on their own behalf, and \textbf{bidders}, the identities through which buyers or the seller participate in the auction. We assume that creating fictitious identities is costless and that participants know which identities they themselves control but not which identities others control.

We study shill bidding in single-object auctions with a \textbf{random} number of buyers, extending the symmetric independent private values benchmark of \citet{Myerson-OptimalAuctionDesign-1981} to population uncertainty. For now, we assume that the value distribution is regular and that the number of buyers is Poisson-distributed and independent of buyers' values. Population uncertainty is essential: if the number of buyers were fixed and commonly known, shill bidding would be detectable whenever the number of bidders exceeded that number. Our framework thus departs from standard auction theory \citep{Vickrey-CounterspeculationAuctionsCompetitive-1961,Myerson-OptimalAuctionDesign-1981,Milgrom-Weber-TheoryAuctionsCompetitive-1982}, which typically assumes common knowledge of a fixed number of buyers.

Under population uncertainty, the auction is designed before the number of buyers is realized.\footnote{Attracting buyers is a key consideration in practical auction design \citep{Klemperer-WhatReallyMatters-2002,Milgrom-PuttingAuctionTheory-2004,Milgrom-DiscoveringPricesAuction-2017}, as competition from additional buyers can substantially increase revenue \citep{Holt-CompetitiveBiddingContracts-1980,Bulow-Klemperer-AuctionsNegotiations-1996,McAfee-McMillan-AuctionsBidding-1987}.} The auction must therefore accommodate different population sizes. We model it as a collection of auctions, one for each fixed population size. In the presence of shill bidding, the game proceeds as follows.

The seller first publicly announces the auction. A random number of buyers then arrives, after which the seller and buyers simultaneously choose how many fictitious identities to create, thereby determining the total number of bidders. The auction is then run with this fixed number of bidders. The auction format may depend on the number of bidders: for example, a first-price auction may be conducted when there are at most five bidders and a second-price auction otherwise.

Shill bidding allows participants to manipulate an auction through two channels. First, by changing the number of bidders, participants can \textbf{indirectly} affect others' bidding strategies, either by changing the auction format---for example, triggering a switch from a first-price to a second-price auction---or, more importantly, by altering their \textbf{perceived competition}. For instance, in a first-price auction, equilibrium bids increase with the number of bidders. Second, shills can \textbf{directly} manipulate auction outcomes. The seller can raise buyers' payments through shill bids, as in a second-price auction, whereas buyers can exploit the expanded strategy space created by multiple identities to undertake profitable coordinated deviations; see Example~\ref{ex:two-stage-auction} below.

As illustrated above, shill bidding makes the number of bidders \textbf{endogenous}. A key innovation of this paper is to treat disclosure of the number of bidders as an information-design choice. We distinguish between \textbf{lit} auctions, in which the number of bidders is publicly revealed, and \textbf{dark} auctions, in which it is concealed. A key insight is that concealing the number of bidders is essential for deterring shill bidding.

To study shill bidding, we introduce several notions of identity compatibility for the seller and buyers in lit and dark auctions. \textbf{Bayesian buyer and seller identity compatibility} require, respectively, that buyers and the seller cannot profit from shill bidding in expectation. Our main analysis focuses on Bayesian buyer identity compatibility and the stronger \textbf{ex-post seller identity compatibility}. This distinction is motivated by the seller's ability to acquire additional information about buyers' values in practice (see Section~\ref{sec:ex-post-seller}). Ex-post seller identity compatibility allows the seller to condition shill-bidding strategies on the buyers' value profile. In dark auctions, the seller may also condition the number of shills on this profile.

Our first two main results, Theorems~\ref{thm:lit-regular-ir} and~\ref{thm:lit-regular-bb}, establish that no optimal lit auction can satisfy ex-post seller identity compatibility together with either ex-post individual rationality (ex-post IR) or Bayesian buyer identity compatibility. Together, these results reveal a fundamental design \textbf{dilemma} in lit auctions. On the one hand, competition is essential for revenue maximization. On the other hand, perceived competition can be manipulated through shill bidding. The intuition that greater perceived competition increases the seller's revenue, as in a lit first-price auction, fails without ex-post individual rationality. For example, in the optimal lit all-pay auction, shill bidding by the seller reduces expected revenue and is therefore unprofitable. Buyers, however, can profit from introducing shills to discourage their rivals from bidding aggressively; see Example~\ref{ex:auctions} below. Thus, the design dilemma in lit auctions concerns not only the seller's incentives but also those of buyers.

Our main positive result shows that this dilemma can be fully resolved in dark auctions. Theorem~\ref{thm:dark-regular} characterizes the dark first-price auction (with the Myersonian reserve) as essentially the \textbf{unique} optimal dark auction satisfying ex-post seller identity compatibility, Bayesian buyer identity compatibility, and ex-post IR. The intuition rests on two complementary features. First, the first-price rule prevents the seller from raising the winner's payment through shill bidding, even when the seller knows buyers' values. Second, the dark format prevents the seller from manipulating buyers' perceived competition and thereby inducing more aggressive bidding. The uniqueness result further implies that, among optimal dark auctions satisfying ex-post IR, any departure from the dark first-price rule creates profitable shill-bidding opportunities for either the seller or buyers.

The contrast between the impossibility results for lit auctions and the attainability of optimality in dark auctions raises a further question: how much revenue can a lit auction retain while deterring shill bidding? We focus on the seller's incentives, since additionally imposing buyer identity compatibility would further restrict the set of feasible auctions and therefore cannot increase the maximum attainable revenue. We strengthen ex-post seller identity compatibility and ex-post IR to account for the seller's ability to manipulate the auction's randomization device, for example, by breaking ties in their favor. Theorem~\ref{thm:lit-posted-price} shows that a posted-price mechanism maximizes revenue among lit auctions satisfying these stronger requirements, whereas Theorem~\ref{thm:dark-regular} remains valid under them. Thus, deterring shill bidding in lit auctions requires more than sacrificing optimality, which remains attainable in dark auctions: the seller must forgo the ability to extract additional revenue from realized competition. The design dilemma in lit auctions therefore takes a particularly stark form: the seller must commit to a rigid payment rule that effectively ties their own hands.

Throughout this paper, we focus on identity-compatible auctions. Theorem~\ref{thm:non-identity-compatible} provides a revelation-principle-type justification for this focus: every equilibrium outcome of a non-identity-compatible auction in which buyers, the seller, or both may use multiple identities can be implemented by a direct dark auction that is Bayesian identity-compatible for each side permitted to use multiple identities. This result shows why the intuition that the seller prefers non-identity-compatible auctions because shill bidding is profitable is incomplete: it holds buyers' strategies fixed and ignores how they respond in equilibrium.

Beyond our main results, we extend the analysis to irregular value distributions (Theorems~\ref{thm:lit-irregular-ir} and~\ref{thm:lit-irregular-bb}) and to environments with interdependent values and correlated signals (Theorem~\ref{thm:interdependent}). Our central insights continue to hold in these settings under additional mild conditions.

To the author's knowledge, this paper is the first to identify the importance of concealing the number of bidders in deterring shill bidding. In our environment, disclosure of the number of bidders becomes irrelevant for revenue maximization in the absence of shill bidding (Proposition~\ref{prop:ret}). The importance of concealment therefore arises precisely from the possibility of shill bidding. Previous work has recognized that transparency in auctions is not always beneficial \citep{Cramton-Schwartz-CollusiveBiddingLessons-2000, Klemperer-WhatReallyMatters-2002, Klemperer-UsingAbusingEconomic-2003}.\footnote{In the Hong Kong 3G auction, transparency was abused by strong bidders to lobby for a ``small'' design change, which ultimately made entry harder and collusion easier. Interestingly, they ridiculed the original design by calling it the ``dark auction'' (see Footnote~48 in \citet{Klemperer-UsingAbusingEconomic-2003}). \label{fn:transparency}} Our results reinforce this point from a novel perspective, showing that when the disclosed information---specifically, the number of bidders---can be manipulated through fictitious identities to mislead participants, it is better to conceal it.

\paragraph{Related Literature}

This paper is far from the first to study shill bidding in auctions. Our main novelty lies in distinguishing between lit and dark auctions in the presence of shill bidding. Previous work has largely focused on auction formats that are lit in our terminology, whereas our results show that restricting attention to such formats entails a loss of generality. While the existing literature typically studies shill bidding by sellers and buyers separately, we consider both simultaneously. For example, Theorems~\ref{thm:lit-regular-bb} and~\ref{thm:dark-regular} explicitly allow both sides to engage in shill bidding.

The closest related work consists of parallel studies by \citet*{Komo-Kominers-Roughgarden-ShillProofAuctions-2025} and \citet{Shinozaki-ShillProofRulesObject-2025}, both of which study similar ex-post notions of robustness to shill bidding by the seller. \citet*{Komo-Kominers-Roughgarden-ShillProofAuctions-2025} is closely related to our Theorem~\ref{thm:dark-regular}, as both papers establish the robustness of the first-price rule to shill bidding. The main difference is that they take the number of bidders as \emph{exogenously fixed} and focus on dynamic auctions in a \emph{public} setting, allowing them to distinguish the lit Dutch auction from the lit first-price auction. With a fixed number of bidders, the seller cannot manipulate perceived competition by introducing additional shills. Their analysis therefore emphasizes dynamic implementation, whereas ours allows the seller to introduce arbitrarily many shills and highlights the importance of concealing the number of bidders. We also briefly discuss dynamic implementation in connection with Theorem~\ref{thm:non-identity-compatible}, but without restricting attention to public auctions.

\citet{Shinozaki-ShillProofRulesObject-2025} is closely related to our Theorem~\ref{thm:lit-posted-price}, as both papers establish the robustness of posted prices to shill bidding by the seller. The main difference lies in the incentive and optimality criteria: \citet{Shinozaki-ShillProofRulesObject-2025} imposes \emph{strategy-proofness} for buyers and adopts a \emph{prior-free} revenue-dominance approach, whereas we maximize expected revenue under a prior without requiring strategy-proofness. Thus, both papers establish the importance of posted prices, but under different incentive requirements and revenue objectives.

Several papers study shill bidding by the seller in second-price or English auctions \citep*{Graham-Marshall-Richard-PhantomBiddingHeterogeneous-1990,Rothkopf-Teisberg-Kahn-WhyAreVickrey-1990,Rothkopf-Harstad-TwoModelsBidTaker-1995,Izmalkov-ShillBiddingOptimal-2004,Chakraborty-Kosmopoulou-AuctionsShillBidding-2004,Porter-Shoham-CheatingSealedbidAuctions-2005,McAdams-Schwarz-WhoPaysWhen-2007,Lamy-ShillBiddingEffect-2009,Lamy-UppingAnteHow-2013,Levin-Peck-MisbehaviorCommonValue-2023}. We contribute to this literature by introducing \emph{identity compatibility} as a design criterion and adopting a mechanism-design approach without restricting attention to a particular auction format. Indeed, the distinction between lit and dark auctions is less salient when attention is restricted to second-price or English auctions. While much of the existing literature studies the seller's Bayesian incentives to engage in shill bidding, our main analysis focuses on ex-post seller identity compatibility.

Several papers study shill bidding by buyers in VCG mechanisms \citep*{Sakurai-Yokoo-Matsubara-LimitationGeneralizedVickrey-1999,Yokoo-Sakurai-Matsubara-EffectFalsenameBids-2004,Ausubel-Milgrom-AscendingAuctionsPackage-2002,Ausubel-Milgrom-LovelyLonelyVickrey-2005,Sher-OptimalShillBidding-2012,Arnosti-Beck-Milgrom-AdverseSelectionAuction-2016}. Much of this literature focuses on how complementarities in combinatorial auctions create incentives for shill bidding. We identify a different channel that arises even in single-object auctions: in lit all-pay auctions, buyers can introduce shills to inflate their rivals' perceived competition, thereby discouraging them from bidding aggressively. This channel does not rely on complementarities and has received little attention in the literature.

Our paper also relates to the literature on mechanisms with imperfect commitment. Credibility \citep{Akbarpour-Li-CredibleAuctionsTrilemma-2020} and auditability \citep{Woodward-SelfAuditableAuctions-2020} concern the seller's ability to misrepresent buyers' preferences or auction histories through bilateral communication. Other papers study settings in which the designer cannot commit to certain components of the outcome \citep*{Baliga-Corchon-Sjostrom-TheoryImplementationWhen-1997,Bester-Strausz-ImperfectCommitmentRevelation-2000,Bester-Strausz-ContractingImperfectCommitment-2001}, or cannot commit to future actions in sequential auctions and bargaining games \citep*{Milgrom-AuctionTheory-1989,McAfee-Vincent-SequentiallyOptimalAuctions-1997,McAdams-Schwarz-CredibleSalesMechanisms-2007,Skreta-SequentiallyOptimalMechanisms-2006,Skreta-OptimalAuctionDesign-2015,Vartiainen-AuctionDesignCommitment-2013,Liu-Mierendorff-Shi-Zhong-AuctionsLimitedCommitment-2019}. We identify a different form of imperfect commitment arising from population uncertainty. An auction specifies an outcome rule for each number of bidders. By introducing shills, the seller can manipulate the number of bidders and thereby determine which outcome rule is applied.

We also contribute to the literature on auctions with population uncertainty \citep{Myerson-PopulationUncertaintyPoisson-1998,Milchtaich-RandomplayerGames-2004}.\footnote{For auctions with endogenous entry, see also \citet*{Samuelson-CompetitiveBiddingEntry-1985}, \citet{McAfee-McMillan-AuctionsEntry-1987}, \citet{Harstad-AlternativeCommonValueAuction-1990}, \citet{Levin-Smith-EquilibriumAuctionsEntry-1994}, \citet{Ye-IndicativeBiddingTheory-2007}, and \citet{Li-Zheng-EntryCompetitionEffects-2009}.} Previous work shows that risk aversion and affiliation can make disclosure of the number of bidders relevant for revenue maximization \citep*{McAfee-McMillan-AuctionsStochasticNumber-1987,Matthews-ComparingAuctionsRisk-1987,Harstad-Kagel-Levin-EquilibriumBidFunctions-1990,Murto-Valimaki-AffiliatedCommonValue-2025}. \citet{Lauermann-Wolinsky-BidderSolicitationAdverse-2017,Lauermann-Wolinsky-CommonvalueAuctionStatedependent-2022} and \citet{Lauermann-Speit-BiddingCommonValueAuctions-2023} study dark first-price auctions with common values. We show that concealing the number of bidders can be essential for deterring shill bidding while maximizing revenue, even in the standard independent private values benchmark, where disclosure would otherwise be irrelevant.


\paragraph{Outline of the Paper}
Section~\ref{sec:model} presents the model with a random number of players and applies it to single-object auctions with independent private values. Section~\ref{sec:identity} introduces shill bidding, defines several notions of identity compatibility, and illustrates them with examples. Section~\ref{sec:results} presents our main results. Section~\ref{sec:discussion} discusses further extensions, and Section~\ref{sec:conclusion} concludes. Omitted proofs are collected in Appendix~\ref{app:proofs} and Online Appendix~\hyperref[app:online]{OA}.

\section{Model}
\label{sec:model}

This section develops the baseline model \emph{without} shill bidding; Section~\ref{sec:identity} introduces shill bidding. We first consider an environment with a random number of players and define lit and dark mechanisms. We then specialize the framework to single-object auctions with independent private values.

\subsection{Environment}

We consider an environment with a finite but random number of players, as in \citet{Myerson-PopulationUncertaintyPoisson-1998} and \citet{Milchtaich-RandomplayerGames-2004}. We introduce only the features of this environment needed for our analysis; for a formal treatment of games with a random number of players, see \citet{Milchtaich-RandomplayerGames-2004}. Let \(\Theta\) denote a standard Borel type space for players, endowed with \(\sigma\)-algebra \(\mathcal F_\Theta\). Conditional on their own type \(\theta\in\Theta\), each player holds a posterior belief over the number and types of the other players. We impose the following common-posterior assumption.

\begin{definition}[Common Posterior]\label{def:common-posterior}
  All players, regardless of type, hold the same posterior belief over the number and types of the other players.
\end{definition}

This assumption is weaker than the environmental equivalence property of \citet{Myerson-PopulationUncertaintyPoisson-1998}, because we do not require the players' common posterior to coincide with the prior of an external game theorist. Poisson games are therefore a special case of our environment.

Let \(\tilde n\) denote the number of players, a nonnegative integer-valued random variable with finite mean \(\lambda > 0\). Let \(\gamma_n=\Pr(\tilde n=n)\) denote the prior probability that there are \(n\) players, and let \(p_n\) denote a player's posterior probability that the total number of players, including themselves, is \(n\). Under the common-posterior assumption, the random-player environment has the following structure.

\begin{proposition}[\citet{Kallenberg-RandomMeasures-1986}, \citet{Milchtaich-RandomplayerGames-2004}]\label{prop:common-posterior}
  Under common posterior, players' types are independent and identically distributed. Moreover, their common posterior over the total number of players is given by
  \[
    p_n=\frac{n\gamma_n}{\lambda}.
  \]
\end{proposition}

Intuitively, the environment can be generated by first drawing the number of players and then drawing that many types independently from a common probability measure \(D:\mathcal F_\Theta\to[0,1]\) that does not depend on the realized population size.\footnote{This formulation was suggested by Paul Milgrom; see \citet[Footnote~3]{McAfee-McMillan-AuctionsStochasticNumber-1987}.} Poisson games constitute the special case in which \(\gamma_n=p_{n+1}=\frac{\lambda^n}{n!}e^{-\lambda}\).\footnote{Here, \(p_{n+1}\) appears because a player who faces \(n\) other players belongs to a population of \(n+1\) players.}

\subsection{Games with a Random Number of Players}
\label{sec:game}

As \citet[p.~377]{Myerson-PopulationUncertaintyPoisson-1998} observes, ``population uncertainty forces us to treat players symmetrically in our game-theoretic analysis.'' Accordingly, we restrict attention to anonymous games and symmetric equilibria, so that neither the rules of the game nor players' strategies depend on player labels. We distinguish between two classes of games.

\begin{definition}[Lit and Dark Games]\label{def:lit-dark}
  A game with a random number of players is \textbf{lit} if the realized number of players is publicly revealed to all players. It is \textbf{dark} if the realized number of players is concealed from them.
\end{definition}

Whether a game is lit or dark, it specifies an underlying \(n\)-player game for each population size \(n\in\N\). We therefore begin by defining the component \(n\)-player games and then build lit and dark mechanisms from these games and their associated strategies.

Conditional on \(n\) players being realized, we may, without loss of generality, label them \(1,2,\ldots,n\). We use \(i\) to denote a generic player. Let \(\theta^n=(\theta_1,\theta_2,\ldots,\theta_n)\in\Theta^n\) denote a realized type profile, and let \(\theta^n_{-i}=(\theta_j)_{j\neq i}\in\Theta^{n-1}\) denote the profile of types of players other than \(i\). For notational convenience, when a statement ranges over different possible population sizes, with a slight abuse of notation, we continue to use \(i\) to denote a tagged player and \(\theta^n_{-i}\) to denote the types of the other \(n-1\) players. Thus, in such statements, \(i\) refers to the tagged player's role rather than to a fixed numerical label across population sizes. Throughout the paper, tildes denote random variables: \(\tilde n\) denotes the random number of players and \(\tilde\theta_i\) the random type of player \(i\), whereas \(n\) and \(\theta_i\) denote their realizations.

Let \(X^n\) denote the set of deterministic outcomes when there are \(n\) players, and let \(\Delta(X^n)\) denote the set of lotteries over \(X^n\). We assume that utility functions \(u_i:\bigcup_{n\in\N}(X^n\times\Theta^n)\to\R\) are permutation invariant: relabeling players, together with the corresponding components of the outcome and type profile, leaves each player's utility unchanged. We extend utilities to lotteries over \(X^n\) by expected utility.

For each \(n\in\N\), let \(\Gamma^n=(\Sigma^n,\varphi^n)\) denote an anonymous \(n\)-player game:
\begin{itemize}

  \item Let \(\Sigma_i^n\) denote player \(i\)'s interim strategy space in the \(n\)-player game. By anonymity, \(\Sigma_i^n=\Sigma^n\) for every \(i\).\footnote{To keep the exposition simple, we relegate the explicit extensive-form formulation to Online Appendix~\ref{app:extensive-form-games}. In a static game, \(\Sigma^n\) can simply be interpreted as the set of messages available to each player, who sends a single message. In a dynamic game, \(\Sigma^n\) can instead be interpreted as the set of contingent plans specifying actions at each observed history.}

  \item An anonymous outcome rule \(\varphi^n:(\Sigma^n)^n\to\Delta(X^n)\) maps each interim strategy profile into a possibly randomized outcome and is invariant to permutations of player labels.

\end{itemize}

If player \(i\), of type \(\theta_i\), plays strategy \(\sigma_i^n\in\Sigma^n\), while the other players have type profile \(\theta^n_{-i}\) and play according to \(\sigma^n_{-i}\in(\Sigma^n)^{n-1}\), then player \(i\)'s payoff in \(\Gamma^n\) is

\[
  u_i^{\Gamma^n}(\sigma_i^n,\sigma^n_{-i},\theta_i,\theta^n_{-i})
  :=
  u_i\!\left(\varphi^n(\sigma_i^n,\sigma^n_{-i}),\theta_i,\theta^n_{-i}\right).
\]

\paragraph{Lit Mechanisms.}

In a lit game, conditional on \(n\) players being realized, the players observe \(n\) and know that they are playing \(\Gamma^n\). Thus, \(\Gamma^n\) induces a Bayesian game among the realized players. A symmetric ex-ante strategy in the \(n\)-player game is a function \(S^n:\Theta\to\Sigma^n\) that assigns an interim strategy to each type. We write player \(i\)'s prescribed interim strategy as \(S_i^n(\theta_i):=S^n(\theta_i)\). A sequence of games \(\Gamma^{\N}:=(\Gamma^n)_{n\in\N}=\bigl((\Sigma^n,\varphi^n)\bigr)_{n\in\N}\), together with a sequence of symmetric ex-ante strategies \(S^{\N}:=(S^n)_{n\in\N}\), constitutes a lit mechanism \((\Gamma^{\N},S^{\N})\).

\paragraph{Dark Mechanisms.}

In a dark game, players do not observe the realized number of players and cannot infer it from their observed histories at any decision stage. We therefore impose a common interim strategy space \(\Sigma^d\) across all \(n\)-player games and require each player's choice of interim strategy to depend only on their type.\footnote{Dark mechanisms may themselves be dynamic, in which case an interim strategy specifies actions contingent on observed histories. These histories may allow players to update their beliefs about population size but cannot reveal its exact realization while they still have decisions to make. For example, a dark English auction may reveal only the current clock price while concealing the number of active bidders.} For each \(n\in\N\), let \(\Gamma^{d,n}:=(\Sigma^d,\varphi^{d,n})\) denote the \(n\)-player component of the dark game, where the outcome rule \(\varphi^{d,n}:(\Sigma^d)^n\to\Delta(X^n)\) may depend on the realized population size. A dark game is therefore \(\Gamma^d:=(\Gamma^{d,n})_{n\in\N}=\bigl((\Sigma^d,\varphi^{d,n})\bigr)_{n\in\N}\). Thus, relative to a lit game, a dark game has a common interim strategy space across population sizes, while its outcome rule may still vary with the realized population size. A symmetric ex-ante strategy is a function \(S^d:\Theta\to\Sigma^d\) assigning an interim strategy to each type. The game \(\Gamma^d\), together with \(S^d\), constitutes a dark mechanism \((\Gamma^d,S^d)\).

The distinction between lit and dark mechanisms determines the information available to players when they choose their strategies and, consequently, the relevant incentive constraints.\footnote{Throughout, the expected payoff under the prescribed strategies is assumed finite, and the expected payoff from every feasible deviation is assumed well defined, possibly equal to \(+\infty\) or \(-\infty\).\label{fn:existence}}

\begin{definition}[Bayesian Incentive Compatibility]{~}
  \label{def:BIC}
  \begin{itemize}[leftmargin=*]
    \item A lit mechanism \((\Gamma^{\N},S^{\N})\) is \textbf{Bayesian incentive-compatible (BIC)} if, for every player \(i\), every \(\theta_i\in\Theta\), every \(n\in\N\), and every \(\sigma_i^n\in \Sigma^n\),
      \[
        \E_{\tilde\theta^n_{-i}}\!\left[
          u_i^{\Gamma^n}\!\left(
            S_i^n(\theta_i),
            S_{-i}^n(\tilde\theta^n_{-i}),
            \theta_i,
            \tilde\theta^n_{-i}
          \right)
        \right]
        \geq
        \E_{\tilde\theta^n_{-i}}\!\left[
          u_i^{\Gamma^n}\!\left(
            \sigma_i^n,
            S_{-i}^n(\tilde\theta^n_{-i}),
            \theta_i,
            \tilde\theta^n_{-i}
          \right)
        \right].
      \]

    \item A dark mechanism \((\Gamma^d,S^d)\) is \textbf{Bayesian incentive-compatible (BIC)} if, for every player \(i\), every \(\theta_i\in\Theta\), and every \(\sigma_i^d\in \Sigma^d\),
      \[
        \E_{\tilde\theta^{\tilde{n}}_{-i}}\!\left[
          u_i^{\Gamma^{d,\tilde{n}}}\!\left(
            S_i^d(\theta_i),
            S_{-i}^d(\tilde\theta^{\tilde{n}}_{-i}),
            \theta_i,
            \tilde\theta^{\tilde{n}}_{-i}
          \right)
        \right]
        \geq
        \E_{\tilde\theta^{\tilde{n}}_{-i}}\!\left[
          u_i^{\Gamma^{d,\tilde{n}}}\!\left(
            \sigma_i^d,
            S_{-i}^d(\tilde\theta^{\tilde{n}}_{-i}),
            \theta_i,
            \tilde\theta^{\tilde{n}}_{-i}
          \right)
        \right].
      \]
  \end{itemize}
\end{definition}

In a lit mechanism, BIC requires players' strategies to be mutual best responses conditional on the realized number of players. In a dark mechanism, by contrast, players choose their strategies without observing the realized population size, so the incentive constraints are evaluated by averaging over population uncertainty.

\begin{observation}\label{obs:lit-dark}
  Every BIC lit mechanism is outcome-equivalent to a BIC direct dark mechanism.
\end{observation}

Outcome equivalence means that the two mechanisms implement the same outcome rule, that is, the same mapping from type profiles to outcomes. A \textbf{direct} lit or dark mechanism is one in which the strategy space coincides with the type space, so that \(\Sigma^n=\Sigma^d=\Theta\) for all \(n\in\N\).

To establish the observation, apply the Revelation Principle \citep{Myerson-OptimalAuctionDesign-1981} separately to each \(n\)-player game in a BIC lit mechanism. For every \(n\), this yields an outcome-equivalent direct mechanism in which truthful reporting is optimal conditional on \(n\). Because the strategy space in a direct mechanism is the type space and is therefore common across all population sizes, we can conceal the realized population size and regard the resulting collection of direct mechanisms as a dark mechanism. Truthful reporting remains optimal after averaging over population uncertainty because it is optimal conditional on every \(n\). Hence, every outcome rule implementable by a BIC lit mechanism can also be implemented by a BIC direct dark mechanism. The converse need not hold.\footnote{The same observation applies under partial disclosure of the realized population size: every outcome rule implementable by a BIC mechanism under partial disclosure can also be implemented by a BIC direct dark mechanism.}

\subsection{Auctions with a Random Number of Buyers}

We now apply the model to single-object auctions with a random number of buyers, a natural extension of the benchmark independent private values (IPV) environment of \citet{Myerson-OptimalAuctionDesign-1981}. From this point onward, we refer to the players as \textbf{buyers}.

There is a single indivisible object for sale, and \(\Theta=[0,1]\). Buyer \(i\)'s private type \(\theta_i \in \Theta\) is their value for the object. Under the common-posterior assumption (Definition~\ref{def:common-posterior}) and Proposition~\ref{prop:common-posterior}, buyers' values are independently and identically distributed and independent of the number of buyers. We assume that the value distribution \(F\) has a continuous, positive density \(f:\Theta\to\R_{>0}\). The number of buyers \(\tilde n\) has finite mean \(\lambda\in\R_{>0}\), with \(\gamma_n=\Pr(\tilde n=n)\). A buyer's posterior probability that the total number of buyers is \(n\) is \(p_n=n\gamma_n/\lambda\).

When there are \(n\) realized buyers, a deterministic outcome is
\[
  x^n \in X^n = \left\{ (q^n,t^n) \in \{0,1\}^n \times \R^n \cond \sum_{i=1}^n q_i^n \leq 1 \right\},
\]
where \(\sum_{i=1}^n q_i^n=0\) denotes no sale, \(q_i^n=1\) denotes allocation to buyer \(i\), and \(t^n=(t_1^n,\ldots,t_n^n)\) is the payment vector. The auction may randomize over \(X^n\), so its outcome lies in \(\Delta(X^n)\).

Buyers have private values and quasi-linear utilities,
\[
  u_i\!\left((q^n,t^n),\theta^n\right)
  =
  \theta_iq^n_i-t_i^n.
\]
The seller, denoted by \(0\), maximizes revenue, and their value for the object is normalized to zero:\footnote{This normalization entails little loss even if the seller's value, or cost, is private information, because we assume that the seller publicly announces the auction format. The substantive change is that the optimal auction allocates the object only when the highest ironed virtual value exceeds the seller's cost. For the case in which the seller's cost is private information and is communicated bilaterally to buyers, see \citet*{Banchio-Skrzypacz-Yang-DynamicThreatsCredible-2025}.}
\[
  u_0\!\left((q^n,t^n),\theta^n\right)
  =
  \sum_{i=1}^n t_i^n.
\]

For each population size \(n\), let \(q^{\Gamma^n}\) and \(t^{\Gamma^n}\) denote the allocation and payment rules of the lit \(n\)-buyer game, and let \(q^{\Gamma^{d,n}}\) and \(t^{\Gamma^{d,n}}\) denote their dark counterparts. Thus,
\[
  (q^{\Gamma^n},t^{\Gamma^n}):(\Sigma^n)^n\to\Delta(X^n),
  \qquad
  (q^{\Gamma^{d,n}},t^{\Gamma^{d,n}}):(\Sigma^d)^n\to\Delta(X^n).
\]
With a slight abuse of notation, when the outcome rule is randomized, \(q_i^{\Gamma^n}(\cdot)\) and \(t_i^{\Gamma^n}(\cdot)\) denote buyer \(i\)'s expected allocation and payment under the induced outcome lottery, with the dark counterparts defined analogously.

Together with the prescribed strategies, these rules induce allocation and payment rules on value profiles:
\[
  (q^{\Gamma^n,S^n},t^{\Gamma^n,S^n}):\Theta^n\to\Delta(X^n),
  \qquad
  (q^{\Gamma^{d,n},S^d},t^{\Gamma^{d,n},S^d}):\Theta^n\to\Delta(X^n).
\]
For example, in the lit auction,
\[
  q_i^{\Gamma^n,S^n}(\theta^n)
  :=
  q_i^{\Gamma^n}\!\left(S^n(\theta^n)\right)\in[0,1],
  \qquad
  t_i^{\Gamma^n,S^n}(\theta^n)
  :=
  t_i^{\Gamma^n}\!\left(S^n(\theta^n)\right)\in\R,
\]
with the dark counterparts defined analogously.\footnote{Unless otherwise stated, conclusions identifying equilibrium allocations or payments are understood to hold almost everywhere with respect to the relevant distribution of values or value profiles. We make this qualification explicit only when the treatment of null sets matters for the argument.}

Expected revenue in lit and dark auctions is, respectively,\footnote{Throughout, we restrict attention to induced payment rules that are absolutely integrable with respect to the joint distribution of population sizes and values. Thus, for a lit auction,
  \(
    \E_{\tilde\theta^{\tilde n}}[
      \sum_{i=1}^{\tilde n}
      |t_i^{\Gamma^{\tilde n},S^{\tilde n}}(\tilde\theta^{\tilde n})|
    ]
    <\infty,
  \)
  and, for a dark auction,
  \(
    \E_{\tilde\theta^{\tilde n}}[
      \sum_{i=1}^{\tilde n}
      |t_i^{\Gamma^{d,\tilde n},S^d}(\tilde\theta^{\tilde n})|
    ]
    <\infty.
\)\label{fn:integrability}}
\begin{align*}
  \Rev(\Gamma^{\N},S^{\N})
  :=& \ \E_{\tilde\theta^{\tilde n}}\!
  \left[
    \sum_{i=1}^{\tilde n}
    t_i^{\Gamma^{\tilde n}}
    \left(S^{\tilde n}(\tilde\theta^{\tilde n})\right)
  \right] =
  \E_{\tilde\theta^{\tilde n}}\!
  \left[
    \sum_{i=1}^{\tilde n}
    t_i^{\Gamma^{\tilde n}, S^{\tilde n}}
    (\tilde\theta^{\tilde n})
  \right], \\
  \Rev(\Gamma^d,S^d)
  :=& \E_{\tilde\theta^{\tilde n}}\!
  \left[
    \sum_{i=1}^{\tilde n}
    t_i^{\Gamma^{d, \tilde n}}
    \left(S^d(\tilde\theta^{\tilde n})\right)
  \right] =
  \E_{\tilde\theta^{\tilde n}}\!
  \left[
    \sum_{i=1}^{\tilde n}
    t_i^{\Gamma^{d,\tilde n}, S^d}
    (\tilde\theta^{\tilde n})
  \right].
\end{align*}

\begin{definition}[Optimality]\label{def:optimality}
  \leavevmode
  \begin{itemize}[leftmargin=*]

    \item A lit auction \((\Gamma^{\N},S^{\N})\) is \textbf{optimal} if it maximizes \(\Rev(\Gamma^{\N},S^{\N})\) subject to:
      \begin{enumerate}
        \item \textit{Incentive compatibility}: \((\Gamma^{\N},S^{\N})\) is BIC.

        \item \textit{Individual rationality}: for every buyer \(i\), every \(\theta_i\in\Theta\), and every \(n\in\N\), buyer \(i\)'s interim expected payoff is nonnegative:
          \[
            U_i^{\Gamma^n,S^n}(\theta_i)
            :=
            \E_{\tilde\theta^n_{-i}}\!\left[
              \theta_i
              q_i^{\Gamma^n,S^n}(\theta_i,\tilde\theta^n_{-i})
              -
              t_i^{\Gamma^n,S^n}(\theta_i,\tilde\theta^n_{-i})
            \right]
            \geq 0.
          \]
      \end{enumerate}

    \item A dark auction \((\Gamma^d,S^d)\) is \textbf{optimal} if it maximizes \(\Rev(\Gamma^d,S^d)\) subject to:
      \begin{enumerate}
        \item \textit{Incentive compatibility}: \((\Gamma^d,S^d)\) is BIC.

        \item \textit{Individual rationality}: for every buyer \(i\) and every \(\theta_i\in\Theta\), buyer \(i\)'s interim expected payoff is nonnegative:
          \[
            U_i^{\Gamma^d,S^d}(\theta_i)
            :=
            \E_{\tilde\theta^{\tilde{n}}_{-i}}\!\left[
              \theta_i
              q_i^{\Gamma^{d,\tilde{n}},S^d}(\theta_i,\tilde\theta^{\tilde{n}}_{-i})
              -
              t_i^{\Gamma^{d,\tilde{n}},S^d}(\theta_i,\tilde\theta^{\tilde{n}}_{-i})
            \right]
            \geq 0.
          \]
      \end{enumerate}

  \end{itemize}
\end{definition}

\citet{Myerson-OptimalAuctionDesign-1981} characterizes the optimal auction with a fixed number of buyers. It follows immediately that an optimal lit auction \((\Gamma^{\N},S^{\N})\) can be obtained by choosing, for each \(n\), an optimal \(n\)-buyer auction \((\Gamma^n,S^n)\) in the sense of \citet{Myerson-OptimalAuctionDesign-1981}. Observation~\ref{obs:lit-dark} then implies that every optimal lit auction is outcome-equivalent to a direct dark auction: both incentive compatibility and individual rationality are preserved when the conditional constraints are averaged over population sizes. The remaining question is whether such a direct dark auction is optimal among all dark auctions.

We answer this question in the affirmative in Proposition~\ref{prop:ret} by extending the Revenue Equivalence Theorem \citep{Myerson-OptimalAuctionDesign-1981} to our environment with a random number of buyers. This extension is not itself new \citep{McAfee-McMillan-AuctionsStochasticNumber-1987}, but it requires a separate argument because the common-posterior environment considered here is not fully covered by the existing literature.\footnote{\citet[Example~5]{Milchtaich-RandomplayerGames-2004} notes that the framework of \citet{McAfee-McMillan-AuctionsStochasticNumber-1987} excludes Poisson games.} The key additional step is to invoke Proposition~\ref{prop:common-posterior}; otherwise, the standard argument applies. For completeness, we provide the proof in Online Appendix~\ref{app:proof-ret}.

\begin{definition}[Virtual Value]
  \label{def:regular}
  The \textbf{virtual value} of value \(\theta\) is \(\phi(\theta):=\theta-\frac{1-F(\theta)}{f(\theta)}\). The value distribution \(F\) is \textbf{regular} if \(\phi\) is strictly increasing. Under regularity, the Myersonian reserve is \(\rho^* := \phi^{-1}(0)\).
\end{definition}

\begin{proposition}[Revenue Equivalence]\label{prop:ret}
  Assume common posterior. Consider a direct dark auction \((\Gamma^d,S^d)\) under truthful reporting, \(S^d(\theta)=\theta\) for every \(\theta\in\Theta\). The auction is Bayesian incentive-compatible (BIC) if and only if, for every buyer \(i\),
  \begin{enumerate}
    \item the interim expected allocation
      \[
        Q_i^{\Gamma^d,S^d}(\theta_i)
        :=
        \E_{\tilde{\theta}^{\tilde n}_{-i}}
        \left[
          q_i^{\Gamma^{d,\tilde n},S^d}
          (\theta_i,\tilde{\theta}^{\tilde n}_{-i})
        \right]
      \]
      is nondecreasing in \(\theta_i\in\Theta\); and
    \item the interim expected payment
      \[
        T_i^{\Gamma^d,S^d}(\theta_i)
        :=
        \E_{\tilde{\theta}^{\tilde n}_{-i}}
        \left[
          t_i^{\Gamma^{d,\tilde n},S^d}
          (\theta_i,\tilde{\theta}^{\tilde n}_{-i})
        \right]
      \]
      satisfies
      \[
        T_i^{\Gamma^d,S^d}(\theta_i)
        =
        \theta_i Q_i^{\Gamma^d,S^d}(\theta_i)
        -
        \int_0^{\theta_i}Q_i^{\Gamma^d,S^d}(z)\dd z
        -
        U_i^{\Gamma^d,S^d}(0)
      \]
      for every \(\theta_i\in\Theta\).
  \end{enumerate}
  The expected revenue of any BIC direct dark auction is
  \[
    \Rev(\Gamma^d,S^d)
    =
    \E_{\tilde{\theta}^{\tilde n}}
    \left[
      \sum_{i=1}^{\tilde n}
      \phi(\tilde{\theta}_i)
      q_i^{\Gamma^{d,\tilde n},S^d}
      (\tilde{\theta}^{\tilde n})
    \right]
    -
    \lambda U_i^{\Gamma^d,S^d}(0).
  \]
\end{proposition}

When the distribution of \(\tilde n\) is degenerate, Proposition~\ref{prop:ret} reduces to the standard Revenue Equivalence Theorem with a fixed number of buyers. To keep the exposition simple, we relegate the corresponding explicit characterization under ironing to Online Appendix~\ref{app:proof-ret}; see Proposition~\ref{prop:optimal-direct-dark-auction-ironed}. The optimal allocation rule is the same in lit and dark auctions, up to tie-breaking and the treatment of zero (ironed) virtual values: regardless of whether the realized number of buyers is disclosed, the object is allocated to a buyer with the highest positive (ironed) virtual value. The interim expected payment is then pinned down by the envelope formula. Consequently, optimal lit and dark auctions generate the same expected revenue, although their realized revenues may differ across population sizes.

To fix terminology and ensure that the auction formats are interpreted consistently throughout, we define several standard formats---first-price, second-price, and all-pay auctions---in an environment with a random number of buyers. A lower bound on the feasible bid space serves as a reserve price.

\begin{definition}[Auction Formats]{~}
  \begin{itemize}
    \item A \textbf{lit first-price auction} is one in which, after observing the disclosed number of buyers \(n\), each buyer submits a bid from a feasible bid space \(B^n\subseteq\mathbb{R}_{\geq 0}\) or declines to bid. A buyer with the highest bid wins and pays their own bid; all other buyers pay zero.\footnote{The strategy space \(\Sigma^n\) is the union of the feasible bid space \(B^n\) and a null message \(\varnothing\) indicating that the buyer declines to bid. More generally, the feasible bid space may be designed for optimality to include caps or jumps; see \citet{Jeong-Pycia-FirstPricePrincipleFailure-2025}.}

    \item A \textbf{dark first-price auction} is defined analogously, except that buyers do not observe the number of buyers and submit bids from a common feasible bid space \(B^d\subseteq\mathbb{R}_{\geq 0}\).

    \item A \textbf{second-price auction} is defined analogously, except that a buyer with the highest bid wins and pays the maximum of the second-highest bid and the reserve price; all other buyers pay zero. A lit second-price auction may use a population-contingent reserve, whereas a dark second-price auction uses a common reserve across all population sizes. When the distinction between lit and dark is immaterial, we refer to either format simply as a second-price auction.

    \item \textbf{Lit and dark all-pay auctions} are defined analogously to their first-price counterparts, except that every buyer pays their submitted bid regardless of the allocation.

    \item We define \textbf{Dutch auctions} and \textbf{English auctions} analogously. Dutch auctions are equivalent to first-price auctions, and we therefore do not distinguish between the two in this paper. English auctions, however, are generally not equivalent to second-price auctions in the presence of shill bidding; see the discussion around Proposition~\ref{prop:second-price}.
  \end{itemize}
\end{definition}

Assume that the feasible bid space is \([\rho,+\infty)\), where \(0\leq\rho<1\) is the reserve price. Let \(G(\theta):=\sum_{n \in \N}p_nF(\theta)^{n-1}\). For all-pay auctions, let \(c_n\) and \(c_d\) be defined by \(c_nF(c_n)^{n-1}=\rho\) and \(c_dG(c_d)=\rho\), respectively.\footnote{The cutoffs \(c_n\) and \(c_d\) are the lowest values that submit bids in the \(n\)-buyer lit all-pay auction and the dark all-pay auction, respectively. This is because, in an all-pay auction, buyers pay their bids regardless of whether they win. Under regularity, implementing the optimal allocation rule therefore requires the reserve to be \(\rho^*F(\rho^*)^{n-1}\) in the lit all-pay auction and \(\rho^*G(\rho^*)\) in the dark all-pay auction.} Using Proposition~\ref{prop:ret}, we can now characterize the symmetric equilibrium strategies for the lit and dark auction formats introduced above; these are summarized in Table~\ref{tab:strategies}.

\begin{table}[htbp]
  \centering
  \footnotesize
  \renewcommand{\arraystretch}{1.35}
  \begin{tabular}{c|c}
    \toprule
    \textbf{Lit} & \textbf{Dark} \\
    \midrule
    \(
      S^n_{\text{FP}}(\theta)=
      \begin{cases}
        \varnothing,&\theta \leq \rho,\\
        \theta-\dfrac{\int_\rho^\theta F(x)^{n-1} \dd x}{F(\theta)^{n-1}},&\theta>\rho.
      \end{cases}
    \)
    &
    \(
      S^d_{\text{FP}}(\theta)=
      \begin{cases}
        \varnothing,&\theta \leq \rho,\\
        \theta-\dfrac{\int_\rho^\theta G(x) \dd x}{G(\theta)},&\theta>\rho.
      \end{cases}
    \)
    \\[6mm]
    \(
      S^{n}_{\text{SP}}(\theta)=
      \begin{cases}
        \varnothing,&\theta \leq \rho,\\
        \theta,&\theta>\rho.
      \end{cases}
    \)
    &
    \(
      S^d_{\text{SP}}(\theta)=
      \begin{cases}
        \varnothing,&\theta \leq \rho,\\
        \theta,&\theta>\rho.
      \end{cases}
    \)
    \\[6mm]
    \(
      S^n_{\text{AP}}(\theta)=
      \begin{cases}
        \varnothing,&\theta \leq c_n,\\
        \theta F(\theta)^{n-1}-\int_{c_n}^{\theta}F(x)^{n-1} \dd x,&\theta > c_n.
      \end{cases}
    \)
    &
    \(
      S^d_{\text{AP}}(\theta)=
      \begin{cases}
        \varnothing,&\theta \leq c_d,\\
        \theta G(\theta)-\int_{c_d}^\theta G(x) \dd x,&\theta > c_d.
      \end{cases}
    \)
    \\
    \bottomrule
  \end{tabular}
  \caption{Symmetric equilibrium strategies in lit and dark first-price (FP) or Dutch, second-price (SP) or English, and all-pay (AP) auctions with a random number of buyers.}
  \label{tab:strategies}
\end{table}

Proposition~\ref{prop:ret} shows that, in the absence of \textbf{shill bidding}, the distinction between lit and dark auctions is irrelevant for revenue maximization: first-price, second-price, and all-pay auctions can each achieve optimality in either setting. As we show in the next section, however, this distinction becomes crucial in the presence of shill bidding.

\section{Identity Compatibility}
\label{sec:identity}
We now introduce shill bidding into auctions. The seller \textbf{publicly commits} to the auction \textbf{before} the buyers arrive. After committing, however, the seller may have an incentive to deviate by creating one or more fictitious identities that pose as buyers---\textbf{shill bidders}---and participating in the auction through them in order to manipulate the outcome. Buyers may similarly benefit from creating additional identities. In the following, we distinguish between \textbf{buyers}, who seek to purchase the object on their own behalf, and \textbf{bidders}, the identities through which buyers or the seller participate in the auction. We will use the terms ``identity'' and ``bidder'' interchangeably. Such identity deviations are absent from the standard framework with a fixed number of buyers. We incorporate them by augmenting the auction with a pre-play identity-selection stage. The timing is as follows:
\begin{enumerate}
  \item The seller publicly announces an auction.
  \item Buyers are realized and learn their values.
  \item Buyers and the seller simultaneously decide how many fictitious identities to create.
  \item The auction is run with a fixed number of bidders.
\end{enumerate}

We assume that creating fictitious identities is costless and that each participant knows which identities they control but not who controls the other identities. Without loss of generality, when there are \(n\) buyers and \(k\) fictitious identities, we always label the \(n\) buyers \(1,2,\ldots,n\) and the \(k\) fictitious identities \(n+1,\ldots,n+k\).

\begin{definition}[Full Support]
  \label{def:full-support}
  The distribution over the number of buyers has \textbf{full support} if \(\Pr(\tilde n=n)>0\) for every \(n\in\N\).
\end{definition}

In what follows, we assume that the distribution of the number of buyers has full support, so that every finite population size occurs with positive probability. The Poisson distribution provides a canonical example. This assumption avoids the need to specify games and beliefs at population sizes assigned zero probability by the prior. It also rules out the degenerate case in which the number of buyers is known to equal some \(n\) with probability one. In that case, shill bidding could be deterred mechanically by specifying no sale whenever the number of bidders differs from \(n\). Under full support, however, such a rule would be costly because it would forgo revenue at population sizes that occur with positive probability.

\textbf{Identity compatibility} requires that participants cannot profit by creating fictitious identities. We study the incentives of buyers and the seller separately, focusing on Bayesian buyer identity compatibility for buyers and the stronger notion of ex-post seller identity compatibility for the seller. We first consider shill bidding by buyers, then introduce Bayesian seller identity compatibility before turning to its ex-post counterpart.

\subsection{Bayesian Buyer Identity Compatibility}

Bayesian buyer identity compatibility requires that no buyer can profit in expectation by unilaterally creating fictitious identities. If buyer \(i\) unilaterally deviates by creating \(k\) fictitious identities when there are \(n\) buyers, the resulting auction has \(n+k\) bidders. Buyer \(i\) jointly chooses the strategies played through the \(k+1\) bidders they control: the \(k\) fictitious identities together with buyer \(i\)'s original identity. Their payoff is determined by the aggregate allocation and payment across these identities.

\begin{definition}[Bayesian Buyer Identity Compatibility]{~}\label{def:bayesian-buyer-identity-compatibility}
  \begin{itemize}[leftmargin=*]
    \item A lit auction \((\Gamma^{\N}, S^{\N})\) is \textbf{Bayesian buyer identity-compatible} if, for every buyer \(i\), every \(\theta_i\in\Theta\), every \(k \in \N_{\geq 0}\), and every \((\boldsymbol \sigma_{ik}^{n+k})_{n\in\N} \in \prod_{n\in\N}(\Sigma^{n+k})^{k+1}\),
      \begin{align*}
        & \E_{\tilde{\theta}^{\tilde{n}}_{-i}}\!\left[
          \theta_i q_i^{\Gamma^{\tilde{n}}}\!\left(S^{\tilde{n}}_i(\theta_i), S^{\tilde{n}}_{-i}(\tilde{\theta}^{\tilde{n}}_{-i})\right)
          -
          t_i^{\Gamma^{\tilde{n}}}\!\left(S^{\tilde{n}}_i(\theta_i), S^{\tilde{n}}_{-i}(\tilde{\theta}^{\tilde{n}}_{-i})\right)
        \right]\\
        \geq &
        \E_{\tilde{\theta}^{\tilde{n}}_{-i}}\!\left[
          \sum_{j \in \{i, \tilde{n}+1, \ldots, \tilde{n}+k\}}
          \left[
            \theta_i q_j^{\Gamma^{\tilde{n}+k}}\!\left(\boldsymbol \sigma_{ik}^{\tilde{n}+k}, S^{\tilde{n}+k}_{-i}(\tilde\theta^{\tilde{n}}_{-i})\right)
            -
            t_j^{\Gamma^{\tilde{n}+k}}\!\left(\boldsymbol \sigma_{ik}^{\tilde{n}+k}, S^{\tilde{n}+k}_{-i}(\tilde\theta^{\tilde{n}}_{-i})\right)
          \right]
        \right].
      \end{align*}
    \item A dark auction \((\Gamma^d, S^d)\) is \textbf{Bayesian buyer identity-compatible} if, for every buyer \(i\), every \(\theta_i\in\Theta\), every \(k \in \N_{\geq 0}\), and every \(\boldsymbol \sigma^d_{ik} \in (\Sigma^{d})^{k+1}\),
      \begin{align*}
        & \E_{\tilde{\theta}^{\tilde{n}}_{-i}}\!\left[
          \theta_i q_i^{\Gamma^{d, \tilde{n}}}\!\left(S^d_i(\theta_i), S^d_{-i}(\tilde{\theta}^{\tilde{n}}_{-i})\right)
          -
          t_i^{\Gamma^{d, \tilde{n}}}\!\left(S^d_i(\theta_i), S^d_{-i}(\tilde{\theta}^{\tilde{n}}_{-i})\right)
        \right]\\
        \geq &
        \E_{\tilde{\theta}^{\tilde{n}}_{-i}}\!\left[
          \sum_{j \in \{i, \tilde{n}+1, \ldots, \tilde{n}+k\}}
          \left[
            \theta_i q_j^{\Gamma^{d, \tilde{n}+k}}\!\left(\boldsymbol \sigma_{ik}^{d}, S^d_{-i}(\tilde\theta^{\tilde{n}}_{-i})\right)
            -
            t_j^{\Gamma^{d, \tilde{n}+k}}\!\left(\boldsymbol \sigma_{ik}^{d}, S^d_{-i}(\tilde\theta^{\tilde{n}}_{-i})\right)
          \right]
        \right].
      \end{align*}
  \end{itemize}
\end{definition}

Because a buyer can always choose not to create fictitious identities, Bayesian buyer identity compatibility implies Bayesian incentive compatibility (BIC). For clarity, we sometimes impose both requirements explicitly, even though BIC is redundant whenever Bayesian buyer identity compatibility is imposed.

\begin{observation}\label{obs:bic}
  Every Bayesian buyer identity-compatible mechanism is BIC.
\end{observation}

\begin{proof}
  For a dark mechanism, the claim follows immediately by setting \(k=0\). For a lit mechanism, suppose instead that BIC fails for some type at some realized population size. The buyer can set \(k=0\), follow the prescribed strategy at every other population size, and make the profitable deviation at that population size. This yields a profitable population-contingent deviation, contradicting Bayesian buyer identity compatibility.
\end{proof}

Buyer \(i\) chooses the number of identities before observing the realized population size. Hence, even in a lit mechanism, the number of fictitious identities \(k\) cannot depend on \(\tilde n\). Once the realized population size is revealed, however, the strategies played through those identities may depend on it, as captured by the population-contingent strategy profile\footnote{When the mechanism is dynamic, buyer \(i\) can pool information across the bidders they control. Consequently, \((\Sigma^{n+k})^{k+1}\) may not capture the full strategy space available to buyer \(i\). This distinction does not affect our analysis, and we relegate the associated subtleties to Online Appendix~\ref{app:extensive-form-games}.}
\[
  \boldsymbol \sigma_{ik}^{n+k}=(\sigma_i^{n+k},\sigma_{n+1}^{n+k},\ldots,\sigma_{n+k}^{n+k})\in(\Sigma^{n+k})^{k+1}.
\]
In a dark mechanism, neither the number of identities nor the strategies played through them can depend on the unobserved population size.

For buyer \(i\), shill bidding affects the auction through two channels. First, fictitious identities increase the number of bidders, which may change both the game being played and the strategies of the nondeviating buyers. Second, controlling multiple identities directly expands buyer \(i\)'s strategy space: instead of choosing a single strategy, buyer \(i\) can coordinate a profile of strategies across the identities they control.\footnote{These identities therefore resemble a bidding ring or cartel under buyer \(i\)'s full control. Shill bidders, however, differ from members of a conventional cartel. A buyer fully controls their fictitious identities, whereas distinct cartel members may have incentives to deviate from their collusive agreement. Moreover, a cartel consists of distinct buyers, while shill bidders are multiple identities controlled by the same buyer. For the literature on collusion in auctions, see \citet{Robinson-CollusionChoiceAuction-1985}, \citet{Graham-Marshall-CollusiveBidderBehavior-1987}, \citet{Mailath-Zemsky-CollusionSecondPrice-1991}, \citet{McAfee-McMillan-BiddingRings-1992}, \citet*{Garratt-Troger-Zheng-CollusionResale-2009}, and \citet{Chassang-Ortner-CollusionAuctionsConstrained-2019}.\label{fn:collusion}} BIC rules out deviations in strategies only when buyer \(i\) participates through a single identity. With shills, buyer \(i\) may instead make a double deviation, changing both the number of identities and the strategies played through them.

The following two examples illustrate these two channels. The first shows how shill bidding can affect the bidding strategies of other buyers, while the second shows how shills expand the deviating buyer's strategy space.\footnote{Some procurement auctions result in no sale when the number of participants falls below a prespecified threshold, and the procurer may instead resort to bargaining when competition is insufficient. In such cases, suppliers may use shills simply to ensure that the auction takes place under limited competition. See \href{https://www.ft.com/content/37dd0d36-109a-11e7-a88c-50ba212dce4d}{``BT's Global Services faces Asian bid rigging probe,'' \emph{Financial Times}, March 24, 2017}.}

\begin{example}
  \label{ex:auctions}
  Unless otherwise specified, we assume, for simplicity, that the feasible bid space is \([0,+\infty)\), so that there is no reserve price.
  \begin{itemize}
    \item Lit and dark first-price (and Dutch) auctions are Bayesian buyer identity-compatible.
    \item Second-price (and English) auctions are Bayesian buyer identity-compatible.
    \item Lit all-pay auctions are not Bayesian buyer identity-compatible.
    \item Dark all-pay auctions are Bayesian buyer identity-compatible.
  \end{itemize}

  It is straightforward to verify that second-price auctions, dark first-price auctions, and dark all-pay auctions are Bayesian buyer identity-compatible, because shill bidding neither affects the strategies of the other buyers nor makes it easier for buyer \(i\) to win. The more interesting comparison is between lit first-price and lit all-pay auctions: the former are Bayesian buyer identity-compatible, whereas the latter are not.

  In both formats, shill bidding increases the number of bidders and thereby changes the \textbf{perceived competition} faced by the other buyers, which in turn affects their bidding strategies. The effect, however, differs across the two formats. In a lit first-price auction, buyers bid more aggressively when there are more bidders, making it more costly for buyer \(i\) to outbid them. Shill bidding is therefore unprofitable for buyer \(i\). In a lit all-pay auction, by contrast, low-value buyers bid less aggressively when there are more bidders,\footnote{For sufficiently high values, the opposite is true. Under the uniform distribution, \(S^{n+1}_{\text{AP}}(\theta)<S^n_{\text{AP}}(\theta)\) if and only if \(\theta<1-\frac{1}{n^2}\) (Table~\ref{tab:strategies}).} making them easier to outbid. In particular, buyer \(i\) can create an additional identity, submit through the original identity the same bid they would have submitted without the additional identity, and submit a zero bid (or decline to bid) through the additional identity. This deviation leaves buyer \(i\)'s payment unchanged while strictly increasing their probability of winning. Intuitively, by artificially inflating perceived competition, buyer \(i\) ``scares away'' low-value opponents, discouraging them from bidding aggressively and making them easier to outbid. Shill bidding can therefore be profitable for a low-value buyer. A recent newspaper article describes exactly the same strategic logic in a contest setting: ``In the cutthroat world of Chinese civil service exams, the most dangerous opponent is not necessarily the person in the next seat: it is the `ghost' candidate designed to scare you away before you even show up.''\footnote{See \href{https://www.scmp.com/news/china/politics/article/3344010/chinas-tight-economy-fake-exam-takers-scheme-has-touched-nerve}{``In China's tight economy, a fake exam-takers scheme has touched a nerve,'' \emph{South China Morning Post}, February 20, 2026}. For related work on disclosure policies in contests with population uncertainty, see \citet*{Lim-Matros-ContestsStochasticNumber-2009}, \citet*{Fu-Jiao-Lu-DisclosurePolicyContests-2011,Fu-Jiao-Lu-DisclosurePolicyMultiprize-2014}, and \citet*{Chen-Jiang-Knyazev-DisclosurePoliciesAllpay-2017}.}
\end{example}

\begin{example}\label{ex:two-stage-auction}
  Consider a two-stage procurement auction adapted from \citet{Seibel-Skoda-CollusionExclusionPublic-2026}. There are at least two suppliers. In the preselection stage, suppliers submit binding sealed bids without observing the number of participants. The \(m\geq2\) suppliers with the lowest bids advance to the main stage, or all suppliers advance if fewer than \(m\) participate. The main stage is a descending English auction starting at the lowest opening bid.\footnote{For the competitive benchmark of \citet{Seibel-Skoda-CollusionExclusionPublic-2026}, let costs be i.i.d. on \([\underline c,\bar c]\), with a positive continuous density, and let the procurement reserve satisfy \(r>\bar c\). Any strictly increasing opening-bid function \(\beta:[\underline c,\bar c]\to[\bar c,r]\), followed by exit at cost, constitutes an equilibrium. With at least two available slots, the two lowest-cost suppliers qualify for the English stage, and the lowest-cost supplier wins at the second-lowest cost. Here, the English stage follows the sealed-bid stage, reversing the order of an Anglo-Dutch auction \citep{Klemperer-AuctionsAlmostCommon-1998}.}

  Shill bidding changes neither other suppliers' opening-bid strategies, because the preselection stage is dark, nor their cost-based exit rules in the main stage. Its benefit instead comes from expanding the supplier's strategy space for controlling who reaches the main stage. A supplier can submit its opening bid through \(m-1\) additional shills. These bids leave the starting price unchanged but occupy all finalist slots whenever the supplier has the lowest cost. The shills then withdraw immediately, allowing the supplier to win at its opening bid and receive a payment above the competitive price.\footnote{These coordinated withdrawals are off-path deviations unavailable in the induced direct auction. The latter is a second-price procurement auction and is Bayesian buyer identity-compatible.} Thus, multistage auctions may be vulnerable to shill bidding because fictitious identities can be used to crowd out competitors before competition begins in the main stage.\footnote{\citet{Seibel-Skoda-CollusionExclusionPublic-2026} derive a similar theoretical prediction in a collusion setting and provide empirical support using administrative data from Slovak public procurement auctions, corroborated by a court-confirmed cartel case.}
\end{example}

\subsection{Bayesian Seller Identity Compatibility}
We next turn to the seller. Suppose first that the auction is operated by a trusted third party with reputational concerns, such as an auction house or an online platform. Even then, the seller may have an incentive to create fictitious identities and participate in the auction through them. The seller has no information about the buyers' values or the number of buyers beyond the prior. The seller's payoff is the sum of the payments received from buyers. Payments made by seller-controlled fictitious identities, including entry fees, are internal transfers and therefore cancel out.

\begin{definition}[Bayesian Seller Identity Compatibility]{~}\label{def:bayesian-seller-identity-compatibility}
  \begin{itemize}[leftmargin=*]
    \item A lit auction \((\Gamma^{\N}, S^{\N})\) is \textbf{Bayesian seller identity-compatible} if, for every \(k \in \N\) and every \((\boldsymbol \sigma_{0k}^{n+k})_{n\in\N} \in \prod_{n\in\N}(\Sigma^{n+k})^{k}\),
      \[
        \E_{\tilde{\theta}^{\tilde{n}}}\!\left[
          \sum_{i = 1}^{\tilde{n}}
          t_i^{\Gamma^{\tilde{n}}}\!\left(S^{\tilde{n}}(\tilde{\theta}^{\tilde{n}})\right)
        \right]
        \geq
        \E_{\tilde{\theta}^{\tilde{n}}}\!\left[
          \sum_{i = 1}^{\tilde{n}}
          t_i^{\Gamma^{\tilde{n}+k}}\!\left(S^{\tilde{n}+k}(\tilde{\theta}^{\tilde{n}}), \boldsymbol \sigma_{0k}^{\tilde{n}+k}\right)
        \right].
      \]
    \item A dark auction \((\Gamma^d, S^d)\) is \textbf{Bayesian seller identity-compatible} if, for every \(k \in \N\) and every \(\boldsymbol \sigma^d_{0k} \in (\Sigma^{d})^{k}\),
      \[
        \E_{\tilde{\theta}^{\tilde{n}}}\!\left[
          \sum_{i = 1}^{\tilde{n}}
          t_i^{\Gamma^{d,\tilde{n}}}\!\left(S^{d}(\tilde{\theta}^{\tilde{n}})\right)
        \right]
        \geq
        \E_{\tilde{\theta}^{\tilde{n}}}\!\left[
          \sum_{i = 1}^{\tilde{n}}
          t_i^{\Gamma^{d,\tilde{n}+k}}\!\left(S^{d}(\tilde{\theta}^{\tilde{n}}), \boldsymbol \sigma_{0k}^{d}\right)
        \right].
      \]
  \end{itemize}
\end{definition}

Bayesian seller identity compatibility ensures that the seller cannot raise revenue in expectation by introducing shills. Unlike buyers, who may use shill bidding to reduce their payments or increase their probability of winning, the seller uses shills to increase the payments collected from buyers. When losing bidders make no payments, the seller avoids winning the auction, since doing so yields nonpositive revenue. For the seller, shill bidding operates through two channels. First, shills can directly affect buyers' payments by manipulating the auction outcome, as in a second-price auction. Second, shills can indirectly affect buyers' bidding strategies by altering their perceived competition, as in first-price and all-pay auctions. We discuss these two channels in turn.

In a (lit or dark) second-price auction, the seller's highest shill bid effectively acts as a reserve price, since buyers with lower bids cannot win. When the value distribution is regular (Definition~\ref{def:regular}), the second-price auction with the Myersonian reserve is optimal, and the seller has no incentive to shill bid. The same is true of the English auction \citep{Izmalkov-ShillBiddingOptimal-2004,Akbarpour-Li-CredibleAuctionsTrilemma-2020}. When the value distribution is irregular, however, the seller may profitably participate in a lit English auction,\footnote{To be specific, this is an open ascending auction with irrevocable exits. A clock tracks the ascending price, and bidders successively withdraw until a single bidder remains, who wins the object and pays the last exit price. The exit prices of the other bidders are publicly observed.} as noted by \citet*{Wang-HidvegiZoltan-Whinston-ShillProofFeeSPF-2004}, because the seller can condition their shill-bidding strategy on the observed second-highest bid among the buyers. Nevertheless, there exists a reserve price that eliminates the seller's incentive to shill bid.

Define
\(
  \bar\rho
  =
  \sup\left\{\{\theta\in[0,1]:\phi(\theta)<0\}\cup\{0\}\right\}.
\)

\begin{proposition}\label{prop:second-price}
  Both second-price auctions and English auctions with reserve price \(\bar\rho\) are Bayesian seller identity-compatible.
\end{proposition}

The result does not require regularity, because every buyer who bids in these auctions has a nonnegative virtual value. Appendix~\ref{app:proof-second-price} provides a formal proof and discusses the distinction between second-price and English auctions. In particular, in a lit second-price auction, a population-contingent reserve schedule can also ensure Bayesian seller identity compatibility, whereas such a schedule may not work in a lit English auction.

Bayesian seller identity compatibility and \textbf{credibility} \citep{Akbarpour-Li-CredibleAuctionsTrilemma-2020} address different forms of seller manipulation. Under credibility, the seller operates the auction and may misrepresent buyers' preferences through bilateral private communication, but cannot alter the number of bidders by introducing shills. By contrast, under Bayesian seller identity compatibility, the auction is operated by a trusted third party, so the seller cannot observe buyers' private messages or misrepresent their preferences. Instead, the seller may introduce shills and thereby alter the number of bidders. Thus, neither notion subsumes the other; they constrain distinct and complementary forms of seller manipulation.

\bigskip

\noindent \textbf{Example~\ref{ex:auctions} (continued).}
\begin{itemize}
  \item Lit first-price auctions are credible but not Bayesian seller identity-compatible.
  \item Dark first-price auctions are both Bayesian seller identity-compatible and credible.
  \item Second-price auctions with reserve price \(\bar\rho\) are Bayesian seller identity-compatible but not credible.
  \item English auctions with reserve price \(\bar\rho\) are Bayesian seller identity-compatible and, under regularity, also credible.
  \item Under regularity, the optimal lit all-pay auction is both Bayesian seller identity-compatible and credible.
  \item Dark all-pay auctions are both Bayesian seller identity-compatible and credible.
\end{itemize}

The lit first-price auction is not Bayesian seller identity-compatible because the seller can artificially inflate the number of bidders and thereby manipulate buyers' perceptions of competition. Intuitively, the seller can make the auction room appear crowded by introducing shill bidders, thereby inducing buyers to bid more aggressively. Both the dark first-price and dark all-pay auctions shut down this channel of manipulation by concealing the number of bidders from buyers.\footnote{Under regularity, an optimal lit auction satisfying Bayesian seller identity compatibility, Bayesian buyer identity compatibility, and ex-post IR need not be a second-price auction; see Example~\ref{ex:non-second-price} in the Online Appendix.}

Interestingly, under regularity, the optimal lit all-pay auction is Bayesian seller identity-compatible. As discussed above, low-value buyers bid less aggressively when there are more bidders in a lit all-pay auction, whereas high-value buyers bid more aggressively. The key additional observation is that, ex ante, the expected payment of a buyer whose value is unknown decreases with the number of bidders.\footnote{Under regularity, the expected payment of a buyer of unknown value in an \(n\)-buyer optimal auction is \(\int_{\rho^*}^{1}\phi(\theta)F(\theta)^{n-1} f(\theta)\dd \theta\), where \(\rho^*\) is the Myersonian reserve. Because \(\phi(\theta)\geq 0\) for \(\theta>\rho^*\), this expression is decreasing in \(n\).} Consequently, by creating fictitious identities, the seller reduces rather than increases the expected payment collected from each buyer and therefore has no incentive to shill bid.

\subsection{Ex-Post Seller Identity Compatibility}
\label{sec:ex-post-seller}

In many applications, however, the seller may acquire information about buyers' values beyond the prior. A seller operating the auction can observe and exploit information submitted by buyers; see \citet{Lucking-Reiley-VickreyAuctionsPractice-2000} for illustrative anecdotes.\footnote{Sellers may also infer buyers' values from other observed behavior. For instance, an experienced real estate agent may infer buyers' valuations from their behavior during a property viewing. Information leakage is also a growing concern in digital markets \citep*{Hafner-Pycia-Zeng-MechanismDesignInformation-2025}.} Past interactions may also reveal information about buyers' values. For example, Google Ad Manager describes using historical bidding data and post-auction analysis to adjust auction floor prices; see \href{https://admanager.google.com/home/resources/feature-brief-yield-management/}{Google Ad Manager, ``Comprehensive Yield Management.''} Similarly, on eBay, if a seller's shill wins the auction, the seller can learn about buyers' values from the bidding history and relist the object. More strategically, the seller can submit high bids through shills and subsequently retract them, exploiting automatic bidding systems to uncover buyers' maximum bids; see \href{https://www.ebay.com/help/buying/bidding/retracting-bid?id=4013}{eBay, ``Retracting a bid.''} To capture the seller's ability to exploit such information, we introduce a stronger notion of seller identity compatibility that allows shill-bidding strategies to depend on the buyers' value profile. The timing of identity creation, however, differs between lit and dark auctions.

\begin{definition}[Ex-Post Seller Identity Compatibility]\label{def:ex-post-seller-identity-compatibility}
  Suppose the expectations below are well defined.\footnote{In this definition and its strengthened version below (Definition~\ref{def:strongly-ex-post}), each expectation of the supremum over deviation payoffs is assumed well defined. For an alternative approach that constructs measurable profitable deviations through countable approximation, see \citet[Lemma~1 and the proof of Theorem~1]{Akbarpour-Li-CredibleAuctionsTrilemma-2020}.}
  \begin{itemize}
    \item A lit auction \((\Gamma^{\N}, S^{\N})\) is \textbf{ex-post seller identity-compatible} if, for every \(n \in \N\) and every \(k \in \N\),
      \[
        \E_{\tilde{\theta}^{n}}\left[\sum_{i=1}^{n} t_i^{\Gamma^{n}}\!\left(S^{n}(\tilde{\theta}^{n})\right)\right] \geq \E_{\tilde{\theta}^{n}}\left[\sup_{\boldsymbol{\sigma}_{0k}^{n+k} \in (\Sigma^{n+k})^k} \sum_{i=1}^{n} t_i^{\Gamma^{n+k}}\!\left(S^{n+k}(\tilde{\theta}^{n}), \boldsymbol{\sigma}_{0k}^{n+k}\right)\right].
      \]
    \item A dark auction \((\Gamma^d, S^d)\) is \textbf{ex-post seller identity-compatible} if
      \[
        \E_{\tilde{\theta}^{\tilde{n}}}\left[\sum_{i=1}^{\tilde{n}} t_i^{\Gamma^{d, \tilde{n}}}\!\left(S^d(\tilde{\theta}^{\tilde{n}})\right)\right] \geq \E_{\tilde{\theta}^{\tilde{n}}}\left[\sup_{k \in \N_{\geq 0},\,\boldsymbol{\sigma}_{0k}^{d} \in (\Sigma^d)^k} \sum_{i=1}^{\tilde{n}} t_i^{\Gamma^{d, \tilde{n}+k}}\!\left(S^d(\tilde{\theta}^{\tilde{n}}), \boldsymbol{\sigma}_{0k}^{d}\right)\right].
      \]
  \end{itemize}
\end{definition}

The distinction between lit and dark auctions concerns when the seller can choose the number of shills. In a lit auction, this choice may depend on the realized number of buyers but must be made before their values are observed. Once the number of shills is fixed, however, their strategies may depend on the realized value profile. This timing captures, for example, a static auction in which buyers submit their messages only after the number of bidders has been fixed. Conditional on each realized number of buyers, ex-post seller identity compatibility requires that no such deviation increase expected revenue, where the expectation is taken before buyers' values are observed. A deviation may therefore raise revenue at some value profiles if these gains are offset by losses at others.

In a dark auction, by contrast, the seller may secretly introduce shills after buyers have submitted their messages. Both the number of shills and the strategies played through them may therefore depend on the realized value profile. Because the seller can choose to introduce no shills at each profile, ex-post seller identity compatibility rules out profitable deviations at almost every realized value profile. In both settings, the auction format remains fixed before buyers arrive: additional information may guide the seller's shill-bidding strategy but cannot be used to redesign the mechanism. In the advertising example, for instance, the auction format is fixed in advance, while information relevant to buyers' values may become available only later.

Full observation of buyers' values should be interpreted as a robustness benchmark rather than as a literal assumption about the seller's information. Any shill-bidding strategy based on an imperfect signal about values is also available to a fully informed seller. Ex-post seller identity compatibility therefore protects against profitable use of such signals as well. Conversely, a violation identifies an incentive to shill bid when sufficiently informative signals are available; depending on the auction format, imperfect information may suffice. This robustness makes ex-post seller identity compatibility substantially more demanding than its Bayesian counterpart.

\bigskip

\noindent \textbf{Example~\ref{ex:auctions} (continued).}
\begin{itemize}
  \item Dark first-price auctions are ex-post seller identity-compatible.
  \item Under regularity, the optimal lit all-pay auction is ex-post seller identity-compatible.
  \item Dark all-pay auctions are ex-post seller identity-compatible.
\end{itemize}

The lit first-price, second-price, and English auctions considered above fail ex-post seller identity compatibility. The optimal lit all-pay auction illustrates why the timing of identity creation matters. Under regularity, adding any fixed number of shills reduces the expected payment collected from each buyer, as discussed above. Conditioning the shills' bids on realized values cannot reverse this effect because each buyer pays their own bid regardless of the shills' bids. Nevertheless, at some high-value profiles, the larger number of bidders induces more aggressive bidding and raises revenue. If the seller could choose the number of shills after observing buyers' values, they could introduce shills precisely at these profiles and refrain from doing so elsewhere. The optimal lit all-pay auction therefore satisfies our definition but would fail the stronger requirement that allows the number of shills itself to depend on realized values.

All-pay auctions are special in that losing bidders also make payments. In many practical settings, however, auctions may need to satisfy \textbf{ex-post individual rationality (ex-post IR)}, since bidders may otherwise attempt to renege on their payments.\footnote{See \citet{Compte-Jehiel-QuittingRightsMechanism-2007} and \citet{Haberman-Jagadeesan-AuctionsWithdrawalRights-2026} for discussions of quitting and withdrawal rights in auctions. Unlike the winner-paying requirement, ex-post IR permits subsidies to losing buyers, which could in principle be used to deter shill bidding by the seller. Theorem~\ref{thm:lit-regular-ir} shows that this additional flexibility is nevertheless insufficient.}

\begin{definition}[Ex-Post Individual Rationality]{~}
  \label{def:ex-post-IR}
  \begin{itemize}
    \item A lit auction \((\Gamma^{\N}, S^{\N})\) is \textbf{ex-post individually rational (ex-post IR)} if, for every \(n \in \N\), every buyer \(i\), and every \(\theta^n \in \Theta^n\),
      \[
        \theta_i q_i^{\Gamma^{n}, S^{n}}(\theta^n) -
        t_i^{\Gamma^{n}, S^{n}}(\theta^n) \geq 0.
      \]
    \item A dark auction \((\Gamma^d, S^d)\) is \textbf{ex-post individually rational (ex-post IR)} if, for every \(n \in \N\), every buyer \(i\), and every \(\theta^n \in \Theta^n\),
      \[
        \theta_i q_i^{\Gamma^{d, n}, S^d}(\theta^n) -
        t_i^{\Gamma^{d, n}, S^d}(\theta^n) \geq 0.
      \]
  \end{itemize}
\end{definition}

Moreover, as discussed above, the optimal lit all-pay auction is not Bayesian buyer identity-compatible. This raises the question of whether other optimal lit auctions can reconcile ex-post seller identity compatibility with either of these requirements.

\section{Main Results}
\label{sec:results}

Our first two results answer this question negatively. We begin with ex-post individual rationality.

\begin{theorem}
  \label{thm:lit-regular-ir}
  Assume common posterior, full support, and regularity. No optimal lit auction can simultaneously be ex-post seller identity-compatible and ex-post individually rational.
\end{theorem}

\paragraph{Proof Sketch.}
The lit first-price auction provides the basic intuition. Buyers bid more aggressively when they face more bidders: for any given value above the Myersonian reserve, the equilibrium bid with \(n+1\) bidders exceeds that with \(n\) bidders. The seller can therefore add a shill who submits a sufficiently low bid and never wins. With the shill present, buyers follow the \((n+1)\)-buyer equilibrium and bid more aggressively. The shill leaves the allocation among buyers unchanged, but the winner now pays more. Thus, in a lit first-price auction, shill bidding raises revenue by increasing the perceived competition among buyers.

For a general optimal lit auction, however, this intuition is incomplete. Optimality pins down interim expected payments, not payments at each realized value profile, and ex-post IR permits subsidies to losing buyers. Hence, a higher payment by the winner need not translate into higher total revenue. The proof therefore compares the total utilities left to buyers instead. Fix \(n\) buyers. The seller adds one shill and, after observing the buyers' values, draws the value it pretends to have from \(F\), conditional on being below the highest buyer value. The shill therefore never wins, so the allocation to buyers is unchanged, while the buyers follow their equilibrium strategies from the \((n+1)\)-buyer auction. Ex-post IR makes their realized utilities nonnegative, and the envelope formula allows us to show that their total expected utilities under the deviation are strictly lower than in the \(n\)-buyer auction. Because the surplus generated by the same allocation is unchanged, this reduction in buyers' utility translates directly into higher revenue for the seller. See Appendix~\ref{app:proof-lit-regular-ir} for details. \qed

\begin{theorem}
  \label{thm:lit-regular-bb}
  Assume common posterior, full support, and regularity. No optimal lit auction can simultaneously be ex-post seller identity-compatible and Bayesian buyer identity-compatible.
\end{theorem}

\paragraph{Proof Sketch.}

The intuition is clearest from the optimal lit all-pay auction. Adding a low-value shill manipulates perceived competition without affecting the allocation. Although such a shill cannot increase the seller's expected revenue, it induces the other buyers to behave as if competition were greater, reducing the bids of low-value buyers and making them easier to outbid. A buyer may therefore benefit from introducing the very same shill. The key observation is that the same fictitious identity imposes restrictions in opposite directions depending on whether it is controlled by the seller or by a buyer.

The proof formalizes this observation using a shill that mimics a value below the Myersonian reserve. Such a shill never wins and leaves the other buyers' allocations unchanged. Ex-post seller identity compatibility implies that adding it cannot increase expected revenue, whereas Bayesian buyer identity compatibility implies that a buyer cannot lower their total payment by adding it. Averaging the buyer's constraint over their value shows that both inequalities must bind at every realized number of buyers. Exploiting this equality and varying the shill's strategy yields bounds on the payment rules that are incompatible with optimality. See Appendix~\ref{app:proof-lit-regular-bb} for details. \qed

\bigskip

The two impossibility results reveal a fundamental design \textbf{dilemma} in lit auctions. Competition is essential for revenue maximization, but the seller can manipulate perceived competition through shill bidding. One might therefore expect the impossibility to arise entirely from the seller's incentive to manufacture competition and thereby raise buyers' payments. Theorem~\ref{thm:lit-regular-ir} confirms this intuition under ex-post IR: even a shill that never wins can increase revenue by affecting buyers' strategies.

Theorem~\ref{thm:lit-regular-bb}, however, reveals another side of the dilemma. Without ex-post IR, artificially heightened competition need not benefit the seller. In the optimal lit all-pay auction, for example, shill bidding by the seller reduces expected revenue and is therefore unprofitable. Yet buyers can exploit the same artificial competition to discourage their rivals from bidding aggressively. The second theorem establishes that this tension between the seller's and buyers' incentives extends beyond the all-pay auction. Thus, the difficulty is not merely that the seller can manufacture competition, but that either side can manipulate perceived competition to its advantage.

By contrast, the next theorem shows that this dilemma can be resolved in dark auctions. Concealing the number of bidders prevents shills from influencing buyers' strategies through changes in perceived competition. An optimal dark auction can therefore satisfy ex-post seller identity compatibility, Bayesian buyer identity compatibility, and ex-post IR simultaneously. Moreover, these requirements essentially characterize the dark first-price auction: for each population size, its allocation probabilities and expected payments conditional on the buyers' value profile coincide almost everywhere with those of a dark first-price auction.\footnote{The distributions of realized payments need not coincide. Holding the allocation fixed, risk-neutral buyers are indifferent between a deterministic payment and a lottery with the same conditional mean.}

\begin{theorem}
  \label{thm:dark-regular}
  Assume common posterior, full support, and regularity. Any optimal dark auction that is ex-post seller identity-compatible, Bayesian buyer identity-compatible, and ex-post individually rational is essentially a dark first-price auction.
\end{theorem}

\paragraph{Proof Sketch.}
Under regularity, optimality pins down the allocation and, by the Revenue Equivalence Theorem with a random number of buyers (Proposition~\ref{prop:ret}), each buyer's interim expected payment. In particular, a buyer of value below the Myersonian reserve obtains zero expected utility. Ex-post IR then leaves no room to charge or subsidize such a buyer when the seller retains the object: every payment must be zero on no-sale value profiles.

The key issue is therefore what happens when the object is sold. Fix a value profile at which the object is sold and consider buyer \(i\), whose value is the highest at that profile. By creating several shills, buyer \(i\) can reproduce the entire profile and win whenever every opponent has a lower value. From the seller's perspective, treat buyer \(i\)'s identities as the original buyers and the opponents as shills. Ex-post seller identity compatibility then bounds the total payment collected from buyer \(i\)'s identities by the revenue at the reproduced profile. From the buyer's perspective, Bayesian buyer identity compatibility requires this revenue to be at least the corresponding dark first-price bid. Otherwise, buyer \(i\) could obtain the same winning event at a lower price whenever they win, while ex-post IR ensures no positive payment when they lose.

This almost pins down the dark first-price payment rule, but there is one loophole. The winner could pay more than the dark first-price bid while some losing buyers receive subsidies, so that interim expected payments remain unchanged. The final argument rules out precisely this cross-subsidization. Weight each buyer's payment difference relative to the dark first-price bid by that buyer's value. Since the winner has the highest value, an extra dollar collected from the winner receives more weight than a dollar paid to a loser. The revenue bound therefore makes the weighted difference nonnegative, with strict inequality whenever revenue exceeds the dark first-price bid or any loser receives a subsidy. Yet revenue equivalence requires its expectation to be zero. Thus, almost surely, losers' payments are zero and the winner's payment equals the dark first-price bid. See Appendix~\ref{app:proof-dark-regular} for details. \qed

\bigskip

Theorem~\ref{thm:dark-regular} is tight: each of its three conditions is essential. Without ex-post IR, the optimal dark all-pay auction satisfies the remaining conditions; Example~\ref{ex:dark-family} in the Online Appendix provides a family of auctions that nests both the dark first-price and dark all-pay auctions and satisfies all the remaining conditions. Without ex-post seller identity compatibility, the optimal second-price auction satisfies the remaining conditions. Finally, without Bayesian buyer identity compatibility, one can construct an optimal dark auction in which the winner's payment decreases with the realized number of buyers while preserving the same interim expected payment as in the dark first-price auction. Such an auction remains ex-post seller identity-compatible because adding shills can only reduce the winner's payment, but violates Bayesian buyer identity compatibility because a buyer can profit by adding shills to lower the payment they face.

Theorem~\ref{thm:dark-regular} also clarifies why both the first-price rule and the dark format are essential. Notice that Theorems~\ref{thm:lit-regular-ir} and~\ref{thm:lit-regular-bb} do not follow from Theorem~\ref{thm:dark-regular} together with Observation~\ref{obs:lit-dark}. Indeed, ex-post seller identity compatibility is stronger in dark auctions than in lit auctions because the seller may condition the number of shills on the realized value profile. Nevertheless, dark auctions can satisfy this stronger requirement together with the other desirable properties, whereas lit auctions cannot.

To understand this contrast, it is useful to distinguish two margins of competition. Holding the number of bidders fixed, the first-price rule prevents the seller from raising the winner's payment through competing shill bids, because the payment is determined by the winner's own bid rather than by a rival's bid. In a lit first-price auction, however, the seller can still manipulate the extensive margin of competition: equilibrium bids depend on the number of bidders, so introducing shills induces buyers to bid more aggressively.

The dark format eliminates this remaining channel. Because buyers do not observe the number of bidders, introducing shills cannot change their bids by increasing perceived competition. Thus, the first-price rule neutralizes manipulation along the intensive margin of competition, while the dark format neutralizes manipulation along the extensive margin. Theorem~\ref{thm:dark-regular} shows that, under regularity and optimality, these two features together are not merely sufficient: they essentially characterize the unique auction satisfying ex-post seller identity compatibility, Bayesian buyer identity compatibility, and ex-post IR.\footnote{In a related literature without shill bidding, \citet{McAfee-McMillan-AuctionsStochasticNumber-1987} and \citet{Matthews-ComparingAuctionsRisk-1987} show that dark first-price auctions can generate strictly higher expected revenue than both lit first-price and second-price auctions when buyers exhibit decreasing absolute risk aversion. For experimental evidence, see \citet*{Dyer-Kagel-Levin-ResolvingUncertaintyNumber-1989} and \citet{Aycinena-Rentschler-AuctionsEndogenousParticipation-2018}. Thus, the dark first-price auction can be attractive not only from the buyers' perspective but also from the seller's.}

A real-world example comes from online advertising. In 2019, Google switched from second-price to first-price auctions for Google Ad Manager.\footnote{\href{https://blog.google/products/admanager/update-first-price-auctions-google-ad-manager/}{``An Update on First Price Auctions for Google Ad Manager,'' May 10, 2019}.} Notably, Google's first-price auction is effectively dark in our terminology, since the number of advertisers is typically not disclosed. Google has faced litigation involving alleged shill-bidding-like manipulation of advertising auctions (Footnote~\ref{fn:google}), but some allegations concern features specific to second-price auctions and therefore do not apply to the dark first-price format. The transition to first-price auctions thus removes one source of concern among market participants about shill bidding. A related example is eBay, where posted prices have become the predominant selling mechanism \citep*{Einav-Farronato-Levin-Sundaresan-AuctionsPostedPrices-2018}. From a dynamic perspective, a sequence of decreasing posted prices can be viewed as a dark first-price (Dutch) auction in which the seller gradually lowers the price until a buyer accepts.

\paragraph{A note on ex-post buyer identity compatibility.}
One can analogously define a stronger, ex-post notion of buyer identity compatibility. Bayesian buyer identity compatibility implies BIC (Observation~\ref{obs:bic}), while ex-post buyer identity compatibility implies ex-post IC. The latter observation yields the following characterization; details are deferred to Online Appendix~\ref{app:proof-ex-post-buyer}.

\begin{observation}\label{obs:ex-post-buyer}
  Assume common posterior, full support, and regularity. Any optimal lit or dark auction that is ex-post buyer identity-compatible and ex-post individually rational is essentially a second-price auction.
\end{observation}

\section{Discussion}\label{sec:discussion}

\subsection{Lit Auctions under Shill Bidding}
\label{sec:posted-price}

The impossibility results above raise the question of how much revenue a lit auction can achieve while eliminating the seller's incentive to use shills. We focus on the seller's incentives, since additionally imposing buyer identity compatibility would only further restrict the feasible set and cannot increase the maximum attainable revenue. To facilitate comparison with the dark first-price auction, we strengthen ex-post seller identity compatibility in two respects. First, as in dark auctions, we allow the seller to condition the number of shills on the realized value profile. Second, we allow the seller to manipulate the auction's randomization device.\footnote{Credibility \citep{Akbarpour-Li-CredibleAuctionsTrilemma-2020} also accounts for manipulation of the randomization device. A similar concern motivates ex-post deterministic implementation in \citet{Dworczak-MechanismDesignAftermarkets-2020}. See \citet{Engelberg-Williams-EBaysProxyBidding-2009} for an analysis of eBay's tie-breaking rule, which effectively benefits the seller under shill bidding.} For example, if the allocation rule randomizes between selling and retaining the object, the seller can select an outcome in which the object is sold and collect the resulting payment. Likewise, when breaking ties, the seller can select the buyer whose payment conditional on winning is highest. Such manipulation cannot be detected from the realized outcome alone, since the selected outcome lies in the support of the prescribed lottery, even if buyers pool their information after the auction. To distinguish randomization generated by the auction from that induced by buyers' mixed strategies, we restrict attention to pure strategies in the following definition.

\begin{definition}
  \label{def:strongly-ex-post}
  A lit auction \((\Gamma^\N,S^\N)\) is \textbf{strongly ex-post seller identity-compatible} if
  \[
    \E_{\tilde\theta^{\tilde n}}\!\left[
      \sum_{i=1}^{\tilde n}
      t_i^{\Gamma^{\tilde n}}
      \left(S^{\tilde n}(\tilde\theta^{\tilde n})\right)
    \right]
    \geq
    \E_{\tilde\theta^{\tilde n}}\!\left[
      \sup_{\substack{
          k\in\N_{\geq0}\\
          \boldsymbol{\sigma}_{0k}^{\tilde n+k}\in(\Sigma^{\tilde n+k})^k
      }}
      \ \sup_{\substack{
          (q^{\tilde n+k},t^{\tilde n+k})\in
          \mathcal X^{\tilde n+k}
          \left(S^{\tilde n+k}(\tilde\theta^{\tilde n}),
          \boldsymbol{\sigma}_{0k}^{\tilde n+k}\right)
      }}
      \sum_{i=1}^{\tilde n}t_i^{\tilde n+k}
    \right],
  \]
  where \(\mathcal X^n(\boldsymbol{\sigma}):=\supp\varphi^n(\boldsymbol{\sigma})\) is the support of the outcome lottery induced by the strategy profile \(\boldsymbol{\sigma}\) in the \(n\)-buyer auction \((\Gamma^n,S^n)\), and \(S^n\) is a pure strategy for every \(n\in\N\).

  A lit auction \((\Gamma^\N,S^\N)\) is \textbf{strongly ex-post individually rational} if, for every \(n\in\N\), every buyer \(i\), every \(\theta^n\in\Theta^n\), and every \((q^n,t^n)\in\mathcal X^n\left(S^n(\theta^n)\right)\),
  \[
    \theta_iq_i^n-t_i^n\geq0.
  \]
\end{definition}

Strong ex-post IR rules out payment lotteries that satisfy IR in expectation but leave buyers with negative utility in some realizations.\footnote{Conditional on the bids in a second-price auction, the winner is indifferent between paying the second-highest bid with certainty and facing a payment lottery with the same mean, even if some realizations exceed their value. Such lotteries could allow the seller to increase revenue by selecting a higher-payment realization.} These stronger requirements leave the dark first-price auction unaffected. It satisfies strong ex-post IR, and its only randomization concerns tie-breaking among equal highest bids, which the seller cannot manipulate to increase revenue. Thus, the characterization in Theorem~\ref{thm:dark-regular} remains valid under the analogous stronger requirements for dark auctions.

\begin{theorem}
  \label{thm:lit-posted-price}
  Assume common posterior and full support. Any revenue-maximizing lit auction that is Bayesian incentive-compatible, strongly ex-post seller identity-compatible, and strongly ex-post individually rational is essentially a posted-price mechanism.
\end{theorem}

Here, ``essentially a posted-price mechanism'' means that there exists a common posted price
\[
  \rho \in \arg \max_{\hat  \rho \in [0,1]} \hat  \rho \sum_{n \in \N} \gamma_n \left[1-F(\hat \rho)^n\right]
\]
such that the object is sold if and only if some buyer's value exceeds \(\rho\); the selected winner has value above \(\rho\) and pays \(\rho\), while every other buyer pays zero. The price \(\rho\) need not coincide with the Myersonian reserve, and the result does not require regularity. If ``strongly'' is dropped from either notion, a posted-price mechanism need not maximize revenue; see Example~\ref{ex:non-random} in the Online Appendix.

Theorem~\ref{thm:lit-posted-price} shows that deterring shill bidding in lit auctions requires more than sacrificing optimality: the seller must forgo the ability to extract additional revenue from realized competition. The key observation is that any payment generated by the presence of other buyers could also be generated through shills. Strong ex-post seller identity compatibility therefore bounds each buyer's payment by what the seller could obtain from that buyer alone.\footnote{Related forces from the buyer side appear in \citet*{Jagadeesan-Segal-Tordjman-LimitsAuctionsExAnte-2025} and \citet{Wu-CollusionRobustAuctionDesign-2025}, where buyers under collusion behave as if they were the sole participant.} Together with strong ex-post IR and BIC, revenue maximization reduces the auction to a common posted price; see Appendix~\ref{app:proof-lit-posted-price} for details.

The posted price can still reflect expected competition through its dependence on the value distribution and the distribution of the number of buyers, but it cannot respond to realized competition. Its role therefore differs from that of a reserve price in standard auction theory: rather than merely setting a price floor while allowing competition to determine the final payment, it fixes the payment to eliminate the seller's incentive to manufacture competition. This result provides an additional explanation for the prevalence of posted-price mechanisms in practice, beyond their simplicity \citep{Pycia-Troyan-TheorySimplicityGames-2023}.

\subsection{Non-Identity-Compatible Auctions}

This paper focuses on identity-compatible auctions. One might nevertheless expect the seller to prefer a non-identity-compatible auction if shill bidding is profitable. To examine this intuition, we consider equilibria of an auction augmented with an identity-selection stage, so that participants optimally choose both the number of identities and the strategies played through them. The following result shows why this intuition is incomplete: every such equilibrium outcome can be replicated by a direct dark auction in which truthful reporting is optimal and no side permitted to use multiple identities can profit from doing so.

\begin{theorem}
  \label{thm:non-identity-compatible}
  Consider a Bayesian Nash equilibrium of a lit or dark auction augmented with the identity-selection stage. There exists an outcome-equivalent direct dark auction that is Bayesian incentive-compatible under truthful reporting and Bayesian identity-compatible for each side permitted to use multiple identities in the original auction.
\end{theorem}

The argument parallels the standard Revelation Principle; see Online Appendix~\ref{app:proof-non-identity-compatible}. The theorem explains why a seller need not prefer a non-identity-compatible auction merely because shill bidding appears profitable. Such reasoning holds buyers' behavior fixed. If buyers anticipate the seller's use of shills and respond optimally, then the seller's shill-bidding strategy and the buyers' equilibrium responses jointly determine the equilibrium outcome. The theorem shows that the same outcome can be implemented by an identity-compatible direct dark auction. It therefore provides a revelation-principle-type justification for restricting attention to identity-compatible auctions.

The original auction may be indirect or dynamic, whereas the outcome-equivalent induced auction is direct and static. This distinction matters in the presence of shill bidding. In general, an indirect or dynamic implementation of a direct dark auction that is BIC and Bayesian seller or buyer identity-compatible need not itself be Bayesian seller or buyer identity-compatible. The reason is that intermediate information and opportunities to react to it may create profitable deviations that are unavailable in the induced direct auction.\footnote{The observation that indirect mechanisms admit more deviations than direct mechanisms is also noted by \citet{Levin-Peck-MisbehaviorCommonValue-2023}, who refer to it as the misbehavior principle.} For example, as discussed around Proposition~\ref{prop:second-price}, a lit second-price auction with an appropriate reserve schedule can be Bayesian seller identity-compatible, while the lit English auction with the same reserve schedule need not be. Similarly, in Example~\ref{ex:two-stage-auction}, the two-stage auction fails Bayesian buyer identity compatibility even though its static direct implementation is Bayesian buyer identity-compatible.

The direct auction in Theorem~\ref{thm:non-identity-compatible} is dark even when the original auction is lit. This is consistent with the logic of the Revelation Principle: the direct auction asks buyers only for their values and maps these reports directly into outcomes, without reproducing the intermediate information generated by the original auction. Online Appendix~\ref{app:proof-non-identity-compatible} explains why a lit implementation is not guaranteed, even when the original auction is lit.

\subsection{Regularity}
Our results can be extended to irregular value distributions. Theorem~\ref{thm:lit-regular-ir} extends without any additional assumptions.
\begin{theorem}
  \label{thm:lit-irregular-ir}
  Assume common posterior and full support. If a lit auction \((\Gamma^\N, S^\N)\) is optimal, then it cannot be both ex-post seller identity-compatible and ex-post individually rational.
\end{theorem}
The proof follows the same basic idea as the proof of Theorem~\ref{thm:lit-regular-ir}. Under ironing, however, the object need not be allocated to the highest-value buyer, so the surplus and utility comparisons used in the regular case no longer hold exactly. The argument therefore requires a more delicate approximation; see Online Appendix~\ref{app:proof-lit-irregular-ir}.

To extend Theorem~\ref{thm:lit-regular-bb} to irregular distributions, we impose an additional consistency condition on the allocation rule across population sizes.

\begin{definition}[Inactive-Identity Neutrality]\label{def:inactive-identity-neutrality}
  A lit auction \((\Gamma^{\N},S^{\N})\) is \textbf{inactive-identity neutral} if, for every \(n\geq2\), every buyer \(i\), every \(\hat\theta_i\in\Theta\) satisfying \(Q_i^{\Gamma^n,S^n}(\hat\theta_i)=0\), every \(\theta^n_{-i}\in\Theta^{n-1}\), and every \(j\neq i\),
  \[
    q_j^{\Gamma^n,S^n}(\hat\theta_i,\theta^n_{-i})
    =
    q_j^{\Gamma^{n-1},S^{n-1}}(\theta^n_{-i}).
  \]
\end{definition}

In an \(n\)-buyer auction, buyer \(i\) with value \(\theta_i\) is \textbf{inactive} if, under the prescribed strategies, their probability of winning is zero for almost every profile of the other buyers' values. For example, in a second-price auction, a buyer whose value lies below the reserve price is inactive. Inactive-identity neutrality requires that removing an inactive buyer leave every remaining buyer's allocation probability unchanged. It imposes no corresponding restriction on payments.

The canonical optimal allocation rule satisfies this property: at every population size, it allocates uniformly among buyers with the highest positive ironed virtual value and retains the object if all ironed virtual values are nonpositive. Under this rule, removing an inactive buyer leaves the set of potential winners, and hence the remaining buyers' allocation probabilities, unchanged. Optimality alone, however, does not require consistency across population sizes in either tie-breaking among buyers or the treatment of zero ironed virtual values.

\begin{theorem}
  \label{thm:lit-irregular-bb}
  Assume common posterior, full support, and inactive-identity neutrality. If a lit auction \((\Gamma^\N, S^\N)\) is optimal, then it cannot be both ex-post seller identity-compatible and Bayesian buyer identity-compatible.
\end{theorem}

The proof closely follows the regular case, with inactive identities playing the role of low-value shills. See Online Appendix~\ref{app:proof-lit-irregular-bb} for details. The result would fail without inactive-identity neutrality; see Example~\ref{ex:inactive-identity-neutrality} in the Online Appendix for a counterexample.

\subsection{Interdependent Values with Correlated Signals}

We extend the characterization of dark first-price auctions in Theorem~\ref{thm:dark-regular} to interdependent values with correlated signals, which include irregular independent private values as a special case. In what follows, we refer to \(\theta_i\) as a signal rather than a value, to distinguish it from a buyer's value function, which may depend on the signals of all buyers. We focus on direct dark auctions under truthful reporting, so that \(\Sigma^d=\Theta\) and \(S^d(\theta_i)=\theta_i\).

\paragraph{Correlated signals and population uncertainty.}

We maintain the full-support assumption that \(\gamma_n=\Pr(\tilde n=n)>0\) for every \(n\in\N\), with \(\lambda=\sum_{n\in\N}n\gamma_n\in(0,\infty)\). Conditional on \(\tilde n=n\), the signal profile \(\theta^n\in[0,1]^n\) has a symmetric, positive, continuous density \(f_n: \Theta^n \to \R_{>0}\). The IPV benchmark is obtained as the special case \(f_n(\theta^n)=\prod_{i\leq n}f(\theta_i)\). Affiliation may be imposed by requiring each \(f_n\) to be log-supermodular, but is not needed for our result. Let \(f_{n,1}\) denote the one-dimensional marginal of \(f_n\), and define
\[
  \ell(\theta_i):=\sum_{n \in \N}n\gamma_n f_{n,1}(\theta_i),
  \qquad
  \mu(d\theta_i):=\frac{\ell(\theta_i)}{\lambda}\dd\theta_i.
\]
Assume that \(0<\ell(\theta_i)<\infty\) for every \(\theta_i\in[0,1]\). The marginal distribution of a buyer's signal is then \(\mu\). After observing signal \(\theta_i\), buyer \(i\)'s posterior over the total number of buyers and the other buyers' signals is
\[
  K_{\theta_i}(n,d\theta^n_{-i})
  =
  \frac{n\gamma_n f_n(\theta_i,\theta^n_{-i})}{\ell(\theta_i)}
  \dd\theta^n_{-i}.
\]
For \(n=1\), this expression denotes an atom at the empty profile of other buyers' signals: \( K_{\theta_i}(1,\{\varnothing\}) = \frac{\gamma_1 f_1(\theta_i)}{\ell(\theta_i)} \). Unlike in the IPV benchmark, correlation allows a buyer's signal to convey information about the realized number of buyers.

\paragraph{Interdependent values.}
For each \(n \in \N\), at an \(n\)-buyer signal profile \(\theta^n\), buyer \(i\)'s value is \(v^n(\theta_i,\theta^n_{-i})\). These value functions are finite, nonnegative, and symmetric in the other buyers' signals. Let \(Y\) be the highest signal among the other buyers, with \(Y=-\infty\) if there are no other buyers. Define \(G\) for \(r\geq0\) and \(V\) for \(r>0\) by
\begin{align*}
  G(r\mid\theta_i)&:=\Pr_{K_{\theta_i}}(Y<r),\\
  V(\theta_i,r)&:=\E_{K_{\theta_i}}\!\left[
  v^{\tilde n}(\theta_i,\tilde\theta^{\tilde n}_{-i})\cond Y<r\right].
\end{align*}
Thus \(G(r\mid\theta_i)\) is the probability that all other buyers' signals are below \(r\). We assume that \(G\) is continuous on \([0,1]^2\). The value \(V(\theta_i,r)\) is the buyer's expected value in those circumstances. Positive densities ensure that the conditioning event has positive probability for every realization of the buyer's own signal and every \(r>0\). We assume these expected values are well defined and finite.

We impose the following assumptions on value functions for each \(n \in \N\):
\begin{itemize}
  \item Monotonicity: For every \(\theta^n_{-i}\in\Theta^{n-1}\), \(v^n(\theta_i,\theta^n_{-i})\) is strictly increasing in \(\theta_i\). Moreover, for every \(\theta^n\in\Theta^n\), if \(\theta_i>\theta_j\), then
    \( v^n(\theta_i,\theta^n_{-i}) > v^n(\theta_j,\theta^n_{-j})\).

  \item Value floor: For every \(\theta^n_{-i}\in\Theta^{n-1}\), \( v^n(\theta_i,\theta^n_{-i}) \geq v^1(\theta_i)\).

  \item Continuity: \(V\) admits a continuous extension to \([0,1]^2\) satisfying \(\lim_{r\downarrow0}V(\theta_i,r) = V(\theta_i,0) = v^1(\theta_i) \) for every \(\theta_i\in\Theta\).
\end{itemize}
For example, \(v^n(\theta_i,\theta^n_{-i})=\theta_i+\frac{1}{3}\max_{j\neq i}\theta_j\), where the maximum of the empty set is zero, satisfies these conditions. The IPV benchmark is obtained as the special case \(v^n(\theta_i,\theta^n_{-i})=\theta_i\).

\paragraph{BIC, IR, and identity compatibility.}

The definitions of BIC, IR, and identity compatibility introduced above extend to the interdependent-value setting by replacing the private value \(\theta_i\) with \(v^n(\theta_i,\theta^n_{-i})\) in the relevant payoff expressions. A buyer's value is unaffected by fictitious identities they create themselves. A fictitious identity created by another party enters the buyer's value function through its reported signal.\footnote{As in Footnote~\ref{fn:existence}, we assume that the relevant expectations are finite or well defined. The absolute integrability requirement in Footnote~\ref{fn:integrability} is evaluated under the joint distribution of population sizes and signals.} The seller's prior over the number of buyers and their signals is given by \(\gamma_n f_n(\theta^n)\). Conditional on signal \(\theta_i\), buyer \(i\)'s posterior over the number of buyers and the other buyers' signals is \(K_{\theta_i}\).

We impose the following restrictions on the induced allocation rule.

\begin{definition}[Conditional Efficiency and Sale Consistency]
  A dark auction \((\Gamma^d,S^d)\) is \textbf{conditionally efficient} if, for every \(n\in\N\) and every \(\theta^n\in\Theta^n\),
  \[
    \sum_{i=1}^n q_i^{\Gamma^{d,n},S^d}(\theta^n)\in\{0,1\},
    \
    q_i^{\Gamma^{d,n},S^d}(\theta^n)>0
    \quad\Longrightarrow\quad
    \theta_i=\max_{j\leq n}\theta_j.
  \]

  A dark auction \((\Gamma^d,S^d)\) is \textbf{sale-consistent} if, for every buyer \(i\), every \(n\in\N\), every \(\theta^n\in\Theta^n\), and every \(\hat\theta<\theta_i\),
  \[
    q_i^{\Gamma^{d,n},S^d}(\theta^n)=1
    \quad\Longrightarrow\quad
    q_i^{\Gamma^{d,n+1},S^d}(\theta^n,\hat\theta)=1.
  \]
\end{definition}

The optimal allocation rule in the regular IPV benchmark satisfies both allocation restrictions almost everywhere. Their economic content is particularly transparent under irregular IPV. Conditional efficiency requires that a unique highest-value buyer receive the object with probability one whenever the object is sold. It therefore rules out reducing the allocation probability of a unique highest-value buyer, even though such distortions may arise under irregularity through ironing. Sale consistency requires that a buyer who wins with \(n\) buyers continue to win when an additional lower-value buyer is introduced. It therefore rules out canceling a sale merely because an additional lower-value buyer is present, even though such cancellations may be revenue-enhancing under irregularity by reducing allocation to buyers with negative virtual values. Such allocation distortions are uncommon in practice. Example~\ref{ex:irregular} in the Online Appendix shows that dropping either restriction can strictly increase revenue relative to every dark first-price auction with feasible bid space \([r,+\infty)\), regardless of the reserve \(r\). We impose these restrictions to rule out such distortions.

\begin{theorem}
  \label{thm:interdependent}
  Consider the interdependent-values environment defined above. Any direct dark auction \((\Gamma^d,S^d)\) that is Bayesian incentive-compatible, ex-post individually rational, conditionally efficient, sale-consistent, ex-post seller identity-compatible, and Bayesian buyer identity-compatible is essentially a dark first-price auction.
\end{theorem}

Intuitively, moving beyond IPV introduces correlation, creating additional channels through which shill bidding can affect the auction. Since the dark first-price auction is essentially the unique auction satisfying the corresponding requirements in the IPV benchmark, it is the leading candidate to continue satisfying them under interdependent values. The argument is similar to that used to prove Theorem~\ref{thm:dark-regular}, although additional care is required because revenue maximization is replaced by restrictions on the allocation rule; see Online Appendix~\ref{app:proof-interdependent} for details. Example~\ref{ex:affiliated} in the Online Appendix provides a construction with strictly affiliated signals and interdependent values satisfying our assumptions. Previous work shows that dark first-price auctions can generate higher revenue than lit first-price auctions under affiliated private values with constant or decreasing absolute risk aversion, including risk neutrality \citep{Matthews-ComparingAuctionsRisk-1987}, and under affiliated common values with costly entry \citep{Murto-Valimaki-AffiliatedCommonValue-2025}. Theorem~\ref{thm:interdependent} therefore provides an additional reason why the dark first-price auction may be attractive beyond the IPV benchmark.

Importantly, this advantage need not come at a revenue cost. Even under strict affiliation, the revenue ranking between dark first-price and second-price auctions can be ambiguous \citep{Matthews-ComparingAuctionsRisk-1987}, and the ranking implied by the linkage principle \citep{Milgrom-Weber-TheoryAuctionsCompetitive-1982} for lit auctions does not generally extend to dark first-price auctions.\footnote{Intuitively, English auctions benefit from positive linkage: higher signals suggest stronger opponents and hence higher payments. Dark first-price auctions benefit from concealing the number of buyers: conditional on winning, higher signals instead suggest fewer opponents, so disclosure can link higher signals to lower expected payments \citep{Matthews-ComparingAuctionsRisk-1987}. The revenue ranking between the two formats depends on which effect dominates.} Thus, deterring shill bidding through a dark first-price auction need not require sacrificing revenue relative to second-price or English auctions. The dark first-price auction may offer the best of both worlds: deterring shill bidding while generating higher revenue than alternative formats.

\section{Concluding Remarks}
\label{sec:conclusion}

This paper highlights the importance of concealing the number of bidders in deterring shill bidding. In our model, the number of buyers is random but exogenous and realized before the auction begins. An interesting direction for future research is to study disclosure when buyers arrive sequentially and decide whether and when to enter at a cost. Even without shill bidding, disclosure may affect not only bidding but also participation. Information about earlier entrants may influence subsequent entry decisions, while early entrants anticipate how disclosure of their own participation affects future competition. Disclosure therefore shapes the competition that emerges rather than merely revealing the number of existing competitors. This feedback makes both the content and timing of disclosure potentially important design choices. When is immediate disclosure preferable to concealment or delayed disclosure? How do revenue-maximizing disclosure policies compare with those that maximize surplus net of entry costs?

\bibliographystyle{econ}
\bibliography{Identity}

\appendix

\section{Omitted Proofs}
\label{app:proofs}
\subsection{Proof of Proposition~\ref{prop:second-price}}\label{app:proof-second-price}

Define
\[
  H(\theta)=\theta\bigl(1-F(\theta)\bigr),
  \qquad
  R_n(\rho)
  =
  n\int_\rho^1
  \phi(\theta)F(\theta)^{n-1}f(\theta) \dd \theta,
\]
where \(R_n(\rho)\) is the expected revenue of a second-price auction with \(n\) buyers and reserve price \(\rho\). By continuity of \(\phi\) and the definition of \(\bar\rho\), we have \(\phi(\theta)\geq0\) for every \(\theta\in[\bar\rho,1]\). Consequently,
\[
  H'(\theta)
  =
  -f(\theta)\phi(\theta)
  \leq0,
  \qquad
  R_n'(\theta)
  =
  -n\phi(\theta)F(\theta)^{n-1}f(\theta)
  \leq0
\]
on \([\bar\rho,1]\). Thus both \(H\) and \(R_n\) are weakly decreasing on this interval.

Let \(\rho_n\) denote the reserve price when there are \(n\) bidders, and let \((\rho_n)_{n\in\N}\) denote the corresponding reserve schedule.

We prove the more general result that both lit second-price and English auctions are Bayesian seller identity-compatible under any nondecreasing reserve schedule satisfying \(\rho_n\geq\bar\rho\) for every \(n\in\N\). A common reserve \(\rho\geq\bar\rho\) gives the corresponding result for their dark counterparts.

Fix the number of buyers \(n\) and the number of shills \(k\). We first consider the lit second-price auction. If \(s\) is the highest shill bid, adding the shills induces the effective reserve
\[
  \hat\rho=\max\{\rho_{n+k},s\}.
\]
Shill bids above \(1\) yield zero revenue and can be disregarded. Otherwise, the expected revenue is exactly that of an \(n\)-buyer second-price auction with reserve price \(\hat\rho\). Since
\[
  \hat\rho\geq\rho_{n+k}\geq\rho_n\geq\bar\rho,
\]
monotonicity of \(R_n\) implies
\[
  R_n(\hat\rho)\leq R_n(\rho_n).
\]
Thus, shill bidding cannot increase expected revenue in any lit second-price auction with a nondecreasing reserve schedule satisfying \(\rho_n\geq\bar\rho\) for every \(n \in \N\), regardless of the realized number of buyers or the number of shills. Hence, every such lit second-price auction is Bayesian seller identity-compatible.

Now consider the lit English auction. Fix \(n\) buyers and \(k\) shills, and write \(r:=\rho_{n+k}\). We first compare the seller's deviation with an \(n\)-buyer English auction without shills but with the same reserve \(r\), whose expected revenue is \(R_n(r)\).

If no buyer accepts the reserve, the seller receives zero under both auctions. Otherwise, consider the history up to the point at which only the highest-valued buyer remains active. Let \(y\) denote the last exit price among the other buyers, setting \(y=r\) if only one buyer accepts the reserve. The auction cannot end before this point while two buyers remain active. Conditional on this history, the remaining buyer's value \(X\) is distributed according to \(F\) truncated from below at \(y\): the buyer's behavior so far reveals only that \(X\geq y\). In the benchmark auction without shills and with reserve \(r\), the same value profile yields revenue \(y\).

If all shills have exited by this point, the seller likewise receives \(y\). Otherwise, along the continuation in which the remaining buyer stays active, the seller's strategy determines a price \(P\geq y\) at which the last shill exits. If the buyer survives to \(P\), they win and pay \(P\); otherwise, a shill wins and the seller receives zero from buyers. A continuation in which the last shill never exits, or exits above \(1\), yields zero revenue. For \(P\leq1\), conditional expected revenue is
\[
  P\frac{1-F(P)}{1-F(y)}
  =\frac{H(P)}{1-F(y)}
  \leq\frac{H(y)}{1-F(y)}
  =y,
\]
where the inequality follows because \(P\geq y\geq r\geq\bar\rho\) and \(H\) is weakly decreasing on \([\bar\rho,1]\).

Averaging over histories, expected revenue under the deviation is therefore at most \(R_n(r)\). Since the reserve schedule is nondecreasing and \(R_n\) is weakly decreasing on \([\bar\rho,1]\),
\[
  R_n(r)=R_n(\rho_{n+k})\leq R_n(\rho_n).
\]
Thus, for every realized number of buyers and every number of shills, no history-dependent shill-bidding strategy increases expected revenue. Hence, the lit English auction is Bayesian seller identity-compatible. \qed

\begin{remark}[Common reserves]
  Consider first a lit English auction with a common reserve \(\rho<\bar\rho\). There exist \(\rho<a<b<1\) such that \(\phi<0\) on \([a,b]\), so \(H(\theta)=\theta[1-F(\theta)]\) is strictly increasing there. A shill can remain active until only one buyer remains and, if the last dropout occurs at \(y\in[a,b)\), continue to \(b\). Conditional on this event, expected revenue increases from \(y\) to
  \[
    b\frac{1-F(b)}{1-F(y)}
    =
    \frac{H(b)}{1-F(y)}
    >
    \frac{H(y)}{1-F(y)}
    =
    y.
  \]
  Hence, a lit English auction with a common reserve \(\rho\) is Bayesian seller identity-compatible if and only if \(\rho\geq\bar\rho\).

  The same characterization holds for a lit second-price auction with a common reserve. Sufficiency follows from Proposition~\ref{prop:second-price}. For necessity, if \(\rho<\bar\rho\), choose \(a,b\) as above. For all sufficiently large \(n\),
  \[
    R_n(b)-R_n(\rho)
    =
    -n\int_\rho^b\phi(\theta)F(\theta)^{n-1}f(\theta)\dd\theta
    >0,
  \]
  because the integral is asymptotically dominated by values near \(b\), where \(\phi<0\). The seller can then introduce one shill, bid \(b\) when the realized number of buyers is at least \(n+1\), and decline to bid at every other population size.

  By contrast, in a dark second-price auction, or a dark English auction in which only the clock price is publicly revealed, a shill cannot condition on the realized population size or on dropout histories. Its highest bid or exit price simply raises the common reserve. Thus, writing
  \[
    \overline R(r):=\sum_{n\in\N}\gamma_nR_n(r),
  \]
  a common reserve \(\rho\) is Bayesian seller identity-compatible in either dark format if and only if
  \[
    \overline R(\rho)\geq\overline R(r)
    \qquad
    \text{for every }r\geq\rho.
  \]
  In particular, any maximizer of \(\overline R\) is Bayesian seller identity-compatible and may lie below \(\bar\rho\).
\end{remark}

\begin{remark}[Population-contingent reserves in the lit second-price auction]
  For every \(n\in\N\), choose
  \[
    \rho_n^*\in\arg\max_{\rho\in[0,1]}R_n(\rho).
  \]
  The resulting population-contingent reserve schedule is Bayesian seller identity-compatible in the lit second-price auction, even if some reserves lie below \(\bar\rho\). Conditional on \(n\) buyers, adding \(k\) shills with highest bid \(s\) induces the effective reserve
  \[
    \hat\rho=\max\{\rho_{n+k}^*,s\},
  \]
  and optimality of \(\rho_n^*\) gives \(R_n(\hat\rho)\leq R_n(\rho_n^*)\).

  The same reserve schedule need not be Bayesian seller identity-compatible in the corresponding lit English auction. Let \(F\) have the continuous piecewise-linear density
  \[
    \begin{array}{c|rrrrrr}
      \theta&0&0.499&0.501&0.749&0.751&1\\ \hline
      f(\theta)&1.60&1.60&0.08&0.08&0.72&0.72.
    \end{array}
  \]
  Set \(a=0.51\), \(b=0.74\). For this distribution, \(\rho_2^*=\rho_3^*=5/16<a\) and \(\phi<0\) on \([a,b]\), so \(H(\theta)=\theta[1-F(\theta)]\) is strictly increasing there.

  With two buyers, the seller can introduce one shill and, after the lower-valued buyer exits at \(y\in[a,b)\), keep the shill active until \(b\). Conditional on the other buyer's value exceeding \(y\), this changes expected revenue from \(y\) to
  \[
    \frac{H(b)}{1-F(y)}
    >
    \frac{H(y)}{1-F(y)}
    =
    y.
  \]
  Hence, conditional on two buyers, the deviation yields the positive expected gain
  \[
    \delta
    :=
    2\int_a^b f(y)\bigl[H(b)-H(y)\bigr]\dd y
    >0.
  \]

  To make the deviation profitable ex ante, let \(\varepsilon=10^{-4}\) and set
  \[
    \gamma_1=\frac{\varepsilon}{2},
    \qquad
    \gamma_2=1-\varepsilon,
    \qquad
    \gamma_n=\varepsilon2^{-(n-1)}
    \quad\text{for }n\geq3.
  \]
  The seller introduces one shill before observing the realized population size, uses the deviation above when there are two buyers, and otherwise withdraws the shill immediately. Although the additional identity may change the population-contingent reserve at other population sizes, those events have total probability \(\varepsilon\), and the resulting loss is at most \(\varepsilon\). The seller's ex-ante gain is therefore at least
  \[
    (1-\varepsilon)\delta-\varepsilon>0.
  \]
  Thus, population-contingent optimal reserves make the lit second-price auction Bayesian seller identity-compatible, but need not do so for the corresponding English auction.
\end{remark}

\subsection{Proof of Theorem~\ref{thm:lit-regular-ir}}
\label{app:proof-lit-regular-ir}
Recall that, for the indirect lit auction, we write
\[
  q_i^{\Gamma^n,S^n}(\theta^n) :=q_i^{\Gamma^n}\!\left(S^n(\theta^n)\right),\qquad t_i^{\Gamma^n,S^n}(\theta^n) :=t_i^{\Gamma^n}\!\left(S^n(\theta^n)\right).
\]
Their interim counterparts are \(Q_i^{\Gamma^n,S^n}\), \(T_i^{\Gamma^n,S^n}\), and \(U_i^{\Gamma^n,S^n}(\theta_i)=\theta_i Q_i^{\Gamma^n,S^n}(\theta_i)-T_i^{\Gamma^n,S^n}(\theta_i)\). A shill bidder can pretend to have value \(z \in \Theta\) by following \(S^n(z)\). All deviations below have this form, so they are available in the original indirect auction.

Full support of the distribution of the number of buyers implies that every \(n\)-buyer component of an optimal lit auction is itself an optimal \(n\)-buyer auction in the sense of \citet{Myerson-OptimalAuctionDesign-1981}. Under regularity, the object goes to the highest-value buyer above the Myersonian reserve \(\rho^*\), and
\begin{align}
  Q_i^{\Gamma^n,S^n}(\theta_i) &=\one\{\theta_i>\rho^*\}F(\theta_i)^{n-1},\notag\\
  U_i^{\Gamma^n,S^n}(\theta_i) &=\one\{\theta_i>\rho^*\}\int_{\rho^*}^{\theta_i}F(s)^{n-1}\dd s,\notag\\
  T_i^{\Gamma^n,S^n}(\theta_i) &=\theta_i Q_i^{\Gamma^n,S^n}(\theta_i)-U_i^{\Gamma^n,S^n}(\theta_i). \label{eq:regular-lit}
\end{align}

\paragraph{Step~1: Add competition without changing the winner.}

Fix \(n \in \N\), and let \(\tilde H_n:=\max_{j\leq n}\tilde\theta_j\). The seller commits ex ante to adding one shill, so the number of shills is fixed before buyers' values are observed, as required by ex-post seller identity compatibility in lit auctions. Only the shill's strategy is conditioned on the realized values. After observing \(\tilde\theta^n\), the seller draws \(\tilde z\) from \(F\) conditional on \(\tilde z<\tilde H_n\) and has the shill follow \(S^{n+1}(\tilde z)\). Thus,
\[
  \Pr(\tilde z\in\dd z\mid\tilde\theta^n)
  =
  \frac{\one\{z<\tilde H_n\}}{F(\tilde H_n)}\,\dd F(z).
\]
Because the seller's optimal deviation payoff is at least the expected payoff from any feasible randomized deviation, it suffices to show that this deviation is profitable.

The shill never has the highest value. Consequently, the deviation preserves both the identity of the winning buyer and the decision whether to sell. The total surplus is therefore the same with and without the deviation:
\[
  W_n:=\E_{\tilde H_n}\left[\tilde H_n\one\{\tilde H_n>\rho^*\}\right].
\]

\paragraph{Step~2: Bound the utilities left to buyers.}
Let
\[
  u_i^{\Gamma^{n+1},S^{n+1}}(\theta^n,z) :=\theta_iq_i^{\Gamma^{n+1},S^{n+1}}(\theta^n,z) -t_i^{\Gamma^{n+1},S^{n+1}}(\theta^n,z).
\]
Ex-post IR makes this utility nonnegative. Relative to \(n+1\) independent draws from \(F\), the deviation reweights an \((n+1)\)-buyer value profile by \(\one\{\tilde z<\tilde H_n\}/F(\tilde H_n)\). Since \(\tilde H_n\geq\tilde\theta_i\), this weight is at most \(1/F(\tilde\theta_i)\). Nonnegativity therefore gives
\begin{align}
  \E_{\tilde{\theta}^n, \tilde{z}}\!\left[u_i^{\Gamma^{n+1},S^{n+1}}\right] &\le\E_{\tilde{\theta}^{n},\tilde{\theta}}\!\left[ \frac{u_i^{\Gamma^{n+1},S^{n+1}}}{F(\tilde\theta_i)} \right]\notag\\
  &=\int_0^1 \frac{U_i^{\Gamma^{n+1},S^{n+1}}(\theta_i)}{F(\theta_i)}\dd F(\theta_i). \label{eq:regular-utility-bound}
\end{align}
This is where ex-post IR enters the proof. It allows us to bound utilities under the seller's shill-bidding strategy using the known interim utilities under the equilibrium strategies, because every buyer's utility is nonnegative.

\paragraph{Step~3: The reduction in utilities raises revenue.}
Revenue from buyers is the total surplus minus their total utilities. Writing \(R_n^{\mathrm{dev}}\) for the seller's deviation revenue, anonymity and \eqref{eq:regular-utility-bound} imply
\begin{align*}
  R_n^{\mathrm{dev}} &\ge W_n-n\int_0^1 \frac{U_i^{\Gamma^{n+1},S^{n+1}}(\theta_i)}{F(\theta_i)}\dd F(\theta_i), \\
  \Rev(\Gamma^n,S^n) &=W_n-n\int_0^1U_i^{\Gamma^n,S^n}(\theta_i)\dd F(\theta_i).
\end{align*}
For every \(\theta_i > \rho^*\), the envelope formula \eqref{eq:regular-lit} yields
\begin{align*}
  U_i^{\Gamma^n,S^n}(\theta_i) -\frac{U_i^{\Gamma^{n+1},S^{n+1}}(\theta_i)}{F(\theta_i)} &=\int_{\rho^*}^{\theta_i}F(s)^{n-1} \left[1-\frac{F(s)}{F(\theta_i)}\right]\dd s>0.
\end{align*}
The strict inequality follows from full support of \(F\). Hence \(R_n^{\mathrm{dev}}>\Rev(\Gamma^n,S^n)\) for every \(n \in \N\), contradicting ex-post seller identity compatibility. \qed

\subsection{Proof of Theorem~\ref{thm:lit-regular-bb}}
\label{app:proof-lit-regular-bb}

Let \(L=(0,\rho^*)\), where \(\rho^*\) is the Myersonian reserve. Consider a shill controlled either by a buyer or by the seller, and draw its value according to
\begin{equation*}
  \nu(dz):=\frac{\one\{z\in L\}}{F(L)}\dd F(z).
\end{equation*}
By the optimal allocation rule in \eqref{eq:regular-lit}, a buyer with value in \(L\) has zero allocation and zero interim expected payment in every \(n\)-buyer auction \((\Gamma^n, S^n)\). Moreover, adding such a shill leaves the allocations of the other buyers unchanged. In particular, for every buyer \(i\) and almost every value \(\theta_i\),
\begin{equation}
  \E_{\tilde\theta^n_{-i},\tilde z} \left[q_i^{\Gamma^{n+1},S^{n+1}}(\theta_i,\tilde\theta^n_{-i},\tilde z)\right] =Q_i^{\Gamma^n,S^n}(\theta_i), \qquad \tilde z\sim\nu. \label{eq:IN-interim}
\end{equation}
The following lemma relies only on these properties of low-value shills. We therefore reuse it without regularity in Online Appendix~\ref{app:proof-lit-irregular-bb}, where such low-value shills are referred to as inactive identities (Definition~\ref{def:inactive-identity-neutrality}).

\begin{lemma}[Adding an inactive identity]\label{lem:common-payment}
  Suppose an optimal lit auction \((\Gamma^\N,S^\N)\) is both ex-post seller identity-compatible and Bayesian buyer identity-compatible. Suppose further that there exists an interval \(L\) with \(F(L)>0\) such that, in every \(n\)-buyer auction \((\Gamma^n, S^n)\), buyers with values in \(L\) have zero interim expected allocations and zero interim expected payments. Let \(\nu=F(\cdot\mid L)\), and suppose that adding an identity whose value is drawn independently from \(\nu\) preserves each buyer's interim expected allocation:
  \[
    \E_{\tilde\theta^n_{-i},\tilde z}\!\left[
      q_i^{\Gamma^{n+1},S^{n+1}}(\theta_i,\tilde\theta^n_{-i},\tilde z)
    \right]
    =Q_i^{\Gamma^n,S^n}(\theta_i),
    \qquad \tilde z\sim\nu.
  \]
  Let \(a:=n+1\) denote the added identity, and define buyer \(i\)'s total interim expected payment through the two identities by
  \[
    \widehat T_i^{\Gamma^n,S^n}(\theta_i)
    :=
    \E_{\tilde\theta^n_{-i},\tilde z}\!\left[
      t_i^{\Gamma^{n+1},S^{n+1}}(\theta_i,\tilde\theta^n_{-i},\tilde z)
      +t_a^{\Gamma^{n+1},S^{n+1}}(\theta_i,\tilde\theta^n_{-i},\tilde z)
    \right].
  \]
  Then, for every \(n\in\N\) and for almost every \(\theta_i\in\Theta\),
  \begin{equation}\label{eq:AequalsT}
    \widehat T_i^{\Gamma^n,S^n}(\theta_i)
    =T_i^{\Gamma^n,S^n}(\theta_i),
  \end{equation}
  and
  \[
    T_i^{\Gamma^2,S^2}(\theta_i)
    \le T_i^{\Gamma^1,S^1}(\theta_i).
  \]
\end{lemma}

\begin{proof}
  We first establish the integrability needed to average over population sizes. We then compare the seller's and buyer's incentives to create the same inactive identity.

  \paragraph{Step~1: Bound absolute payments uniformly across population sizes.}
  Write \(R_1:=\Rev(\Gamma^1,S^1)\). For every \(m\ge2\), define
  \[
    M_m(\theta_i)
    :=
    \sup_{\boldsymbol{\sigma}_{0,m-1}^{m} \in (\Sigma^{m})^{m-1}}
    t_i^{\Gamma^m}\left(S^m(\theta_i),\boldsymbol{\sigma}_{0,m-1}^{m}\right).
  \]
  With one buyer, the seller can create \(m-1\) shills. Ex-post seller identity compatibility therefore implies
  \begin{equation}\label{eq:M-integral}
    \int_0^1 M_m(\theta_i)\dd F(\theta_i)\le R_1.
  \end{equation}
  Optimality and the envelope formula give \(0\le T_i^{\Gamma^m,S^m}(\theta_i)\le1\). Since the supremum dominates the interim average, \(M_m(\theta_i)\ge T_i^{\Gamma^m,S^m}(\theta_i)\ge0\).

  Let \(x^+:=\max\{x,0\}\). The standing integrability assumption (Footnote~\ref{fn:integrability}) ensures that payments are absolutely integrable at each fixed population size. Using \eqref{eq:M-integral} and nonnegative expected payments, we obtain
  \begin{align*}
    \E_{\tilde\theta^m}\!\left[
      \left(t_i^{\Gamma^m,S^m}(\tilde\theta^m)\right)^+
    \right]
    &\le R_1,\\
    \E_{\tilde\theta^m}\!\left[
      \left|t_i^{\Gamma^m,S^m}(\tilde\theta^m)\right|
    \right]
    &=
    2\E_{\tilde\theta^m}\!\left[
      \left(t_i^{\Gamma^m,S^m}(\tilde\theta^m)\right)^+
    \right]
    -\int_0^1 T_i^{\Gamma^m,S^m}(\theta_i)\dd F(\theta_i)
    \le2R_1.
  \end{align*}
  Thus each buyer's expected absolute payment is bounded uniformly in \(m\).

  Set \(\alpha:=F(L)>0\). The distribution \(F^n\otimes\nu\) has density at most \(1/\alpha\) relative to \(F^{n+1}\). Consequently,
  \[
    \E_{\tilde\theta^n,\tilde z}\!\left[
      \left|t_i^{\Gamma^{n+1},S^{n+1}}(\tilde\theta^n,\tilde z)\right|
      +\left|t_a^{\Gamma^{n+1},S^{n+1}}(\tilde\theta^n,\tilde z)\right|
    \right]
    \le\frac{4R_1}{\alpha}.
  \]
  Because \(\sum_{n\in\N}p_n=1\), this bound implies
  \[
    \sum_{n\in\N}p_n
    \int_0^1
    \left|
    \widehat T_i^{\Gamma^n,S^n}(\theta_i)
    -T_i^{\Gamma^n,S^n}(\theta_i)
    \right|
    \dd F(\theta_i)
    \le\frac{4R_1}{\alpha}+1<\infty.
  \]
  The exchanges of sums and integrals below are therefore justified.

  \paragraph{Step~2: Pin down the revenue under shill bidding.}
  Let
  \[
    R_n^{\mathrm{add}}
    :=
    \E_{\tilde\theta^n,\tilde z}\!\left[
      \sum_{j=1}^n
      t_j^{\Gamma^{n+1},S^{n+1}}(\tilde\theta^n,\tilde z)
    \right]
  \]
  denote the seller's expected revenue from \(n\) buyers after adding a shill whose value is drawn from \(\nu\). Ex-post seller identity compatibility gives \(R_n^{\mathrm{add}}\le\Rev(\Gamma^n,S^n)\). The added identity's expected payment is zero:
  \[
    \E_{\tilde\theta^n,\tilde z}\!\left[
      t_a^{\Gamma^{n+1},S^{n+1}}(\tilde\theta^n,\tilde z)
    \right]
    =
    \int_L T_a^{\Gamma^{n+1},S^{n+1}}(z)\nu(dz)
    =0.
  \]
  Anonymity therefore implies
  \[
    d_n
    :=
    \int_0^1
    \left[
      \widehat T_i^{\Gamma^n,S^n}(\theta_i)
      -T_i^{\Gamma^n,S^n}(\theta_i)
    \right]\dd F(\theta_i)
    =
    \frac{R_n^{\mathrm{add}}-\Rev(\Gamma^n,S^n)}{n}
    \le0.
  \]

  Now let buyer \(i\) control the added identity. The original identity follows the equilibrium strategy for \(\theta_i\), and the added identity follows the equilibrium strategy for an independent draw \(\tilde z\sim\nu\). The added identity receives zero allocation almost surely, because its interim expected allocation is zero and allocations are nonnegative. By assumption, the original identity's interim expected allocation remains unchanged. Bayesian buyer identity compatibility thus gives
  \[
    D(\theta_i)
    :=
    \sum_{n\in\N}p_n
    \left[
      \widehat T_i^{\Gamma^n,S^n}(\theta_i)
      -T_i^{\Gamma^n,S^n}(\theta_i)
    \right]
    \ge0.
  \]
  Integrating and applying Step~1,
  \[
    0
    \le\int_0^1D(\theta_i)\dd F(\theta_i)
    =\sum_{n\in\N}p_nd_n
    \le0.
  \]
  Hence \(D(\theta_i)=0\) for almost every \(\theta_i\). Since \(p_n>0\) for every \(n\), we also have
  \[
    d_n=0,
    \qquad
    R_n^{\mathrm{add}}=\Rev(\Gamma^n,S^n)
    \qquad\text{for every }n\in\N.
  \]

  \paragraph{Step~3: Pin down the interim expected payment under shill bidding.}
  Fix \(n\). Consider the deviation from Step~2, which leaves the buyer indifferent. Modify only its continuation when the number of bidders is \(n+1\): in that case, the original identity follows the equilibrium strategy for \(\hat\theta_i\) instead of \(\theta_i\). The added identity still follows the strategy for \(\tilde z\sim\nu\), and both identities follow the preceding plan at every other number of bidders. This is an admissible deviation in a lit auction: only the continuation strategies depend on the disclosed number of bidders, while the number of added identities remains fixed.

  Take \(\theta_i\) and \(\hat\theta_i\) such that \(D(\theta_i) = D(\hat \theta_i)=0\) as established in Step~2. Since the baseline deviation yields the buyer's truthful expected utility, Bayesian buyer identity compatibility implies
  \[
    0\ge
    p_n\Bigl\{
      \theta_i\bigl[
        Q_i^{\Gamma^n,S^n}(\hat\theta_i)
        -Q_i^{\Gamma^n,S^n}(\theta_i)
      \bigr]
      -
      \bigl[
        \widehat T_i^{\Gamma^n,S^n}(\hat\theta_i)
        -\widehat T_i^{\Gamma^n,S^n}(\theta_i)
      \bigr]
    \Bigr\}.
  \]
  Dividing by \(p_n>0\) gives the Bayesian incentive compatibility constraints for the interim expected allocation \(Q_i^{\Gamma^n,S^n}\) and interim expected payment \(\widehat T_i^{\Gamma^n,S^n}\) on a full-measure set \(E\subseteq[0,1]\). Define \(\widehat U_i^{\Gamma^n,S^n}(\theta_i):=\theta_iQ_i^{\Gamma^n,S^n}(\theta_i)-\widehat T_i^{\Gamma^n,S^n}(\theta_i)\) on \(E\). Applying these inequalities in both directions and using \(0\leq Q_i^{\Gamma^n,S^n}\leq1\) shows that \(\widehat U_i^{\Gamma^n,S^n}\) is \(1\)-Lipschitz on \(E\). Since \(f>0\), \(E\) has full Lebesgue measure and is dense in \([0,1]\), so this utility extends uniquely to a \(1\)-Lipschitz function on \([0,1]\). The incentive inequalities imply that the extension has derivative \(Q_i^{\Gamma^n,S^n}\) almost everywhere and therefore satisfies the envelope formula. Since the original auction has the same interim expected allocation, there is a constant \(c_n\) such that
  \[
    \widehat T_i^{\Gamma^n,S^n}(\theta_i)
    =
    T_i^{\Gamma^n,S^n}(\theta_i)+c_n
    \qquad\text{for almost every }\theta_i.
  \]
  Integrating gives \(c_n=d_n=0\), proving \eqref{eq:AequalsT}.

  \paragraph{Step~4: Bound two-buyer payments by one-buyer payments.}
  For \(n=1\), Step~2 gives
  \[
    \E_{\tilde\theta_i,\tilde z}\!\left[
      t_i^{\Gamma^2,S^2}(\tilde\theta_i,\tilde z)
    \right]
    =R_1.
  \]
  Since \(t_i^{\Gamma^2,S^2}(\theta_i,z)\le M_2(\theta_i)\) and \eqref{eq:M-integral} bounds the expectation of \(M_2\) by \(R_1\), the nonnegative difference must vanish:
  \begin{equation}\label{eq:max-low}
    t_i^{\Gamma^2,S^2}(\theta_i,z)
    =M_2(\theta_i)
    \qquad\text{for }F\otimes\nu\text{-almost every }(\theta_i,z).
  \end{equation}

  We next show that \(M_2(z)=0\) for \(\nu\)-almost every \(z\). Because \(\nu\) is absolutely continuous with respect to \(F\), \eqref{eq:max-low} also holds when both values are drawn from \(\nu\). Applying anonymity and \eqref{eq:AequalsT}, together with zero interim expected payments on \(L\), gives
  \begin{align*}
    0
    &=
    \int_L\widehat T_i^{\Gamma^1,S^1}(\theta_i)\nu(d\theta_i)\\
    &=
    \E_{\nu\otimes\nu}\!\left[
      t_i^{\Gamma^2,S^2}(\tilde\theta_i,\tilde z)
      +t_a^{\Gamma^2,S^2}(\tilde\theta_i,\tilde z)
    \right]
    =
    2\int_L M_2(z)\nu(dz).
  \end{align*}
  Since \(M_2\ge0\), it follows that \(M_2(z)=0\) for \(\nu\)-almost every \(z\).

  By anonymity and the definition of \(M_2\),
  \[
    t_a^{\Gamma^2,S^2}(\theta_i,z)
    =
    t_i^{\Gamma^2,S^2}(z,\theta_i)
    \le M_2(z)=0
    \qquad\text{for }F\otimes\nu\text{-almost every }(\theta_i,z).
  \]
  The added identity's expected payment is zero, as established in Step~2. Its payment is therefore zero almost surely under \(F\otimes\nu\). Combining this observation with \eqref{eq:AequalsT} and \eqref{eq:max-low} yields
  \[
    T_i^{\Gamma^1,S^1}(\theta_i)
    =
    \widehat T_i^{\Gamma^1,S^1}(\theta_i)
    =
    M_2(\theta_i)
    \ge
    T_i^{\Gamma^2,S^2}(\theta_i),
  \]
  as required.
\end{proof}

Applying Lemma~\ref{lem:common-payment} to the interval \(L=(0,\rho^*)\) yields
\begin{equation}
  T_i^{\Gamma^2,S^2}(\theta_i) \le T_i^{\Gamma^1,S^1}(\theta_i). \label{eq:regular-order}
\end{equation}
For \(\theta_i>\rho^*\), however, optimality pins down the interim expected payments:
\begin{equation*}
  T_i^{\Gamma^1,S^1}(\theta_i)=\rho^*, \qquad T_i^{\Gamma^2,S^2}(\theta_i)=\theta_i F(\theta_i)-\int_{\rho^*}^{\theta_i}F(s)\dd s.
\end{equation*}
In particular,
\begin{equation*}
  \lim_{\theta_i \uparrow 1} T_i^{\Gamma^2,S^2}(\theta_i)=\rho^*+\int_{\rho^*}^1[1-F(s)]\dd s > \rho^*.
\end{equation*}
By continuity, this strict inequality continues to hold on an interval of values below one, contradicting \eqref{eq:regular-order}. This proves Theorem~\ref{thm:lit-regular-bb}.\qed

\subsection{Proof of Theorem~\ref{thm:dark-regular}}
\label{app:proof-dark-regular}

Recall that, for the indirect dark auction, we write
\[
  q_i^{\Gamma^{d,n},S^d}(\theta^n) :=q_i^{\Gamma^{d,n}}\!\left(S^d(\theta^n)\right),\qquad
  t_i^{\Gamma^{d,n},S^d}(\theta^n) :=t_i^{\Gamma^{d,n}}\!\left(S^d(\theta^n)\right).
\]
Their interim counterparts are \(Q_i^{\Gamma^d,S^d}\), \(T_i^{\Gamma^d,S^d}\), and \(U_i^{\Gamma^d,S^d}(\theta_i)=\theta_i Q_i^{\Gamma^d,S^d}(\theta_i)-T_i^{\Gamma^d,S^d}(\theta_i)\). A shill bidder can pretend to have value \(\hat\theta\) by following \(S^d(\hat\theta)\). All deviations below have this form, so they are available in the original indirect auction.

\paragraph{Step~1: Optimality pins down allocations and interim expected payments.}
Under regularity, optimality requires the object to be allocated to the highest-value buyer whenever their value exceeds the Myersonian reserve \(\rho^*\) (Proposition~\ref{prop:ret}). Thus
\begin{equation}
  q_i^{\Gamma^{d,n},S^d}(\theta^n)=\one\left\{\theta_i>\rho^*,\ \theta_i>\max_{j\ne i}\theta_j\right\}. \label{eq:dark-regular-allocation}
\end{equation}
As in the main text, let \(G(\theta_i):=\sum_{n \in \N}p_nF(\theta_i)^{n-1}\) be the probability that all other buyers have values below \(\theta_i\). The envelope formula gives
\begin{align}
  Q_i^{\Gamma^d,S^d}(\theta_i)&=\one\left\{\theta_i>\rho^*\right\}G(\theta_i),\notag\\
  U_i^{\Gamma^d,S^d}(\theta_i)&=\one\left\{\theta_i>\rho^*\right\}\int_{\rho^*}^{\theta_i}G(s)\dd s,\notag\\
  T_i^{\Gamma^d,S^d}(\theta_i)&=\theta_iQ_i^{\Gamma^d,S^d}(\theta_i)-U_i^{\Gamma^d,S^d}(\theta_i). \label{eq:dark-regular-interim}
\end{align}
For \(\theta_i>\rho^*\), the equilibrium bidding strategy in the optimal dark first-price auction is given by
\[
  S^d_{\text{FP}}(\theta_i):=\theta_i-\frac{\int_{\rho^*}^{\theta_i}G(s)\dd s}{G(\theta_i)}.
\]

Consequently,
\begin{equation}
  T_i^{\Gamma^d,S^d}(\theta_i)=S^d_{\text{FP}}(\theta_i)Q_i^{\Gamma^d,S^d}(\theta_i). \label{eq:dark-regular-interim-payment}
\end{equation}
Below the reserve, products of the bid and an allocation probability are understood to be zero.

Ex-post IR makes \(u_i^{\Gamma^{d,n},S^d}(\theta^n):=\theta_iq_i^{\Gamma^{d,n},S^d}(\theta^n)-t_i^{\Gamma^{d,n},S^d}(\theta^n)\) nonnegative. For \(\theta_i \leq \rho^*\), its expectation over the number and values of the other buyers is \(U_i^{\Gamma^d,S^d}(\theta_i)=0\). Since every \(p_n>0\), this utility must be zero almost surely in every \(n\)-buyer auction \((\Gamma^{d,n},S^d)\). Such a buyer is never allocated the object, so their payment is also zero. In particular,
\begin{equation}
  t_i^{\Gamma^{d,n},S^d}(\theta^n)=0\quad\text{for every }i\text{ whenever the object is retained.} \label{eq:no-sale-zero}
\end{equation}

\paragraph{Step~2: Identity compatibility bounds revenue at each value profile.}

Write
\[
  \mathcal R^{\Gamma^{d,n},S^d}(\theta^n)
  :=
  \sum_{i=1}^n t_i^{\Gamma^{d,n},S^d}(\theta^n)
\]
for the seller's revenue at value profile \(\theta^n\). Because the seller can always add no shills, the supremum over deviations is at least \(\mathcal R^{\Gamma^{d,n},S^d}(\theta^n)\) for every \(\theta^n\). Hence, the difference between this supremum and the original revenue is pointwise nonnegative. Ex-post seller identity compatibility requires its expectation to be nonpositive, so this difference must equal zero almost surely. Therefore, for almost every \(\theta^n\),

\begin{equation}
  \sum_{i=1}^n
  t_i^{\Gamma^{d,n+k},S^d}(\theta^n,\hat\theta^k)
  \le
  \mathcal R^{\Gamma^{d,n},S^d}(\theta^n)
  \quad
  \text{for every finite \(k\) and every \(\hat\theta^k\)}.
  \label{eq:dark-seller-ae}
\end{equation}
Here, each shill follows the equilibrium strategy \(S^d(\hat\theta_j)\). Thus, adding any finite number of shills cannot increase the total payment collected from the buyers.

Now fix a selling profile \(\theta^n\) with a unique highest-value buyer \(i\), where \(\theta_i>\rho^*\).\footnote{We choose \(\theta^n\) outside the \(F^n\)-null set in \eqref{eq:dark-seller-ae} and so that the optimal allocation rule \eqref{eq:dark-regular-allocation} continues to hold after appending almost every value profile of additional opponents, for every finite number of such opponents. Fubini's theorem permits this choice: for each \(k\geq0\), outside an \(F^n\)-null set of the value profiles \(\theta^n\), the optimal-allocation property continues to hold for \(F^k\)-almost every appended opponent value profile. Taking the countable union of these exceptional sets over \(k\) still yields an \(F^n\)-null set.} Buyer \(i\) can reproduce the entire profile \(\theta^n\) using \(n\) identities, with identity \(j\) following \(S^d(\theta_j)\). For each \(n\), this construction defines a feasible deviation for buyer \(i\), using \(n-1\) shills and no shill when \(n=1\). Because \(\theta_i\) is the highest value among these identities, buyer \(i\) wins whenever every other buyer has value below \(\theta_i\), an event of probability \(G(\theta_i)\).

Conditional on winning, \eqref{eq:dark-seller-ae} bounds the total payment across the identities controlled by buyer \(i\) by \(\mathcal R^{\Gamma^{d,n},S^d}(\theta^n)\). The roles of the two value profiles are now reversed relative to the seller's deviation: in \eqref{eq:dark-seller-ae}, \(\theta^n\) is the buyers' value profile and \(\hat\theta^k\) is chosen by the seller through shill bidding, whereas here buyer \(i\) reproduces \(\theta^n\) through the identities they control and the appended profile \(\hat\theta^k\) is generated by the other buyers. If instead some other buyer has value above \(\theta_i\), all controlled identities lose, and ex-post IR implies that their payments are nonpositive. Bayesian buyer identity compatibility therefore implies\footnote{Recall that \(p_{k+1}\) is buyer \(i\)'s posterior probability that the total number of buyers is \(k+1\), or equivalently, that buyer \(i\) faces \(k\) other buyers.}
\begin{align*}
  &U_i^{\Gamma^d,S^d}(\theta_i)\\
  =& \sum_{k=0}^{\infty} p_{k+1}\E_{\tilde{\hat\theta}^{k}}\!\left[\theta_i q_i^{\Gamma^{d,1+k},S^d}(\theta_i,\tilde{\hat\theta}^{k}) - t_i^{\Gamma^{d,1+k},S^d}(\theta_i,\tilde{\hat\theta}^{k})\right]\\
  \geq& \sum_{k=0}^{\infty} p_{k+1}\E_{\tilde{\hat\theta}^{k}}\!\left[\theta_i \sum_{j=1}^n q_j^{\Gamma^{d,n + k},S^d}(\theta^n,\tilde{\hat\theta}^{k}) - \sum_{j=1}^n t_j^{\Gamma^{d,n+k},S^d}(\theta^n,\tilde{\hat\theta}^{k})\right]\\
  \geq&\sum_{k=0}^{\infty} p_{k+1}F(\theta_i)^k\E_{\tilde{\hat\theta}^{k}}\!\left[\theta_i \sum_{j=1}^n q_j^{\Gamma^{d,n + k},S^d}(\theta^n,\tilde{\hat\theta}^{k}) - \sum_{j=1}^n t_j^{\Gamma^{d,n+k},S^d}(\theta^n,\tilde{\hat\theta}^{k}) \cond \theta_i > \max_{j} \tilde{\hat \theta}_j\right]\\
  & + \sum_{k=0}^{\infty} p_{k+1}(1-F(\theta_i)^k)\E_{\tilde{\hat\theta}^{k}}\!\left[\theta_i \sum_{j=1}^n q_j^{\Gamma^{d,n + k},S^d}(\theta^n,\tilde{\hat\theta}^{k}) - \sum_{j=1}^n t_j^{\Gamma^{d,n+k},S^d}(\theta^n,\tilde{\hat\theta}^{k})\cond \theta_i \le \max_{j} \tilde{\hat \theta}_j\right]\\
  \geq&\sum_{k=0}^{\infty} p_{k+1}F(\theta_i)^k\E_{\tilde{\hat\theta}^{k}}\!\left[\theta_i  - \mathcal R^{\Gamma^{d,n},S^d}(\theta^n) \cond \theta_i > \max_{j} \tilde{\hat \theta}_j\right] \\
  \ge& \ G(\theta_i)\left[\theta_i-\mathcal R^{\Gamma^{d,n},S^d}(\theta^n)\right].
\end{align*}
Using \eqref{eq:dark-regular-interim}, we obtain
\begin{equation}
  \mathcal R^{\Gamma^{d,n},S^d}(\theta^n)
  \geq
  S^d_{\text{FP}}(\theta_i)
  \quad
  \text{whenever \(i\) is the winner.}
  \label{eq:regular-dark-revenue-bound}
\end{equation}

\paragraph{Step~3: Revenue equivalence identifies the payment rule.}
At a selling profile with winner \(i\), ex-post IR allows each losing buyer \(j\ne i\) to receive a subsidy \(s_j:=-t_j^{\Gamma^{d,n},S^d}(\theta^n)\ge0\). To establish uniqueness of the payment rule, consider the following value-weighted difference between the actual payments and the dark first-price payments. Since the winner has a higher value than every loser, \eqref{eq:regular-dark-revenue-bound} implies
\begin{align}
  D_n(\theta^n):= & \sum_{j=1}^n\theta_j\left[t_j^{\Gamma^{d,n},S^d}(\theta^n)-S^d_{\text{FP}}(\theta_j)q_j^{\Gamma^{d,n},S^d}(\theta^n)\right]\notag\\
  = & \sum_{j=1}^n\theta_i\left[t_j^{\Gamma^{d,n},S^d}(\theta^n)-S^d_{\text{FP}}(\theta_j)q_j^{\Gamma^{d,n},S^d}(\theta^n)\right]+\sum_{j\ne i}(\theta_j-\theta_i)\left[t_j^{\Gamma^{d,n},S^d}(\theta^n)\right]\notag\\
  = & \theta_i\left[\sum_{j=1}^n t_j^{\Gamma^{d,n},S^d}(\theta^n)-\sum_{j=1}^n S^d_{\text{FP}}(\theta_j)q_j^{\Gamma^{d,n},S^d}(\theta^n)\right]+\sum_{j\ne i}(\theta_j-\theta_i)\left[t_j^{\Gamma^{d,n},S^d}(\theta^n)\right]\notag\\
  = & \theta_i\left[\mathcal R^{\Gamma^{d,n},S^d}(\theta^n)-S^d_{\text{FP}}(\theta_i)\right]+\sum_{j\ne i}(\theta_i-\theta_j)s_j\ge0. \label{eq:regular-weighted}
\end{align}
At a no-sale profile, \(D_n=0\) by \eqref{eq:no-sale-zero}. On the other hand, the interim expected payment identity in \eqref{eq:dark-regular-interim-payment} gives
\begin{align*}
  \sum_{n \in \N} \gamma_n \E_{\tilde\theta^n}\left[D_n(\tilde\theta^n)\right] &= \sum_{n \in \N} \gamma_n \E_{\tilde\theta^n}\left[\sum_{j=1}^n \tilde\theta_j\left[t_j^{\Gamma^{d,n},S^d}(\tilde\theta^n)-S^d_{\text{FP}}(\tilde\theta_j)q_j^{\Gamma^{d,n},S^d}(\tilde\theta^n)\right]\right]\\
  &= \sum_{n \in \N} n \gamma_n \E_{\tilde\theta^n}\left[\tilde\theta_j\left[t_j^{\Gamma^{d,n},S^d}(\tilde\theta^n)-S^d_{\text{FP}}(\tilde\theta_j)q_j^{\Gamma^{d,n},S^d}(\tilde\theta^n)\right]\right]\\
  &= \sum_{n \in \N} \lambda p_n \E_{\tilde\theta^n}\left[\tilde\theta_j\left[t_j^{\Gamma^{d,n},S^d}(\tilde\theta^n)-S^d_{\text{FP}}(\tilde\theta_j)q_j^{\Gamma^{d,n},S^d}(\tilde\theta^n)\right]\right]\\
  &= \lambda \E_{\tilde\theta_j}\left[\tilde\theta_j \sum_{n \in \N}p_n \E_{\tilde\theta^n_{-j}}\left[t_j^{\Gamma^{d,n},S^d}(\tilde\theta^n)-S^d_{\text{FP}}(\tilde\theta_j)q_j^{\Gamma^{d,n},S^d}(\tilde\theta^n)\right]\right] \\
  &= \lambda \E_{\tilde\theta_j}\left[\tilde\theta_j \left[T_j^{\Gamma^d,S^d}(\tilde\theta_j)-S^d_{\text{FP}}(\tilde\theta_j)Q_j^{\Gamma^d,S^d}(\tilde\theta_j)\right]\right] = 0.
\end{align*}
Every term is nonnegative, and full support implies \(\gamma_n>0\) for every \(n\). Hence, for each \(n\), \(D_n(\theta^n)=0\) almost surely. Equation~\eqref{eq:regular-weighted} therefore implies that, whenever buyer \(i\) wins, \(\mathcal R^{\Gamma^{d,n},S^d}(\theta^n)=S^d_{\text{FP}}(\theta_i)\) and that every losing buyer receives zero subsidy, \(s_j=0\). Together with \eqref{eq:no-sale-zero}, it follows that, for every buyer \(i\), every \(n\), and almost every \(\theta^n\),
\[
  t_i^{\Gamma^{d,n},S^d}(\theta^n)=S^d_{\text{FP}}(\theta_i)q_i^{\Gamma^{d,n},S^d}(\theta^n).
\]
Thus the equilibrium-induced allocation and expected payment conditional on the full value profile coincide almost everywhere with those of the dark first-price auction (with the Myersonian reserve \(\rho^*\)). This proves Theorem~\ref{thm:dark-regular}. \qed

\subsection{Proof of Theorem~\ref{thm:lit-posted-price}}
\label{app:proof-lit-posted-price}
Recall that, for the indirect lit auction, we write
\[
  q_i^{\Gamma^n,S^n}(\theta^n):=q_i^{\Gamma^n}\!\left(S^n(\theta^n)\right),\qquad
  t_i^{\Gamma^n,S^n}(\theta^n):=t_i^{\Gamma^n}\!\left(S^n(\theta^n)\right).
\]
Their interim counterparts are \(Q_i^{\Gamma^n,S^n}\), \(T_i^{\Gamma^n,S^n}\), and \(U_i^{\Gamma^n,S^n}(\theta_i)=\theta_iQ_i^{\Gamma^n,S^n}(\theta_i)-T_i^{\Gamma^n,S^n}(\theta_i)\). We reserve \((q^n,t^n)\) for a realized outcome in \(\mathcal X^n\left(S^n(\theta^n)\right)\). Write
\[
  \mathcal R^{\Gamma^n,S^n}(\theta^n)
  :=\sum_{i=1}^n t_i^{\Gamma^n,S^n}(\theta^n)
\]
for the seller's expected revenue at value profile \(\theta^n\).

\paragraph{Step~1: The seller's constraint holds pointwise almost everywhere.}
The seller can choose to add no shills and select any supported outcome. Hence the supremum in strong ex-post seller identity compatibility is at least \(\mathcal R^{\Gamma^n,S^n}(\theta^n)\) at every value profile. Its difference from this revenue is therefore pointwise nonnegative. Strong ex-post seller identity compatibility requires the expectation of this difference to be nonpositive, so the difference equals zero almost surely. Full support makes this conclusion valid at every population size. Thus, for every \(n\), and for almost every \(\theta^n\),
\begin{equation}
  \sum_{i=1}^n t_i^{n+k}
  \leq\mathcal R^{\Gamma^n,S^n}(\theta^n)
  \label{eq:posted-seller-ae}
\end{equation}
for every finite \(k\geq0\), every shill-bidding strategy profile \(\boldsymbol\sigma_{0k}^{n+k}\), and every outcome \((q^{n+k},t^{n+k})\in\mathcal X^{n+k}\left(S^{n+k}(\theta^n),\boldsymbol\sigma_{0k}^{n+k}\right)\). The exceptional set has measure zero and is independent of the deviation, because the definition takes the supremum over all deviations before integration.

In particular, taking \(k=0\), every outcome in \(\mathcal X^n\left(S^n(\theta^n)\right)\) has total payment at most \(\mathcal R^{\Gamma^n,S^n}(\theta^n)\). Its expected total payment equals this number, so its total payment equals this number almost surely under the outcome lottery. Consequently, for almost every \(\theta^n\),
\[
  \sum_{i=1}^n t_i^n=\mathcal R^{\Gamma^n,S^n}(\theta^n)
  \quad\text{for every }(q^n,t^n)\in\mathcal X^n\left(S^n(\theta^n)\right).
\]
For \(n=1\), this says that buyer \(i\)'s realized payment is always \(t_i^{\Gamma^1,S^1}(\theta_i)\), even if the allocation is randomized.

\paragraph{Step~2: The one-buyer auction is almost a posted price.}
Fix a value \(\theta_i\) outside the exceptional set in Step~1. If \(t_i^{\Gamma^1,S^1}(\theta_i)>0\), strong ex-post IR excludes every supported outcome in which buyer \(i\) does not receive the object. Thus \(q_i^{\Gamma^1,S^1}(\theta_i)=1\) and \(t_i^{\Gamma^1,S^1}(\theta_i)\leq\theta_i\). If two values induce positive payments, both induce allocation with certainty. Applying BIC in both directions therefore makes their payments equal. Consequently, whenever positive payments occur, there is a common price \(\rho>0\) such that, outside the exceptional set,
\begin{equation}
  t_i^{\Gamma^1,S^1}(\theta_i)>0
  \quad\Longrightarrow\quad
  t_i^{\Gamma^1,S^1}(\theta_i)=\rho\leq\theta_i.
  \label{eq:singleton-price}
\end{equation}
This does not yet characterize the auction: nonpositive payments and subsidized allocation lotteries have not been excluded.

\paragraph{Step~3: The seller's shill-bidding strategy bounds every buyer's payment.}
Fix \(\theta_i\) outside the exceptional set for the one-buyer auction. Starting with buyer \(i\) alone, the seller can add \(n-1\) shills, have each shill follow \(S^n(\hat\theta_j)\) for any chosen value, and select any supported outcome. Buyer \(i\) uses \(S^n(\theta_i)\) in the resulting lit auction. Equation~\eqref{eq:posted-seller-ae} therefore gives
\begin{equation}
  t_i^n\leq t_i^{\Gamma^1,S^1}(\theta_i)
  \quad\text{for every }(q^n,t^n)\in\mathcal X^n\left(S^n(\theta_i,\hat\theta_{-i}^n)\right).
  \label{eq:singleton-bound}
\end{equation}
Anonymity permits the buyer to occupy any coordinate \(i\).

If \(t_i^{\Gamma^1,S^1}(\theta_i)\leq0\) almost everywhere, \eqref{eq:singleton-bound} makes every buyer's realized payment nonpositive at almost every profile. Expected revenue is then nonpositive. This cannot maximize revenue: any fixed price strictly between zero and one, with uniform selection among accepting buyers, is feasible and generates positive revenue, as verified below. Thus positive one-buyer payments occur on a set of positive measure and define \(\rho\) in \eqref{eq:singleton-price}.

Consider a profile whose coordinates avoid the one-buyer exceptional set. Strong ex-post IR makes every losing buyer's payment nonpositive. If the winner makes a positive payment, \eqref{eq:singleton-bound} and \eqref{eq:singleton-price} imply that this payment is at most \(\rho\) and that the winner's value is at least \(\rho\). Every supported outcome therefore satisfies
\begin{equation}
  \sum_{i=1}^n t_i^n
  \leq\rho\one\left\{\max_{1\leq i\leq n}\theta_i \geq \rho\right\}.
  \label{eq:price-profile-bound}
\end{equation}

In particular, \(\rho<1\) at a revenue maximizer, since otherwise this bound implies nonpositive expected revenue. Taking expectations gives
\begin{equation}
  \begin{aligned}
    \Rev(\Gamma^\N,S^\N)
    &\leq\rho\sum_{n \in \N}\gamma_n\left[1-F(\rho)^n\right]\\
    &\leq\max_{\hat\rho\in[0,1]}\hat\rho\sum_{n \in \N}\gamma_n\left[1-F(\hat\rho)^n\right].
  \end{aligned}
  \label{eq:price-revenue-bound}
\end{equation}
It is straightforward to verify that a posted-price mechanism that selects uniformly among buyers willing to pay the posted price attains the right-hand side of \eqref{eq:price-revenue-bound} and is strongly ex-post seller identity-compatible and strongly ex-post IR.

\paragraph{Step~4: The revenue bound is essentially achieved by a posted price.}
At a revenue maximizer, both inequalities in \eqref{eq:price-revenue-bound} bind. Hence \(\rho\) maximizes posted-price revenue. The difference between the right-hand side of \eqref{eq:price-profile-bound} and realized revenue is nonnegative and has expectation zero. Full support implies that this difference vanishes at almost every value profile for every population size. Step~1 then makes the bound an equality at every supported outcome of such a profile.

If \(\max_i\theta_i \geq \rho\), revenue equals \(\rho>0\), so there is a winner. The winner pays at most \(\rho\), while every loser pays at most zero. Equality therefore requires the winner to pay exactly \(\rho\) and every loser to pay zero. Strong ex-post IR requires the winner's value to be at least \(\rho\). If \(\max_i\theta_i < \rho\), every payment is nonpositive by \eqref{eq:singleton-bound} and \eqref{eq:singleton-price}; because total revenue is zero, every payment is zero.

It remains to exclude free allocation when all values are below \(\rho\). Fix \(n\). The preceding conclusions imply, for almost every value,
\begin{align*}
  T_i^{\Gamma^n,S^n}(\hat\theta_i)&=0
  &&\text{if }\hat\theta_i < \rho,\\
  T_i^{\Gamma^n,S^n}(\theta_i)&=\rho Q_i^{\Gamma^n,S^n}(\theta_i)
  &&\text{if }\theta_i \geq \rho.
\end{align*}
Fix \(\hat\theta_i < \rho\) satisfying the first equality. For any \(\theta_i \geq \rho\) satisfying the second, BIC against the strategy \(S^n(\hat\theta_i)\) gives
\[
  U_i^{\Gamma^n,S^n}(\theta_i)
  =(\theta_i-\rho)Q_i^{\Gamma^n,S^n}(\theta_i)
  \geq\theta_iQ_i^{\Gamma^n,S^n}(\hat\theta_i).
\]
Let \(\theta_i\downarrow\rho\) through the full-measure set of such values. Since \(Q_i^{\Gamma^n,S^n}(\theta_i)\leq1\) and \(\rho>0\), this implies \(Q_i^{\Gamma^n,S^n}(\hat\theta_i)=0\). Nonnegativity of allocation probabilities implies that a buyer of value below \(\rho\) receives no allocation almost everywhere. At each such profile, zero allocation probability also excludes allocation in any supported outcome. Thus the object is retained whenever all buyers' values are below \(\rho\), completing the characterization of a posted-price mechanism.

\begin{remark}[Rationing among buyers willing to pay the price]
  For \(\theta_i>\rho\), the proof gives \(T_i^{\Gamma^n,S^n}(\theta_i)=\rho Q_i^{\Gamma^n,S^n}(\theta_i)\). Applying BIC in both directions between any two such values makes \(Q_i^{\Gamma^n,S^n}\) constant above the price, almost everywhere. Anonymity and the sale probability identify this constant:
  \[
    Q_i^{\Gamma^n,S^n}(\theta_i)
    =\frac{1-F(\rho)^n}{n\left[1-F(\rho)\right]},
    \qquad\theta_i>\rho.
  \]
  This interim restriction does not require uniform selection among buyers with values above \(\rho\) at each realized value profile. Example~\ref{ex:non-uniform} in the Online Appendix provides a counterexample. The theorem therefore identifies the common price, payments, and sale decision without imposing a particular rationing rule.
\end{remark}

\onlineappendix

\begin{center}
  {\Large\bfseries Online Appendix}\\[2pt]
  \textit{Not for publication}
\end{center}

\section{Additional Proofs}
\label{app:online}

\subsection{Proof of Proposition~\ref{prop:ret}}\label{app:proof-ret}

Buyer \(i\)'s interim expected payoff in the direct dark auction \((\Gamma^d, S^d)\) is given by
\[
  U^{\Gamma^d, S^d}_i\left(\theta_i\right) = \theta_i Q^{\Gamma^d, S^d}_i(\theta_i) - T^{\Gamma^d, S^d}_i(\theta_i).
\]
Following the argument in \citet[Lemma~2]{Myerson-OptimalAuctionDesign-1981}, we can show that BIC is equivalent to the following two conditions: for every buyer \(i\) and every \(\theta_i \in \Theta\),
\begin{enumerate}
  \item Monotonicity: \(Q^{\Gamma^d, S^d}_i(\theta_i)\) is nondecreasing in \(\theta_i\).
  \item Envelope formula:
    \(
      U^{\Gamma^d, S^d}_i(\theta_i) = U^{\Gamma^d, S^d}_i(0) + \int_0^{\theta_i} Q^{\Gamma^d, S^d}_i(z) \dd z.
    \)

    Equivalently, \(T^{\Gamma^d, S^d}_i(\theta_i) = \theta_i Q^{\Gamma^d, S^d}_i(\theta_i) - \int_0^{\theta_i} Q^{\Gamma^d, S^d}_i(z) \dd z - U^{\Gamma^d, S^d}_i(0) \).
\end{enumerate}
The expected revenue is given by
\begin{align*}
  \Rev(\Gamma^d, S^d)
  &= \E_{\tilde{\theta}^{\tilde{n}}}\left[\sum_{i=1}^{\tilde{n}} t_i^{\Gamma^{d,\tilde{n}}, S^d}(\tilde{\theta}^{\tilde{n}})\right]
  =  \sum_{n \in \N} \Pr\left(\tilde{n} = n\right) \E_{\tilde{\theta}^n}\left[\sum_{i=1}^n t_i^{\Gamma^{d,n}, S^d}(\tilde{\theta}^n)\right] \tag{Independence}\\
  &= \sum_{n \in \N} n \gamma_n \E_{\tilde{\theta}_{i}}\left[\E_{\tilde{\theta}^{n}_{-i}}\left[t_i^{\Gamma^{d,n}, S^d}(\tilde{\theta}_{i}, \tilde{\theta}^n_{-i}) \cond \tilde{\theta}_{i}\right]\right] \tag{Anonymity}\\
  &= \sum_{n \in \N} \lambda p_n  \E_{\tilde{\theta}_{i}}\left[\E_{\tilde{\theta}^{n}_{-i}}\left[t_i^{\Gamma^{d,n}, S^d}(\tilde{\theta}_{i}, \tilde{\theta}^n_{-i}) \cond \tilde{\theta}_{i}\right]\right] \tag{Proposition~\ref{prop:common-posterior}}\\
  &= \lambda \E_{\tilde{\theta}_{i}}\left[ \sum_{n \in \N}  p_n \E_{\tilde{\theta}^{n}_{-i}}\left[t_i^{\Gamma^{d,n}, S^d}(\tilde{\theta}_{i}, \tilde{\theta}^n_{-i}) \cond \tilde{\theta}_{i}\right]\right]
  = \lambda \E_{\tilde{\theta}_{i}}\left[ T^{\Gamma^d, S^d}_i(\tilde{\theta}_i)\right] \\
  & = \lambda \E_{\tilde{\theta}_{i}}\left[ \tilde{\theta}_i Q^{\Gamma^d, S^d}_i(\tilde{\theta}_i) - \int_0^{\tilde{\theta}_i} Q^{\Gamma^d, S^d}_i(z) \dd z - U^{\Gamma^d, S^d}_i(0)\right] \tag{Envelope formula}\\
  & = \lambda \E_{\tilde{\theta}_{i}}\left[ \tilde{\theta}_i Q^{\Gamma^d, S^d}_i(\tilde{\theta}_i)\right] - \lambda \int_{0}^{1}\int_0^{\tilde{\theta}_i} Q^{\Gamma^d, S^d}_i(z) f(\tilde{\theta}_i) \dd z \dd \tilde{\theta}_i - \lambda U^{\Gamma^d, S^d}_i(0) \\
  & = \lambda \E_{\tilde{\theta}_{i}}\left[ \tilde{\theta}_i Q^{\Gamma^d, S^d}_i(\tilde{\theta}_i)\right] - \lambda \int_{0}^{1}\int_{z}^{1} Q^{\Gamma^d, S^d}_i(z) f(\tilde{\theta}_i) \dd \tilde{\theta}_i \dd z  - \lambda U^{\Gamma^d, S^d}_i(0) \\
  & = \lambda \E_{\tilde{\theta}_{i}}\left[ \tilde{\theta}_i Q^{\Gamma^d, S^d}_i(\tilde{\theta}_i)\right] - \lambda \int_{0}^{1} \left(1-F(z)\right) Q^{\Gamma^d, S^d}_i(z) \dd z  - \lambda U^{\Gamma^d, S^d}_i(0) \\
  & = \lambda \E_{\tilde{\theta}_{i}}\left[ \left(\tilde{\theta}_i - \frac{1-F(\tilde{\theta}_i)}{f(\tilde{\theta}_i)}\right) Q^{\Gamma^d, S^d}_i(\tilde{\theta}_i)\right] - \lambda U^{\Gamma^d, S^d}_i(0) \\
  & = \lambda \E_{\tilde{\theta}_{i}}\left[ \phi(\tilde{\theta}_i) \sum_{n \in \N} p_n\E_{\tilde{\theta}^{n}_{-i}}\left[q_i^{\Gamma^{d,n}, S^d}(\tilde{\theta}_{i}, \tilde{\theta}^{n}_{-i}) \cond \tilde{\theta}_{i}\right]\right] - \lambda U^{\Gamma^d, S^d}_i(0)\\
  & = \sum_{n \in \N} \lambda p_n \E_{\tilde{\theta}^{n}}\left[ \phi(\tilde{\theta}_i) q_i^{\Gamma^{d,n}, S^d}(\tilde{\theta}^{n}) \right] - \lambda U^{\Gamma^d, S^d}_i(0)\\
  & = \sum_{n \in \N} n \gamma_n \E_{\tilde{\theta}^{n}}\left[ \phi(\tilde{\theta}_i) q_i^{\Gamma^{d,n}, S^d}(\tilde{\theta}^{n}) \right] - \lambda U^{\Gamma^d, S^d}_i(0)\\
  & = \sum_{n \in \N} \gamma_n \E_{\tilde{\theta}^{n}}\left[\sum_{i=1}^{n} \phi(\tilde{\theta}_i) q_i^{\Gamma^{d,n}, S^d}(\tilde{\theta}^{n}) \right] - \lambda U^{\Gamma^d, S^d}_i(0)\\
  & = \E_{\tilde{\theta}^{\tilde{n}}}\left[\sum_{i=1}^{\tilde{n}} \phi(\tilde{\theta}_i) q_i^{\Gamma^{d,\tilde{n}}, S^d}(\tilde{\theta}^{\tilde{n}})\right] - \lambda U^{\Gamma^d, S^d}_i(0).
\end{align*}
\qed

To maximize the expected revenue, it suffices to maximize the first term, which is the expected virtual surplus, subject to the constraints that \(Q^{\Gamma^d, S^d}_i(\theta_i)\) is nondecreasing in \(\theta_i\) and that \(T^{\Gamma^d, S^d}_i(\theta_i)\) satisfies the envelope formula. The second term \(-\lambda U^{\Gamma^d, S^d}_i(0)\) is maximized by setting \(U^{\Gamma^d, S^d}_i(0) = 0\), which is consistent with IR. Therefore, the optimal direct dark auction is characterized by maximizing the expected virtual surplus subject to the monotonicity and envelope conditions.

When the value distribution is regular, i.e., the virtual value is strictly increasing, the expected virtual surplus is maximized by allocating the object to the buyer with the highest positive virtual value. The monotonicity condition is automatically satisfied in this case. Interim expected payments are then determined by the envelope formula.

When the value distribution is irregular, we can apply the ironing procedure described in \citet{Myerson-OptimalAuctionDesign-1981}. For completeness, we briefly summarize the ironing procedure.

For \(q\in[0,1]\), let
\begin{align*}
  h(q)
  &:=\phi\!\left(F^{-1}(q)\right),\\
  H(q)
  &:=\int_0^q h(s) \dd s
  =-(1-q)F^{-1}(q).
\end{align*}
Let \(G\) be the greatest convex function on \([0,1]\) satisfying
\(G(q)\leq H(q)\) for every \(q\in[0,1]\). Let \(g\) be a nondecreasing
version of the derivative of \(G\), for example,
\[
  g(q):=
  \begin{cases}
    G'_+(q), & q\in[0,1),\\
    G'_-(1), & q=1,
  \end{cases}
\]
where \(G'_{+}\) is the right derivative of \(G\) and \(G'_{-}\) is the left derivative of \(G\), and define the \textbf{ironed virtual value} by
\[
  \bar{\phi}(\theta):=g(F(\theta)).
\]
The function \(\bar{\phi}\) is nondecreasing. A maximal open interval \((a,b)\subseteq(0,1)\) on which \(G(q)<H(q)\) is called an ironing interval; \(\bar{\phi}\) is constant on \(F^{-1}((a,b))\).

Now we can characterize an optimal allocation rule directly in terms of ironed virtual values: the object is allocated to a buyer with the highest positive ironed virtual value, with ties broken uniformly at random. Because the ironed virtual value is nondecreasing even when the value distribution is irregular, the resulting allocation rule satisfies the monotonicity condition. Interim expected payments are then determined by the envelope formula, which ensures BIC. Revenue maximization makes the IR constraint bind at the lowest type, so \(U^{\Gamma^d,S^d}_i(0)=0\) for every buyer \(i\). The following proposition summarizes the resulting optimal direct dark auction under ironing.

\begin{proposition}[Ironing]\label{prop:optimal-direct-dark-auction-ironed}
  Consider a direct dark auction \((\hat \Gamma^d, \hat S^d)\) with the truthful reporting strategy \(\hat S^d(\theta)=\theta\). Define the allocation rule as follows:
  \begin{align*}
    W(\theta^n)
    &=
    \arg \max_{i \in \{1,\dots,n\}}\left\{\bar{\phi}(\theta_i)\right\}, \\
    q_i^{\hat \Gamma^{d,n}, \hat S^d}(\theta^n)
    &=
    \begin{cases}
      \frac{1}{|W(\theta^n)|} & \text{if } i \in W(\theta^n) \text{ and } \bar{\phi}(\theta_i) > 0,\\
      0, & \text{otherwise.}
    \end{cases}
  \end{align*}
  Interim expected payments are determined by the envelope formula with \(U^{\hat \Gamma^d, \hat S^d}_i(0) = 0\) for every buyer \(i\). Then \((\hat \Gamma^d, \hat S^d)\) is optimal. The maximal expected revenue is given by
  \[
    \Rev(\hat \Gamma^d, \hat S^d)
    =
    \E_{\tilde{\theta}^{\tilde{n}}}\left[\sum_{i=1}^{\tilde{n}} \bar{\phi}(\tilde{\theta}_i) q_i^{\hat \Gamma^{d,\tilde{n}}, \hat S^d}(\tilde{\theta}^{\tilde{n}})\right].
  \]
\end{proposition}
\begin{proof}
  Take any BIC direct dark auction \((\Gamma^d, S^d)\) with the truthful reporting strategy \(S^d(\theta)=\theta\). Let \(X_i^{\Gamma^d, S^d}(q_i) = Q_i^{\Gamma^d, S^d}\left(F^{-1}(q_i)\right)\). Notice that
  \begin{align*}
    &\E_{\tilde{\theta}_{i}}\left[ \phi(\tilde{\theta}_i) Q^{\Gamma^d, S^d}_i(\tilde{\theta}_i)\right] \\
    =& \int_{0}^{1} \phi(\tilde{\theta}_i) Q^{\Gamma^d, S^d}_i(\tilde{\theta}_i) f(\tilde{\theta}_i) \dd \tilde{\theta}_i  \\
    =& \int_{0}^{1} \phi(F^{-1}(q_i)) X_i^{\Gamma^d, S^d}(q_i) \dd q_i
    = \int_{0}^{1} h(q_i) X_i^{\Gamma^d, S^d}(q_i) \dd q_i \tag{Change of variables from \(\tilde{\theta}_i\) to \(q_i\)}\\
    \leq& \int_{0}^{1} g(q_i) X_i^{\Gamma^d, S^d}(q_i) \dd q_i \\
    =& \int_{0}^{1} \bar{\phi}(F^{-1}(q_i)) X_i^{\Gamma^d, S^d}(q_i) \dd q_i
    = \E_{\tilde{\theta}_{i}}\left[ \bar{\phi}(\tilde{\theta}_i)   Q^{\Gamma^d, S^d}_i(\tilde{\theta}_i)\right].
  \end{align*}
  The inequality follows from
  \begin{align*}
    \int_0^1
    \left(h(q_i)-g(q_i)\right)X_i^{\Gamma^d, S^d}(q_i)\,\dd q_i
    &=
    \left[
      \left(H(q_i)-G(q_i)\right)X_i^{\Gamma^d, S^d}(q_i)
    \right]_{q_i=0}^{q_i=1} \\
    &\qquad
    -
    \int_0^1
    \left(H(q_i)-G(q_i)\right) \dd X_i^{\Gamma^d, S^d}(q_i).
  \end{align*}
  Because \(G\) is the greatest convex minorant of \(H\),
  \[
    H(q_i)-G(q_i)\geq 0,
    \qquad
    G(0)=H(0),
    \qquad
    G(1)=H(1).
  \]
  Therefore, we have shown that
  \begin{align*}
    \E_{\tilde{\theta}_{i}}\left[ \phi(\tilde{\theta}_i) Q^{\Gamma^d, S^d}_i(\tilde{\theta}_i)\right]
    &=
    \E_{\tilde{\theta}_{i}}\left[ \bar{\phi}(\tilde{\theta}_i)   Q^{\Gamma^d, S^d}_i(\tilde{\theta}_i)\right] -
    \int_0^1
    \left(H(q_i)-G(q_i)\right) \dd X_i^{\Gamma^d, S^d}(q_i) \\
    &
    \leq \E_{\tilde{\theta}_{i}}\left[ \bar{\phi}(\tilde{\theta}_i)   Q^{\Gamma^d, S^d}_i(\tilde{\theta}_i)\right].
  \end{align*}

  We observe that the upper bound is achieved by \((\hat \Gamma^d, \hat S^d)\). This is because whenever \(H(q_i)>G(q_i)\), we have \(g'(q_i) = G''(q_i) = 0\). Since \(\bar \phi(F^{-1}(q_i)) = g(q_i)\), the ironed virtual value is constant in a neighborhood of \(q_i\) in quantile space. Hence, the allocation rule \(X_i^{\hat \Gamma^d, \hat S^d}(q_i)\) is also constant in that neighborhood. Therefore, we have
  \begin{align*}
    \int_0^1
    \left(H(q_i)-G(q_i)\right) \dd X_i^{\hat \Gamma^d, \hat S^d}(q_i) & = 0 \\
    \E_{\tilde{\theta}_{i}}\left[ \phi(\tilde{\theta}_i) Q^{\hat \Gamma^d,\hat S^d}_i(\tilde{\theta}_i)\right] &= \E_{\tilde{\theta}_{i}}\left[ \bar{\phi}(\tilde{\theta}_i)   Q^{\hat \Gamma^d,\hat S^d}_i(\tilde{\theta}_i)\right].
  \end{align*}
  Substituting this into the expected revenue expression yields
  \[
    \Rev(\hat \Gamma^d, \hat S^d)
    =
    \E_{\tilde{\theta}^{\tilde{n}}}\left[\sum_{i=1}^{\tilde{n}} \bar{\phi}(\tilde{\theta}_i) q_i^{\hat \Gamma^{d,\tilde{n}}, \hat S^d}(\tilde{\theta}^{\tilde{n}})\right] - \lambda U^{\hat \Gamma^d, \hat S^d}_i(0).
  \]
  Revenue maximization is achieved by setting \(U^{\hat \Gamma^d, \hat S^d}_i(0) = 0\) for all buyers \(i\) to satisfy the individual rationality constraint. This completes the proof.
\end{proof}

\subsection{Proof of Observation~\ref{obs:ex-post-buyer}}
\label{app:proof-ex-post-buyer}
\begin{definition}[Ex-Post Buyer Identity Compatibility]\label{def:ex-post-buyer-identity-compatibility}
  \leavevmode
  \begin{itemize}[leftmargin=*]
    \item A lit auction \((\Gamma^{\N},S^{\N})\) is \textbf{ex-post buyer identity-compatible} if, for every \(n\in\N\), every buyer \(i\), every \(\theta^n\in\Theta^n\), every \(k\in\N_{\geq0}\), and every \(\boldsymbol\sigma_{ik}^{n+k}\in(\Sigma^{n+k})^{k+1}\),
      \begin{align*}
        &\theta_i q_i^{\Gamma^n,S^n}(\theta^n)-t_i^{\Gamma^n,S^n}(\theta^n)\\
        &\quad\geq
        \sum_{j\in\{i,n+1,\ldots,n+k\}}
        \Bigl[
          \theta_i q_j^{\Gamma^{n+k}}\!\left(\boldsymbol\sigma_{ik}^{n+k},S_{-i}^{n+k}(\theta^n_{-i})\right)
          -t_j^{\Gamma^{n+k}}\!\left(\boldsymbol\sigma_{ik}^{n+k},S_{-i}^{n+k}(\theta^n_{-i})\right)
        \Bigr].
      \end{align*}
    \item A dark auction \((\Gamma^d,S^d)\) is \textbf{ex-post buyer identity-compatible} if, for every \(n\in\N\), every buyer \(i\), every \(\theta^n\in\Theta^n\), every \(k\in\N_{\geq0}\), and every \(\boldsymbol\sigma_{ik}^{d}\in(\Sigma^d)^{k+1}\),
      \begin{align*}
        &\theta_i q_i^{\Gamma^{d,n},S^d}(\theta^n)-t_i^{\Gamma^{d,n},S^d}(\theta^n)\\
        &\quad\geq
        \sum_{j\in\{i,n+1,\ldots,n+k\}}
        \Bigl[
          \theta_i q_j^{\Gamma^{d,n+k}}\!\left(\boldsymbol\sigma_{ik}^{d},S_{-i}^{d}(\theta^n_{-i})\right)
          -t_j^{\Gamma^{d,n+k}}\!\left(\boldsymbol\sigma_{ik}^{d},S_{-i}^{d}(\theta^n_{-i})\right)
        \Bigr].
      \end{align*}
  \end{itemize}
\end{definition}

When \(k=0\), the set of controlled identities is \(\{i\}\). An auction is \textbf{ex-post incentive-compatible (ex-post IC)} if the corresponding inequality in Definition~\ref{def:ex-post-buyer-identity-compatibility} holds for \(k=0\): at every realized population size and value profile, no buyer can profit by deviating without shill bidding. Ex-post buyer identity compatibility therefore implies ex-post IC.

To prove Observation~\ref{obs:ex-post-buyer}, it suffices to show that, under the maintained assumptions, any optimal lit or dark auction that is ex-post IC and ex-post IR is essentially the second-price auction with the Myersonian reserve.

For each \(n\in\N\), let \(R_n\) denote the auction's expected revenue conditional on there being \(n\) buyers, and let \(R_n^*\) denote the revenue of the second-price auction with the Myersonian reserve in the corresponding \(n\)-buyer environment. Ex-post IC and ex-post IR imply that, conditional on each population size, the induced direct mechanism is BIC and IR. Hence, \(R_n\leq R_n^*\) by \citet{Myerson-OptimalAuctionDesign-1981}. The second-price auction with the Myersonian reserve is feasible in both the lit and dark environments and attains expected revenue \(\sum_{n\in\N}\gamma_nR_n^*\). Optimality (Proposition~\ref{prop:ret}) and full support therefore imply \(R_n=R_n^*\) for every \(n\in\N\).

Fix \(n\). Under regularity, optimality implies that, almost everywhere, the object is allocated to the highest-value buyer whenever their value exceeds the Myersonian reserve and is retained otherwise. Ex-post IC then determines buyer \(i\)'s payment, conditional on \(\theta_{-i}^n\), up to their utility at value zero: the payment equals the corresponding second-price payment minus \(U_i(0,\theta_{-i}^n)\). Ex-post IR implies \(U_i(0,\theta_{-i}^n)\geq0\). Since \(R_n=R_n^*\), the expectation of these nonnegative utility terms must be zero. Hence, \(U_i(0,\theta_{-i}^n)=0\) for almost every \(\theta_{-i}^n\) and every buyer \(i\). The induced expected payments conditional on the full value profile therefore coincide almost everywhere with those of the second-price auction with the Myersonian reserve. Since this holds for every \(n\), the auction is essentially a second-price auction.

\subsection{Proof of Theorem~\ref{thm:non-identity-compatible}}
\label{app:proof-non-identity-compatible}

The equilibrium strategies in the original auction may prescribe the use of multiple identities by buyers, the seller, or both, depending on which side is permitted to use multiple identities. Equilibrium requires each player to optimize jointly over both the finite number of identities they create and the strategies played through those identities. Identity choices are made before the total number of identities is disclosed. In a lit auction, the strategies played through these identities may subsequently depend on the realized total number of identities, whereas in a dark auction they cannot.

For expositional simplicity, we focus on pure-strategy equilibria and discuss the extension to mixed strategies later. Let \(n^*(\theta_i)\geq1\) denote the finite number of identities chosen in equilibrium by a buyer of value \(\theta_i\), and let \(n_0^*\geq0\) denote the number chosen by the seller. In a lit auction, when the total number of identities is \(m\), write their joint continuation strategies as
\[
  \boldsymbol S^m(\theta_i)\in(\Sigma^m)^{n^*(\theta_i)},
  \qquad
  \boldsymbol S_0^m\in(\Sigma^m)^{n_0^*}.
\]
In a dark auction, write
\[
  \boldsymbol S^d(\theta_i)\in(\Sigma^d)^{n^*(\theta_i)},
  \qquad
  \boldsymbol S_0^d\in(\Sigma^d)^{n_0^*}.
\]

For buyers, the strategy spaces in lit and dark auctions are, respectively,
\begin{align*}
  \mathcal D_i^{\mathrm{lit}}
  &:=\bigcup_{n_i\geq1}\left(\{n_i\}\times\prod_{m\geq n_i}(\Sigma^m)^{n_i}\right),\\
  \mathcal D_i^{\mathrm{dark}}
  &:=\bigcup_{n_i\geq1}\left(\{n_i\}\times(\Sigma^d)^{n_i}\right).
\end{align*}
Thus, a deviation in a lit auction is \(D_i^{\mathrm{lit}}=(n_i,(\boldsymbol\sigma_i^m)_{m\geq n_i})\), whereas a deviation in a dark auction is \(D_i^{\mathrm{dark}}=(n_i,\boldsymbol\sigma_i^d)\). The number of identities \(n_i\) chosen under the deviation may be smaller or larger than the equilibrium number prescribed for the buyer's value. Moreover, the buyer may choose any joint strategy across the identities they control, not only equilibrium strategies prescribed for other values.\footnote{These displays use the product notation introduced in the main text. In a dynamic auction in which a player can pool information across controlled identities, each product should be replaced by the full admissible joint-strategy space; see Online Appendix~\ref{app:extensive-form-games} for details.}

Consider a unilateral deviation by buyer \(i\) using \(n_i\) identities. If there are \(n\) buyers with value profile \((\theta_i,\theta_{-i}^n)\), the total number of identities in the auction is
\[
  m_i^{n_i}(\theta_{-i}^n):=n_i+n_0^*+\sum_{j\ne i}n^*(\theta_j).
\]
Let \(J_i^{n_i}\) denote the set of identities controlled by buyer \(i\). Let \(\boldsymbol S_{-i}^m(\theta_{-i}^n)\) denote the equilibrium continuation strategies of the other buyers when the total number of identities is \(m\), and define
\begin{align*}
  \boldsymbol s_i^{\mathrm{lit},n_i}(\theta_{-i}^n)
  &:=
  \left(
    \boldsymbol\sigma_i^{m_i^{n_i}(\theta_{-i}^n)},
    \boldsymbol S_{-i}^{m_i^{n_i}(\theta_{-i}^n)}(\theta_{-i}^n),
    \boldsymbol S_0^{m_i^{n_i}(\theta_{-i}^n)}
  \right),\\
  \boldsymbol s_i^{\mathrm{dark},n_i}(\theta_{-i}^n)
  &:=
  \left(
    \boldsymbol\sigma_i^d,
    \boldsymbol S_{-i}^d(\theta_{-i}^n),
    \boldsymbol S_0^d
  \right).
\end{align*}
The corresponding deviation payoffs are
\begin{align*}
  \mathcal U_i^{\mathrm{lit}}(\theta_i;D_i)
  &:=
  \sum_{n \in \N}p_n\E_{\tilde\theta_{-i}^n}\left[
    \sum_{j\in J_i^{n_i}}
    \left(
      \theta_iq_j^{\Gamma^{m_i^{n_i}}}(\boldsymbol s_i^{\mathrm{lit},n_i})
      -t_j^{\Gamma^{m_i^{n_i}}}(\boldsymbol s_i^{\mathrm{lit},n_i})
    \right)
  \right],\\
  \mathcal U_i^{\mathrm{dark}}(\theta_i;D_i)
  &:=
  \sum_{n \in \N}p_n\E_{\tilde\theta_{-i}^n}\left[
    \sum_{j\in J_i^{n_i}}
    \left(
      \theta_iq_j^{\Gamma^{d,m_i^{n_i}}}(\boldsymbol s_i^{\mathrm{dark},n_i})
      -t_j^{\Gamma^{d,m_i^{n_i}}}(\boldsymbol s_i^{\mathrm{dark},n_i})
    \right)
  \right].
\end{align*}
Let
\[
  D_i^{*,\mathrm{lit}}(\theta_i)
  :=
  \left(n^*(\theta_i),(\boldsymbol S^m(\theta_i))_m\right),
  \qquad
  D_i^{*,\mathrm{dark}}(\theta_i)
  :=
  \left(n^*(\theta_i),\boldsymbol S^d(\theta_i)\right)
\]
denote the equilibrium strategies in the lit and dark auctions, respectively. Buyer \(i\)'s Bayesian incentive compatibility constraints under multiple identities are
\begin{equation}
  \mathcal U_i^\alpha(\theta_i;D_i^{*,\alpha}(\theta_i))
  \geq
  \mathcal U_i^\alpha(\theta_i;D_i)
  \quad
  \text{for every \(\theta_i\) and every \(D_i\in\mathcal D_i^\alpha\),}
  \qquad
  \alpha\in\{\mathrm{lit},\mathrm{dark}\}.
  \label{eq:full-source-multiple-ic}
\end{equation}

For the seller, define \(\mathcal D_0^\alpha\) analogously, allowing \(n_0\geq0\). Let \(\mathcal R^\alpha(D_0)\) denote the expected payments from identities controlled by buyers when the seller plays \(D_0\) and buyers follow their equilibrium strategies. With
\[
  D_0^{*,\mathrm{lit}}
  :=
  \left(n_0^*,(\boldsymbol S_0^m)_m\right),
  \qquad
  D_0^{*,\mathrm{dark}}
  :=
  \left(n_0^*,\boldsymbol S_0^d\right),
\]
Bayesian incentive compatibility for the seller under multiple identities requires
\begin{equation}
  \mathcal R^\alpha(D_0^{*,\alpha})
  \geq
  \mathcal R^\alpha(D_0)
  \quad
  \text{for every \(D_0\in\mathcal D_0^\alpha\),}
  \qquad
  \alpha\in\{\mathrm{lit},\mathrm{dark}\}.
  \label{eq:full-source-seller-ic}
\end{equation}

We now construct an outcome-equivalent direct dark auction and verify that it is Bayesian incentive-compatible and Bayesian identity-compatible.

\paragraph{Step~1: Simulate the equilibrium strategies in the original auction.}
Construct a direct dark auction \((\overline \Gamma^d,\overline S^d)\) with message space \(\Theta\) and truthful reporting strategy \(\overline S^d(\theta_i):=\theta_i\). For any report profile \(\hat\theta^n\), let
\[
  m(\hat\theta^n):=n_0^*+\sum_{i=1}^n n^*(\hat\theta_i).
\]
If the original auction is lit, simulate the original game \(\Gamma^{m(\hat\theta^n)}\) using the strategy profile
\[
  \boldsymbol s(\hat\theta^n)
  :=
  \left(
    \boldsymbol S_0^{m(\hat\theta^n)},
    \boldsymbol S^{m(\hat\theta^n)}(\hat\theta_1),\ldots,
    \boldsymbol S^{m(\hat\theta^n)}(\hat\theta_n)
  \right).
\]
Let \(J_i(\hat\theta^n)\) denote the set of identities assigned to report \(\hat\theta_i\), and let \(J_0(\hat\theta^n)\) denote the identities controlled by the seller. For each realized outcome \((q^m,t^m)\in X^{m(\hat\theta^n)}\), define its aggregate outcome \(A_{\hat\theta^n}(q^m,t^m)\in X^n\) by
\[
  A_{\hat\theta^n}(q^m,t^m)
  :=
  \left(
    \left(\sum_{j\in J_i(\hat\theta^n)}q_j^m\right)_{i=1}^n,
    \left(\sum_{j\in J_i(\hat\theta^n)}t_j^m\right)_{i=1}^n
  \right).
\]
Thus, an allocation to an identity controlled by the seller is recorded as no sale, while payments made by seller-controlled identities are omitted because they are internal transfers.

Let \(\mu_{\hat\theta^n}\in\Delta(X^{m(\hat\theta^n)})\) denote the outcome lottery generated by the simulated strategy profile \(\boldsymbol s(\hat\theta^n)\) in the original auction. Define the outcome lottery of the constructed direct auction by the pushforward
\[
  \varphi^{\overline\Gamma^{d,n}}(\hat\theta^n)
  :=
  (A_{\hat\theta^n})_{\#}\mu_{\hat\theta^n}.
\]
If the original auction is dark, apply the same construction using \(\Gamma^{d,m(\hat\theta^n)}\), \(\boldsymbol S^d\), and \(\boldsymbol S_0^d\). If the original auction is lit, the equilibrium strategies' dependence on the total number of identities is simulated internally after reports are submitted. The constructed auction itself does not disclose the number of reports and is therefore dark.

Under truthful reporting, with one report from each buyer and none from the seller, the constructed auction induces, at every value profile, exactly the same outcome as the equilibrium of the original auction. Hence, the two auctions are outcome-equivalent.

\paragraph{Step~2: Map every buyer deviation in the direct dark auction to a corresponding buyer deviation in the original auction.}
Fix \(k\geq0\) and an arbitrary report vector \(\hat{\boldsymbol\theta}_{ik}=(\hat\theta_{i0},\ldots,\hat\theta_{ik})\in\Theta^{k+1}\) submitted through the \(k+1\) identities controlled by buyer \(i\). Neither \(k\) nor this report vector depends on the unobserved total number of reports. If the original auction is lit, buyer \(i\) can implement the corresponding deviation in the original auction by choosing
\begin{equation}
  \begin{aligned}
    n_i&:=\sum_{\ell=0}^k n^*(\hat\theta_{i\ell}),\\
    \boldsymbol\sigma_i^m
    &:=\left(\boldsymbol S^m(\hat\theta_{i0}),\ldots,\boldsymbol S^m(\hat\theta_{ik})\right).
  \end{aligned}
  \label{eq:merge-equilibrium-identities}
\end{equation}
Denote this deviation by \(D_i(\hat{\boldsymbol\theta}_{ik})\). It belongs to \(\mathcal D_i^{\mathrm{lit}}\): the number of identities \(n_i\) is chosen before disclosure, while the subsequent joint strategy depends only on information available in the original auction. If the original auction is dark, replacing \(\boldsymbol S^m\) by \(\boldsymbol S^d\) analogously yields a deviation in \(\mathcal D_i^{\mathrm{dark}}\).

Let \(I_{ik}^n:=\{i,n+1,\ldots,n+k\}\) denote the identities controlled by buyer \(i\) in the constructed direct dark auction. By the simulation in Step~1, the payoff from the report vector \(\hat{\boldsymbol\theta}_{ik}\) in the constructed auction coincides with the payoff from the corresponding deviation \(D_i(\hat{\boldsymbol\theta}_{ik})\) in the original auction:
\begin{align*}
  &\sum_{n \in \N}p_n\E_{\tilde\theta_{-i}^n}\left[
    \sum_{j\in I_{ik}^n}\left(
      \theta_iq_j^{\overline\Gamma^{d,n+k}}(\hat{\boldsymbol\theta}_{ik},\tilde\theta_{-i}^n)
      -t_j^{\overline\Gamma^{d,n+k}}(\hat{\boldsymbol\theta}_{ik},\tilde\theta_{-i}^n)
  \right)\right]\\
  &\qquad=
  \mathcal U_i^\alpha\!\left(\theta_i;D_i(\hat{\boldsymbol\theta}_{ik})\right).
\end{align*}
Applying \eqref{eq:full-source-multiple-ic} therefore gives
\[
  \begin{aligned}
    U_i^{\overline\Gamma^d,\overline S^d}(\theta_i)
    &=\mathcal U_i^\alpha(\theta_i;D_i^{*,\alpha}(\theta_i))\\
    &\geq\mathcal U_i^\alpha\!\left(\theta_i;D_i(\hat{\boldsymbol\theta}_{ik})\right)\\
    &=\sum_{n \in \N}p_n\E_{\tilde\theta_{-i}^n}\left[
      \sum_{j\in I_{ik}^n}\left(
        \theta_iq_j^{\overline\Gamma^{d,n+k}}(\hat{\boldsymbol\theta}_{ik},\tilde\theta_{-i}^n)
        -t_j^{\overline\Gamma^{d,n+k}}(\hat{\boldsymbol\theta}_{ik},\tilde\theta_{-i}^n)
    \right)\right].
  \end{aligned}
\]
Whether buyers may use multiple identities is taken to be the same in the original auction and in the constructed direct dark auction. When buyers are not permitted to use multiple identities, \(n^*(\theta_i)\equiv1\). In this case, the deviation \(D_i(\hat{\boldsymbol\theta}_{ik})\) is admissible only for \(k=0\), and the inequality above reduces to the BIC constraint in the constructed direct dark auction. When buyers are permitted to use multiple identities, \(D_i(\hat{\boldsymbol\theta}_{ik})\) is admissible for every \(k\geq0\), and the inequality above is precisely the Bayesian buyer identity compatibility constraint in the constructed direct dark auction, with BIC corresponding to the special case \(k=0\). Thus, every admissible buyer deviation in the constructed direct dark auction can be replicated by a feasible deviation in the original auction. Incentive compatibility in the original auction therefore implies BIC in the constructed direct dark auction and, when buyers are permitted to use multiple identities, Bayesian buyer identity compatibility as well.

\paragraph{Step~3: Map every seller deviation in the direct dark auction to a corresponding seller deviation in the original auction.}
Fix \(k\geq0\) and an arbitrary report vector \(\hat{\boldsymbol\theta}_{0k}=(\hat\theta_{01},\ldots,\hat\theta_{0k})\in\Theta^k\) submitted through the \(k\) identities controlled by the seller in the constructed direct dark auction. If the original auction is lit, the seller can implement the corresponding deviation in the original auction by choosing
\[
  n_0:=n_0^*+\sum_{\ell=1}^k n^*(\hat\theta_{0\ell}),
  \qquad
  \boldsymbol\sigma_0^m
  :=
  \left(
    \boldsymbol S_0^m,
    \boldsymbol S^m(\hat\theta_{01}),\ldots,
    \boldsymbol S^m(\hat\theta_{0k})
  \right).
\]
Denote this deviation by \(D_0(\hat{\boldsymbol\theta}_{0k})\). If the original auction is dark, replacing the continuation strategies above by their dark counterparts analogously yields a feasible seller deviation. In either case, the seller chooses \(n_0\) before disclosure, as required.

The seller's revenue is the sum of payments made by identities controlled by buyers; transfers among identities controlled by the seller cancel. Hence, by the simulation in Step~1, the revenue generated by \(\hat{\boldsymbol\theta}_{0k}\) in the constructed auction coincides with the revenue from \(D_0(\hat{\boldsymbol\theta}_{0k})\) in the original auction. Applying \eqref{eq:full-source-seller-ic} therefore gives
\[
  \begin{aligned}
    \Rev(\overline\Gamma^d,\overline S^d)
    &=\mathcal R^\alpha(D_0^{*,\alpha})\\
    &\geq\mathcal R^\alpha(D_0(\hat{\boldsymbol\theta}_{0k}))\\
    &=\sum_{n \in \N}\gamma_n\E_{\tilde\theta^n}\left[
      \sum_{i=1}^n
      t_i^{\overline\Gamma^{d,n+k}}(\tilde\theta^n,\hat{\boldsymbol\theta}_{0k})
    \right].
  \end{aligned}
\]
This is precisely the Bayesian seller identity compatibility constraint. Step~3 is unnecessary when the seller is not permitted to use multiple identities. Together with Step~2, this completes the proof for pure-strategy equilibria.

\paragraph{Mixed-strategy equilibria.}
The same argument extends to mixed-strategy equilibria. Interpret each equilibrium strategy in the original auction as a lottery over the pure strategies defined above, including both the number of identities and the continuation strategies played through them. For each report, the constructed auction draws a pure strategy from the corresponding equilibrium lottery and then carries out the simulation in Step~1. The seller's equilibrium strategy is treated analogously.

The deviation arguments in Steps~2 and~3 extend in the same way. For each report submitted through a controlled identity in the constructed auction, a pure strategy is drawn from the corresponding equilibrium strategy in the original auction. The deviator can reproduce this randomization in the original auction, thereby generating a feasible mixed deviation there. Taking expectations over the equilibrium randomization then yields the same incentive inequalities as above.

\paragraph{Why the implementation is dark.}
Outcome equivalence requires preserving allocations and payments, but not the information available to players during the auction. This distinction matters when the original auction is lit. There, buyers choose how many identities to create before the total number of identities is disclosed. A lit implementation of the constructed direct auction with one prescribed identity per buyer would instead reveal the number of buyers before they report their values.

A buyer of a given value in the original auction can imitate the behavior prescribed for another value, including creating the same number of identities prescribed for that value. The restriction is that the buyer must choose the number of identities before learning the realized population size. Suppose, for example, that imitating another value requires one additional identity, which is profitable when the buyer is alone but unprofitable when another buyer is present. The original incentive constraint averages these gains and losses across population sizes. In a lit implementation, however, the buyer could imitate that value only when alone and report truthfully otherwise, obtaining the gain while avoiding the loss. Mapping this deviation back through \eqref{eq:merge-equilibrium-identities} would require the buyer in the original auction to create the additional identity only when alone, which is infeasible because the realized population size is not yet known when identities are chosen. Thus, a lit implementation creates deviations that are unavailable in the original auction, and the original incentive constraints need not imply incentive compatibility and identity compatibility of the constructed auction.

The dark implementation eliminates this problem. Because the number of reports is not disclosed, a buyer must choose their reports without conditioning on the realized population size, just as the corresponding identity choice must be made before disclosure in the original auction. The constructed auction can then simulate internally any subsequent behavior prescribed in the original auction, including responses to the disclosed number of identities. Hence, the above argument yields an outcome-equivalent direct dark auction regardless of whether the original auction is lit or dark. A lit implementation may exist under additional restrictions, but it is not implied by the original incentive constraints.

\subsection{Proof of Theorem~\ref{thm:lit-irregular-ir}}
\label{app:proof-lit-irregular-ir}

We follow the proof strategy from the regular case. The complication is that ironing creates ties with positive probability, so the winner need not be the buyer with the highest value. As a result, the exact surplus and utility comparisons used in the regular case no longer apply. In the irregular case, we replace these identities with first-order estimates.

In what follows, we work in quantile space rather than value space. Write
\[
  g(u):=\bar\phi(F^{-1}(u)),
\]
where \(\bar\phi\) is the ironed virtual value; thus \(g\) is the same function as the one defined in Online Appendix~\ref{app:proof-ret}. Let \(\tilde M_n :=\max_{i \leq n} F(\tilde \theta_i)\) be the highest quantile among \(n\) buyers, so \(\Pr(\tilde M_n\le u)=u^n\). Optimality gives
\begin{equation*}
  \Rev(\Gamma^n,S^n)=\E_{\tilde M_n}[(g(\tilde M_n))_+].
\end{equation*}

For a quantile \(u\), let
\[
  \ell(u):=\inf\{x\in[0,1]:g(x)=g(u)\},
  \qquad
  r(u):=\sup\{x\in[0,1]:g(x)=g(u)\}
\]
denote the lower and upper endpoints of the interval on which the ironed virtual value equals \(g(u)\). Since \(g\) is nondecreasing, any shill with quantile below \(\ell(u)\) has a strictly lower ironed virtual value.

\begin{lemma}[Ironing near the upper endpoint]
  \label{lem:endpoint-ironing}
  Under the maintained assumption that the density \(f\) is continuous and positive, as \(u\uparrow1\),
  \begin{align}
    1-F^{-1}(u)
    &\sim\frac{1-u}{f(1)},\notag\\
    1-g(u)
    &\sim\frac{2(1-u)}{f(1)},\label{eq:psi-endpoint}\\
    \ell(u)
    &=u+o(1-u),
    \qquad
    r(u)=u+o(1-u).
    \label{eq:thin-class}
  \end{align}
\end{lemma}

\begin{proof}
  First, continuity and positivity of \(f\) at \(1\) imply
  \[
    1-F^{-1}(u)
    =
    \frac{1-u}{f(1)}+o(1-u).
  \]
  Therefore,
  \[
    \phi(F^{-1}(u))
    =
    F^{-1}(u)-\frac{1-u}{f(F^{-1}(u))}
    =
    1-\frac{2(1-u)}{f(1)}+o(1-u).
  \]

  Second, we show that any ironing interval near the upper endpoint is negligible relative to its distance from \(1\). Recall that \(H(u):=-(1-u)F^{-1}(u)\), and that the ironed virtual value is the slope of the greatest convex minorant \(G\) of \(H\). At any interior contact point \(x\), where \(G(x) = H(x)\), differentiability of \(H\), convexity of \(G\), and \(G\le H\) give
  \[
    G'_-(x)\ge H'(x)\ge G'_+(x).
  \]
  Convexity also gives \(G'_-(x)\le G'_+(x)\), so all three derivatives coincide. Away from contact points, \(G\) is locally affine. Thus \(G\) is differentiable on \((0,1)\), and its derivative \(g\) is continuous there.

  For \(u\) in a nondegenerate constant-value class of \(g\), continuity of \(g\) implies that its endpoints \(a=\ell(u)\) and \(b=r(u)\) delimit a maximal nondegenerate affine segment of \(G\). As \(u\uparrow1\), both endpoints must approach \(1\). To see this, if the lower endpoints stayed bounded away from \(1\), a subsequence of these segments would contain a common interior point. Maximal affine segments have disjoint interiors, so they would all be the same segment, necessarily ending at \(1\). However, \(b \neq 1\). Otherwise, the chord joining \((a,H(a))\) and \((1,H(1))\) has value \( -(1-u)F^{-1}(a) \) at \(u\in(a,1)\), which is strictly greater than \( H(u)=-(1-u)F^{-1}(u)\), because \(F^{-1}(u)>F^{-1}(a)\). Such a chord therefore cannot belong to a convex minorant of \(H\). Hence, for \(u\) sufficiently close to \(1\), both endpoints are interior contact points. Tangency gives
  \[
    g(u)=h(a)=h(b).
  \]
  Applying the expansion for \(h\) at both endpoints,
  \[
    \frac{2(1-a)}{f(1)}+o(1-a)
    =
    \frac{2(1-b)}{f(1)}+o(1-b),
  \]
  so
  \[
    \frac{1-a}{1-b}\longrightarrow1.
  \]
  Since \(a\le u\le b\), it follows that \(1-a\sim1-u\sim1-b\), and therefore
  \[
    a=u+o(1-u),
    \qquad
    b=u+o(1-u).
  \]
  Furthermore,
  \[
    g(u)=h(a)
    =1-\frac{2(1-a)}{f(1)}+o(1-a)
    =1-\frac{2(1-u)}{f(1)}+o(1-u).
  \]

  If the constant-value class of \(g(u)\) is a singleton, then \(\ell(u)=r(u)=u\). Such a point must be a contact point, since otherwise \(G\) would be locally affine. Hence \(g(u)=h(u)\), and the same expansion follows directly. Combining the two cases proves \eqref{eq:psi-endpoint} and \eqref{eq:thin-class}.
\end{proof}

We will also use the following immediate consequence of optimal allocation. Whenever \(g(u)>0\),
\begin{equation}
  \ell(u)^{n-1} \le Q_i^{\Gamma^n,S^n}(F^{-1}(u)) \le r(u)^{n-1}. \label{eq:allocation-sandwich}
\end{equation}
Buyer \(i\) must win if every other buyer has a strictly lower ironed virtual value, and cannot win if another buyer has a strictly higher one. This observation does not require a particular rule for breaking ties.

\paragraph{Step~1: Add competition without winning the auction beyond a cutoff \(c\).}
Choose a quantile \(c<1\) sufficiently close to one that \(g(c)>0\), and write
\[
  A_0:=\ell(c),
  \qquad
  B_0:=r(c).
\]
By \eqref{eq:thin-class},
\[
  A_0=c+o(1-c),
  \qquad
  B_0=c+o(1-c).
\]
Hence,
\[A_0-B_0^2=c-c^2+o(1-c)=c(1-c)+o(1-c)>0\]
for \(c\) sufficiently close to one. We may therefore choose \(c\) so that
\begin{equation}
  \frac{B_0^2}{A_0}<1.
  \label{eq:AB}
\end{equation}
Fix a number of buyers \(n\). The seller adds exactly \(n\) shills. After observing the buyers' values, the seller draws the shills' quantiles independently according to
\begin{equation*}
  \tilde z_j\sim
  \begin{cases}
    \operatorname{Unif}[0,\ell(\tilde M_n)],&\tilde M_n\ge c,\\
    \operatorname{Unif}[0,1],&\tilde M_n<c.
  \end{cases}
\end{equation*}
Each shill then follows \(S^{2n}(F^{-1}(\tilde z_j))\), pretending to have value \(F^{-1}(\tilde z_j)\). The number of shills is fixed before values are observed, as required by ex-post seller identity compatibility in lit auctions.

On \(\tilde M_n\ge c\), every shill's ironed virtual value is strictly below the highest ironed virtual value among buyers. The object must therefore go to a buyer, whose quantile is at least \(\ell(\tilde M_n)\). The total surplus is consequently at least
\begin{equation*}
  W_n^-:=\E_{\tilde M_n}\left[F^{-1}(\ell(\tilde M_n))\one\{\tilde M_n\ge c\}\right].
\end{equation*}

\paragraph{Step~2: Bound the buyers' total utilities.}

Write \(u^n\) and \(z^n\) for the buyers' and shills' quantile vectors, respectively, and let \(m(u^n):=\max_{1\le j\le n}u_j\). Relative to the uniform product measure on \([0,1]^{2n}\), the joint distribution induced by the deviation has density
\[
  L_n(u^n,z^n):=
  \begin{cases}
    1, & m(u^n)<c,\\
    \displaystyle\frac{\one\{\max_{1\le j\le n} z_j<\ell(m(u^n))\}}{\ell(m(u^n))^n}, & m(u^n)\ge c.
  \end{cases}
\]
Let \(\ell_c(u):=\ell(\max\{c,u\})\). Since \(m(u^n)\ge u_i\), monotonicity of \(\ell\) implies, in both cases,
\[
  L_n(u^n,z^n)\le\frac{1}{\ell_c(u_i)^n}
  \qquad\text{for every buyer }i.
\]

Ex-post IR makes each buyer's utility nonnegative at every quantile (value) profile. We can therefore multiply this density bound by the buyer's utility and integrate under the uniform product measure. As in \eqref{eq:regular-utility-bound}, summing over buyers bounds their total expected utilities under the deviation by
\[
  I_n:=n\int_0^1\frac{U_i^{\Gamma^{2n},S^{2n}}(F^{-1}(u))}{\ell_c(u)^n}\dd u.
\]
Revenue collected from buyers equals the total surplus minus their total utilities. Thus, the deviation yields
\[
  R_{n,n}^{\mathrm{dev}}\ge W_n^- - I_n.
\]

\paragraph{Step~3: Compare the revenue losses when there are many buyers.}
Lemma~\ref{lem:endpoint-ironing} gives the three approximations
\begin{align}
  \Rev(\Gamma^n,S^n) &=1-\frac{2}{nf(1)}+o(1/n),\notag \\
  W_n^-&=1-\frac{1}{nf(1)}+o(1/n),\notag \\
  I_n&=\frac{1}{2nf(1)}+o(1/n).\label{eq:I-asymp}
\end{align}
The first two follow directly from \(\E[1-\tilde M_n]=1/(n+1)\) and \(\Pr(\tilde M_n<c)=c^n\). The calculation of the utility term is provided just below. Combining the three terms,
\begin{align*}
  R_{n,n}^{\mathrm{dev}} &\ge1-\frac{3}{2nf(1)}+o(1/n),\\
  R_{n,n}^{\mathrm{dev}}-\Rev(\Gamma^n,S^n) &\ge\frac{1}{2nf(1)}+o(1/n)>0
\end{align*}
for all sufficiently large \(n\). Full support of the distribution of the number of buyers ensures that these population sizes occur with positive probability. Hence, this contradicts ex-post seller identity compatibility.

\paragraph{Calculation of the utility term.}
Using the envelope formula and changing the order of integration yields
\begin{align}
  I_n
  &=n\int_0^1 Q_i^{\Gamma^{2n},S^{2n}}(F^{-1}(s))(F^{-1})'(s)
  \left[\int_s^1\ell_c(u)^{-n}\dd u\right]\dd s\notag\\
  &=\frac1n\int_0^n Q_i^{\Gamma^{2n},S^{2n}}(F^{-1}(1-y/n))
  (F^{-1})'(1-y/n)J_n(y)\dd y,
  \label{eq:I-rescale}
\end{align}
where \(J_n(y):=\int_0^y\ell_c(1-x/n)^{-n}\dd x\). For fixed \(y>0\), Lemma~\ref{lem:endpoint-ironing} and \eqref{eq:allocation-sandwich} imply
\[
  Q_i^{\Gamma^{2n},S^{2n}}(F^{-1}(1-y/n))\longrightarrow e^{-2y}, \quad J_n(y)\longrightarrow e^y-1, \quad (F^{-1})'(1-y/n)\longrightarrow1/f(1).
\]
Therefore
\[
  nI_n\longrightarrow \frac1{f(1)}\int_0^\infty(e^{-y}-e^{-2y})\dd y =\frac1{2f(1)},
\]
as claimed in \eqref{eq:I-asymp}.

For completeness, this passage to the limit is justified by the following bounds. Fix \(0<\varepsilon<1/4\) and increase \(c\) if needed. For \(0\le t\le1-c\), Lemma~\ref{lem:endpoint-ironing} gives
\[
  \ell(1-t)\ge1-(1+\varepsilon)t, \qquad r(1-t)\le1-(1-\varepsilon)t.
\]
For \(0\le y\le n(1-c)\), these inequalities imply
\[
  Q_i^{\Gamma^{2n},S^{2n}}(F^{-1}(1-y/n)) \le C e^{-2(1-\varepsilon)y}, \qquad J_n(y)\le y e^{(1+2\varepsilon)y}.
\]
Here and below \(C\) is a finite constant independent of \(n\). Since \((F^{-1})'\) is bounded near one, the product in \eqref{eq:I-rescale} is bounded by \(Cy e^{-(1-4\varepsilon)y}\), an integrable function. The remaining contribution to \(I_n\), from \(s<c\) in its first expression, is at most
\[
  nB_0^{2n-1}A_0^{-n}\int_0^c (F^{-1})'(s)\dd s \le Cn(B_0^2/A_0)^n.
\]
It is exponentially small by \eqref{eq:AB}. This completes the proof of Theorem~\ref{thm:lit-irregular-ir}. \qed

\subsection{Proof of Theorem~\ref{thm:lit-irregular-bb}}
\label{app:proof-lit-irregular-bb}

Choose a nonempty interval \(L\) near zero on which the ironed virtual value is strictly negative. Optimality implies that buyers with values in \(L\) are inactive (Definition~\ref{def:inactive-identity-neutrality}) and have zero interim expected payment. Applying inactive-identity neutrality to such an identity and integrating over \(\tilde z\sim F(\cdot\mid L)\) and the other buyers' values yields \eqref{eq:IN-interim}. Lemma~\ref{lem:common-payment} therefore applies and gives
\begin{equation}
  T_i^{\Gamma^2,S^2}(\theta_i)
  \leq
  T_i^{\Gamma^1,S^1}(\theta_i).
  \label{eq:irregular-order}
\end{equation}

To obtain a contradiction, we use the following consequence of the envelope formula.

\begin{lemma}[Payment ordering with common endpoints]
  \label{lem:payment-endpoint}
  Suppose that the lit auction \((\Gamma^\N,S^\N)\) is BIC. Assume that, for each \(m\in\{1,2\}\), there exists \(\ell \in(0,1)\) such that \(U_i^{\Gamma^m,S^m}(\theta_i)=Q_i^{\Gamma^m,S^m}(\theta_i)=0\) for every \(\theta_i\in[0,\ell]\), and \(\lim_{\theta_i\uparrow1}Q_i^{\Gamma^m,S^m}(\theta_i)=1\). If
  \[
    T_i^{\Gamma^2,S^2}(\theta_i)
    \leq
    T_i^{\Gamma^1,S^1}(\theta_i)
  \]
  almost everywhere, then the two interim expected payment rules, and hence the corresponding interim expected allocation rules, coincide almost everywhere.
\end{lemma}

\begin{proof}
  For each \(m\in\{1,2\}\), BIC implies the envelope formula and hence, for almost every \(\theta_i > 0\),
  \[
    \frac{\dd}{\dd\theta_i}
    \left(
      \frac{U_i^{\Gamma^m,S^m}(\theta_i)}{\theta_i}
    \right)
    =
    \frac{T_i^{\Gamma^m,S^m}(\theta_i)}{\theta_i^2}.
  \]
  Because allocation, utility, and payment vanish in a neighborhood of zero, integrating this identity and using
  \[
    \theta_iQ_i^{\Gamma^m,S^m}(\theta_i)
    =
    T_i^{\Gamma^m,S^m}(\theta_i)
    +
    U_i^{\Gamma^m,S^m}(\theta_i)
  \]
  gives, for almost every \(\theta_i>0\),
  \begin{equation}
    Q_i^{\Gamma^m,S^m}(\theta_i)
    =
    \frac{T_i^{\Gamma^m,S^m}(\theta_i)}{\theta_i}
    +
    \int_0^{\theta_i}
    \frac{T_i^{\Gamma^m,S^m}(s)}{s^2}\dd s.
    \label{eq:equal-allocation}
  \end{equation}

  Letting \(\theta_i\uparrow1\) and using
  \[
    \lim_{\theta_i\uparrow1}
    Q_i^{\Gamma^m,S^m}(\theta_i)
    =
    1
  \]
  yields
  \[
    1
    =
    \lim_{\theta_i\uparrow1}
    T_i^{\Gamma^m,S^m}(\theta_i)
    +
    \int_0^1
    \frac{T_i^{\Gamma^m,S^m}(s)}{s^2}\dd s.
  \]
  Subtracting the identity for \(m=2\) from that for \(m=1\), we obtain
  \[
    0
    =
    \lim_{\theta_i\uparrow1}
    \left[
      T_i^{\Gamma^1,S^1}(\theta_i)
      -
      T_i^{\Gamma^2,S^2}(\theta_i)
    \right]
    +
    \int_0^1
    \frac{
      T_i^{\Gamma^1,S^1}(s)
      -
      T_i^{\Gamma^2,S^2}(s)
    }{s^2}\dd s.
  \]
  Both terms on the right-hand side are nonnegative by the assumed payment ordering. Hence both must be zero. In particular,
  \[
    T_i^{\Gamma^1,S^1}(\theta_i)
    =
    T_i^{\Gamma^2,S^2}(\theta_i)
  \]
  almost everywhere. Equation~\eqref{eq:equal-allocation} then implies that the two interim expected allocation rules also coincide almost everywhere.
\end{proof}

The endpoint conditions in Lemma~\ref{lem:payment-endpoint} follow from optimality. Both interim utilities and allocations vanish on \(L\). At the upper endpoint, Lemma~\ref{lem:endpoint-ironing} gives \(g(u)\to1\) and \(\ell(u)\to1\), so the allocation bounds in \eqref{eq:allocation-sandwich} imply
\[
  \lim_{\theta_i\uparrow1}
  Q_i^{\Gamma^1,S^1}(\theta_i)
  =
  \lim_{\theta_i\uparrow1}
  Q_i^{\Gamma^2,S^2}(\theta_i)
  =
  1.
\]
Combining \eqref{eq:irregular-order} with Lemma~\ref{lem:payment-endpoint} therefore gives
\[
  Q_i^{\Gamma^2,S^2}(\theta_i)
  =
  Q_i^{\Gamma^1,S^1}(\theta_i)
\]
almost everywhere.

With one buyer, optimality assigns the object whenever the ironed virtual value is positive. Since \(g(u)\to1\), there exists \(b<1\) such that
\[
  Q_i^{\Gamma^1,S^1}(\theta_i)=1
  \qquad
  \text{for every \(\theta_i>b\)}.
\]
Hence,
\[
  Q_i^{\Gamma^2,S^2}(\theta_i)=1
\]
for almost every \(\theta_i>b\). An interim expected allocation of one implies that such a buyer receives the object with probability one against almost every value of the other buyer. Applying this conclusion to both buyers implies that, at almost every profile in \((b,1)^2\), both must receive the object with probability one, contradicting feasibility. Therefore, \eqref{eq:irregular-order} cannot hold, proving Theorem~\ref{thm:lit-irregular-bb}. \qed

\subsection{Proof of Theorem~\ref{thm:interdependent}}
\label{app:proof-interdependent}

Recall that
\[
  q_i^{\Gamma^{d,n},S^d}(\theta^n):=q_i^{\Gamma^{d,n}}\!\left(S^d(\theta^n)\right),
  \qquad
  t_i^{\Gamma^{d,n},S^d}(\theta^n):=t_i^{\Gamma^{d,n}}\!\left(S^d(\theta^n)\right).
\]
Their interim counterparts are \(Q_i^{\Gamma^d,S^d}\), \(T_i^{\Gamma^d,S^d}\), and
\[
  U_i^{\Gamma^d,S^d}(\theta_i)
  :=\E_{K_{\theta_i}}\!\left[
    v^{\tilde n}(\theta_i,\tilde\theta^{\tilde n}_{-i})
    q_i^{\Gamma^{d,\tilde n},S^d}(\theta_i,\tilde\theta^{\tilde n}_{-i})
  -t_i^{\Gamma^{d,\tilde n},S^d}(\theta_i,\tilde\theta^{\tilde n}_{-i})\right].
\]
The argument first uses the two identity compatibility constraints to identify payments: losers pay zero, and a winner's payment depends only on their own signal. It then shows that sale depends only on whether the highest signal exceeds a common cutoff. Finally, BIC makes the winning payment a strictly increasing bid function, yielding a dark first-price auction.

Every deviation below consists of submitting signal reports through finitely many controlled identities and is therefore available in the direct auction. All statements about outcomes are understood to hold with probability one for each number of buyers. Whenever we fix a buyer's signal or a fictitious signal report, we choose it so that the relevant statements continue to hold with probability one over the other buyers' signals.\footnote{Fubini's theorem guarantees that such choices are possible. For each number of buyers, a property that holds with probability one over the full signal profile also holds, for almost every fixed signal, with probability one over the remaining signals. We exclude the exceptional signals simultaneously across all numbers of buyers; because there are only countably many such numbers, the excluded set still has probability zero. The same argument applies when both a buyer's signal and an additional fictitious report are fixed. Strict positivity of the relevant densities ensures that these probability-one statements also hold under every buyer's posterior over the other buyers' signals.}

\paragraph{Step~1: Bound the seller's revenue.}
Write
\[
  \mathcal R^{\Gamma^{d,n},S^d}(\theta^n)
  :=\sum_{i=1}^n t_i^{\Gamma^{d,n},S^d}(\theta^n)
\]
for the seller's revenue at signal profile \(\theta^n\). Because the seller can always add no shills, the supremum over deviations is at least \(\mathcal R^{\Gamma^{d,n},S^d}(\theta^n)\) for every \(\theta^n\). Hence, the difference between this supremum and the original revenue is pointwise nonnegative. Ex-post seller identity compatibility requires its expectation to be nonpositive, so this difference must equal zero almost surely. Therefore, there is a set \(B_n\) of signal profiles with \(\Pr_{f_n}(\theta^n \in B_n) = 1\) such that, for every \(\theta^n \in B_n\),
\begin{equation}
  \sum_{i=1}^n t_i^{\Gamma^{d,n+k},S^d}(\theta^n,\hat\theta^k)
  \leq \mathcal R^{\Gamma^{d,n},S^d}(\theta^n)
  \quad\text{for every finite \(k\) and every \(\hat\theta^k\)}.
  \label{eq:interdependent-seller}
\end{equation}

\paragraph{Step~2: The auction cannot distribute a net subsidy.}
Suppose some signal profile \(\theta^n\in B_n\) has negative revenue. A buyer can reproduce this profile \(m\) times, using \(mn\) controlled identities. Apply \eqref{eq:interdependent-seller} to each copy, treating the remaining copies and all other buyers as additional identities. Each copy receives a net subsidy of at least \(-\mathcal R^{\Gamma^{d,n},S^d}(\theta^n)>0\), regardless of the other buyers' signals. By increasing \(m\), the buyer obtains arbitrarily large utility, contradicting Bayesian buyer identity compatibility.

Revenue is therefore nonnegative for every \(\theta^n \in B_n\). If no sale occurs, ex-post IR makes every payment nonpositive. Nonnegative revenue then forces every payment to be zero:
\[
  t_i^{\Gamma^{d,n},S^d}(\theta^n)=0
  \quad\text{for every \(i\), whenever no sale occurs and \(\theta^n\in B_n\)}.
\]
Ex-post IR also implies that a losing identity cannot be charged a positive payment.

\paragraph{Step~3: Pin down revenue and interim expected payments.}
Call a signal \(\theta_i\) \emph{active} if \(Q_i^{\Gamma^d,S^d}(\theta_i)>0\). Fix an active signal \(\theta_i>0\) and a selling profile in \(B_n\) whose unique highest signal is \(\theta_i\). Buyer \(i\) with this signal can reproduce the entire selling profile through \(n\) identities. Whenever all the other buyers have lower signals, sale consistency implies that buyer \(i\) still wins. Step~1 bounds buyer \(i\)'s total payment by the revenue at the reproduced profile. If any other buyer has a higher signal, every identity controlled by buyer \(i\) loses and has nonpositive payment by ex-post IR. Bayesian buyer identity compatibility therefore gives
\[
  U_i^{\Gamma^d,S^d}(\theta_i)
  \geq G(\theta_i\mid\theta_i)
  \left[V(\theta_i,\theta_i)-\mathcal R^{\Gamma^{d,n},S^d}(\theta^n)\right].
\]
Define
\[
  \beta(\theta_i):=V(\theta_i,\theta_i)
  -\frac{U_i^{\Gamma^d,S^d}(\theta_i)}{G(\theta_i\mid\theta_i)}
\]
at active signals, and set \(\beta(\theta_i)=0\) elsewhere. The denominator is positive for \(\theta_i>0\). If buyer \(i\) wins whenever their signal is the highest, pays \(\beta(\theta_i)\) when they win, and pays zero otherwise, their expected utility is exactly \(U_i^{\Gamma^d,S^d}(\theta_i)\). The preceding deviation therefore implies that, at any selling profile in which buyer \(i\) has the highest signal, revenue must satisfy
\begin{equation}
  \mathcal R^{\Gamma^{d,n},S^d}(\theta^n)\geq \beta(\theta_i).
  \label{eq:revenue-lower-bound}
\end{equation}

We next show that \(\beta(\theta_i)\leq v^1(\theta_i)\). Fix \(\varepsilon>0\). Continuity of \(V\) allows us to choose \(0<r<\theta_i\) with \(V(\theta_i,r)<v^1(\theta_i)+\varepsilon\). Conditional on all the other buyers having signals below \(r\), some positive-probability set of signal profiles must therefore give buyer \(i\) a value below \(v^1(\theta_i)+\varepsilon\). Choose one such profile with \(m\) buyers and write it as \((\theta_i,z^{m-1})\), where every component of \(z^{m-1}\) is below \(r\). We can also require this profile to belong to \(B_m\): for almost every fixed \(\theta_i\), the bound \eqref{eq:interdependent-seller} holds with conditional probability one over the other buyers' signals. Ex-post IR and conditional efficiency imply
\[
  \mathcal R^{\Gamma^{d,m},S^d}(\theta_i,z^{m-1})
  \leq v^m(\theta_i,z^{m-1})
  <v^1(\theta_i)+\varepsilon.
\]

Buyer \(i\) now reproduces the earlier selling profile and adds the low reports \(z^{m-1}\). As long as the other buyers' signals are below \(\theta_i\), sale consistency ensures that buyer \(i\) continues to win. Apply the bound \eqref{eq:interdependent-seller} to the group \((\theta_i,z^{m-1})\): its total payment is below \(v^1(\theta_i)+\varepsilon\). All other controlled identities lose and have nonpositive payments. Thus, buyer \(i\) can secure the object whenever their signal is the highest while paying, in total, at most \(v^1(\theta_i)+\varepsilon\). Bayesian buyer identity compatibility yields
\[
  U_i^{\Gamma^d,S^d}(\theta_i)
  \geq G(\theta_i\mid\theta_i)
  \left[V(\theta_i,\theta_i)-v^1(\theta_i)-\varepsilon\right].
\]
Letting \(\varepsilon\downarrow0\) proves \(\beta(\theta_i)\leq v^1(\theta_i)\).

The value floor now ensures that receiving the object at price \(\beta(\theta_i)\) is weakly profitable at every signal profile. Conditional efficiency and the definition of \(\beta(\theta_i)\) therefore imply, for an active signal,
\begin{align*}
  &T_i^{\Gamma^d,S^d}(\theta_i)-\beta(\theta_i)Q_i^{\Gamma^d,S^d}(\theta_i)\\
  =& \E_{K_{\theta_i}}\!\left[v^{\tilde n}(\theta_i,\tilde\theta^{\tilde n}_{-i})
    q_i^{\Gamma^{d,\tilde n},S^d}(\theta_i,\tilde\theta^{\tilde n}_{-i})
  \right]-U_i^{\Gamma^d,S^d}(\theta_i)-\beta(\theta_i)Q_i^{\Gamma^d,S^d}(\theta_i)\\
  =&-\E_{K_{\theta_i}}\!\left[
    \left(v^{\tilde n}(\theta_i,\tilde\theta^{\tilde n}_{-i})-\beta(\theta_i)\right)
    \left(\one\{Y<\theta_i\}
    -q_i^{\Gamma^{d,\tilde n},S^d}(\theta_i,\tilde\theta^{\tilde n}_{-i})\right)
  \right] \leq 0.
\end{align*}
At an inactive signal, the same payment inequality follows from ex-post IR.

We now combine the two bounds. By \eqref{eq:revenue-lower-bound}, total revenue is bounded below by the winner's \(\beta(\cdot)\), while each buyer's interim expected payment is bounded above by \(\beta(\cdot)\) times their probability of winning. Using the definitions of \(K_{\theta_i}\) and \(\mu\), we obtain\footnote{The exchange of summation and integration is justified because \(|\beta(\theta_i)Q_i^{\Gamma^d,S^d}(\theta_i)|\leq |T_i^{\Gamma^d,S^d}(\theta_i)|+v^1(\theta_i)\), and the right-hand side is absolutely integrable.}
\begin{align*}
  0&\leq\sum_{n \in \N}\gamma_n\E_{f_n}\!\left[
    \mathcal R^{\Gamma^{d,n},S^d}(\tilde\theta^n)
  -\sum_{i=1}^n \beta(\tilde\theta_i)q_i^{\Gamma^{d,n},S^d}(\tilde\theta^n)\right]\\
  &=\lambda\int_0^1\left[
    T_i^{\Gamma^d,S^d}(\theta_i)
  -\beta(\theta_i)Q_i^{\Gamma^d,S^d}(\theta_i)\right]\mu(d\theta_i)
  \leq0.
\end{align*}
Both bounds must hold with equality. By full support, revenue equals the winner's \(\beta(\cdot)\) with probability one for each \(n \in \N\); the interim expected payment bound also holds with equality for almost every signal:
\begin{equation}
  \begin{aligned}
    \mathcal R^{\Gamma^{d,n},S^d}(\theta^n)
    &=\sum_{i=1}^n \beta(\theta_i)q_i^{\Gamma^{d,n},S^d}(\theta^n),\\
    T_i^{\Gamma^d,S^d}(\theta_i)
    &=\beta(\theta_i)Q_i^{\Gamma^d,S^d}(\theta_i).
  \end{aligned}
  \label{eq:interdependent-binding-payments}
\end{equation}

\paragraph{Step~4: The winner cannot finance subsidies to losers.}
The preceding step identifies total revenue and each buyer's expected payment. A remaining possibility is that the winner pays more than \(\beta(\theta_i)\) and the excess is transferred to losers. We rule this out by weighting each buyer's payment by their signal as in Appendix~\ref{app:proof-dark-regular}. Define
\[
  D_n(\theta^n):=\sum_{j=1}^n\theta_j
  \left[t_j^{\Gamma^{d,n},S^d}(\theta^n)
  -\beta(\theta_j)q_j^{\Gamma^{d,n},S^d}(\theta^n)\right].
\]
At a selling profile with unique winner \(i\), \eqref{eq:interdependent-binding-payments} gives
\[
  D_n(\theta^n)=\theta_i \left(\sum_{j=1}^n
  t_j^{\Gamma^{d,n},S^d}(\theta^n) -  \beta(\theta_i)\right) + \sum_{j\ne i}(\theta_i-\theta_j)
  \left[-t_j^{\Gamma^{d,n},S^d}(\theta^n)\right]\geq0.
\]
Any positive subsidy to a losing buyer \(j\) is financed by buyer \(i\), so it makes this weighted difference positive. At a no-sale profile, \(D_n = 0\). Yet the interim expected payment equality implies
\[
  \sum_{n \in \N}\gamma_n\E_{f_n}[D_n(\tilde\theta^n)]
  =\lambda\int_0^1\theta_i
  \left[T_i^{\Gamma^d,S^d}(\theta_i)
  -\beta(\theta_i)Q_i^{\Gamma^d,S^d}(\theta_i)\right]\mu(d\theta_i)=0.
\]
Thus every subsidy to a losing buyer vanishes with probability one. Losers pay zero, and a winner with signal \(\theta_i\) pays \(\beta(\theta_i)\):
\begin{align*}
  t_i^{\Gamma^{d,n},S^d}(\theta^n)
  &=\beta(\theta_i)q_i^{\Gamma^{d,n},S^d}(\theta^n)
  \quad\text{with probability one for each \(n \in \N\)},\\
  0\leq \beta(\theta_i)&\leq v^1(\theta_i)
  \quad\text{at almost every active signal}.
\end{align*}
The lower payment bound follows from nonnegative revenue as established in Step~2.

\paragraph{Step~5: The sale decision depends only on the highest signal.}
Payments are now identified, but the auction could still condition the decision to sell on lower signals or population size. We rule out this dependence by adding lower reports.

Let buyer \(i\) retain their own report \(\theta_i>0\) and add one identity whose report is drawn uniformly from \((0,\theta_i)\), independently of the other buyers' signals. This deviation does not depend on the unobserved number of buyers. For almost every \(\theta_i\), the payment identities from Step~4 hold almost surely at both the original and augmented profiles, simultaneously across all population sizes.

Existing wins and their payments are unchanged by sale consistency and Step~4. If another buyer has a higher signal, both controlled identities lose and pay zero. The only possible gain is therefore a new win at a profile where buyer \(i\) had the highest signal but the seller retained the object. Original revenue was then zero, so ex-post seller identity compatibility bounds the total payment collected from buyers after the addition by zero. Step~4 makes losing buyers' payments zero and the winner's payment nonnegative. Hence the new win is free. It gives buyer \(i\) a positive gain because their value remains \(v^n(\theta_i,\theta^n_{-i})\geq v^1(\theta_i)>v^1(0)\geq0\).

The deviation therefore never hurts buyer \(i\) and strictly benefits them whenever it creates a new win. Bayesian buyer identity compatibility rules out such wins with positive probability: for every \(n\),
\begin{equation}\label{eq:interdependent-lower-entry}
  q_i^{\Gamma^{d,n+1},S^d}(\theta_i,\theta^n_{-i},z)
  =
  q_i^{\Gamma^{d,n},S^d}(\theta_i,\theta^n_{-i})
\end{equation}
for almost every \((\theta_i,\theta^n_{-i},z)\) with \(0<z<\theta_i\). By Step~4, the buyer's total payment through the two identities is also unchanged.

Define \(e(\theta_i):=q_i^{\Gamma^{d,1},S^d}(\theta_i)\in\{0,1\}\). We claim that, for every \(n\),
\[
  q_i^{\Gamma^{d,n},S^d}(\theta_i,\theta^n_{-i})
  =e(\theta_i)
\]
for almost every profile at which \(\theta_i\) is strictly highest. This follows by induction. The claim holds by definition for \(n=1\). If it holds for \(n\), it also holds for almost every augmented profile obtained by adjoining another lower report. Equation~\eqref{eq:interdependent-lower-entry} then identifies the allocation in the \((n+1)\)-buyer auction with that in the \(n\)-buyer auction, completing the induction.

Conditional efficiency makes buyer \(i\)'s allocation zero whenever another buyer has a higher signal. Therefore, for every \(n\),
\[
  q_i^{\Gamma^{d,n},S^d}(\theta^n)
  =
  e(\theta_i)\one\!\left\{\theta_i>\max_{j\ne i}\theta_j\right\}
\]
almost everywhere. Hence
\(Q_i^{\Gamma^d,S^d}(\theta_i) =e(\theta_i)G(\theta_i\mid\theta_i)\) for almost every \(\theta_i\). Since \(G(\theta_i\mid\theta_i)>0\) for \(\theta_i>0\), a signal is active exactly when \(e(\theta_i)=1\), apart from a null set. An active buyer wins whenever their signal is highest and pays \(\beta(\theta_i)\); an inactive buyer never wins and pays zero.


\paragraph{Step~6: Higher signals are eligible and pay more.}
Let \(E\subseteq\Theta\) denote the set of signals for which the preceding allocation and payment conclusions hold, for every number of buyers, with probability one over the other buyers' signals. The excluded signals have probability zero, as explained at the start of the proof.

First, an inactive signal cannot be higher than an active signal. If \(\theta_i>\hat \theta_i\) and \(\hat \theta_i\) is active, reporting \(\hat \theta_i\) wins with positive probability. Every such win gives positive payoff because \( v^n(\theta_i,\cdot) \geq v^1(\theta_i)>v^1(\hat \theta_i)\geq\beta(\hat \theta_i)\). An inactive buyer's equilibrium payoff is zero, so this deviation contradicts BIC. Thus eligibility has a common cutoff \(c\in[0,1]\): apart from the cutoff itself and the excluded signals, a buyer is active exactly when their signal exceeds \(c\). The allocation is
\[
  q_i^{\Gamma^{d,n},S^d}(\theta^n)
  =\one\{\theta_i>c\}
  \one\!\left\{\theta_i>\max_{j\ne i}\theta_j\right\}.
\]

Reporting an active \(\hat \theta_i\) wins exactly when \(Y<\hat \theta_i\) and costs \(\beta(\hat \theta_i)\). BIC therefore implies that, for every \(\theta_i\in E\) and every \(\hat \theta_i\in E\cap(c,1]\),
\begin{equation}
  e(\theta_i)G(\theta_i\mid\theta_i)
  \left[V(\theta_i,\theta_i)-\beta(\theta_i)\right]
  \geq G(\hat \theta_i\mid\theta_i)\left[V(\theta_i,\hat \theta_i)-\beta(\hat \theta_i)\right].
  \label{eq:interdependent-report-incentive}
\end{equation}

Second, the winning payment is strictly increasing in the winner's signal. Consider two active signals \(\theta_i<\hat\theta_i\). If \(\beta(\hat\theta_i)<\beta(\theta_i)\), a buyer with signal \(\theta_i\) would profitably deviate by reporting \(\hat\theta_i\), which yields a strictly higher probability of winning at a strictly lower payment. Hence, \(\beta(\theta_i)\leq\beta(\hat\theta_i)\). The inequality must be strict. Otherwise, if \(\beta(\theta_i)=\beta(\hat\theta_i)\), a buyer with signal \(\theta_i'\in(\theta_i,\hat\theta_i)\) could profitably report \(\hat\theta_i\): the deviation leaves the winning payment unchanged but strictly increases the probability of winning, with positive utility conditional on winning. Thus, \(\beta(\cdot)\) is strictly increasing on the set of active signals.

\paragraph{Step~7: The auction is essentially a dark first-price auction.}
If \(c=1\), there is no sale, which is equivalent to a trivial dark first-price auction with a sufficiently high reserve to preclude trade. Henceforth assume \(c < 1\). The function \(\beta(\cdot)\) cannot jump: a buyer with a signal just above a jump could profitably report just below it. Formally, for any \(x\in(c,1)\), let an active signal \(\theta_i\) decrease to \(x\) and an active report \(\hat\theta_i\) increase to \(x\) in \eqref{eq:interdependent-report-incentive}. Continuity of \(G,V\) and \(G(x\mid x)>0\) give
\(
  \lim_{\theta_i\downarrow x}\beta(\theta_i)
  \leq
  \lim_{\hat\theta_i\uparrow x}\beta(\hat\theta_i)
\),
while monotonicity gives the reverse inequality. Hence, \(\beta\) extends to a continuous, strictly increasing function on \([c,1]\). Let \(\rho:=\beta(c)\), so
\(
  \beta([c,1])=[\rho,\beta(1)]
\).

If \(c>0\), let inactive signals approach \(c\) from below and active reports approach \(c\) from above in \eqref{eq:interdependent-report-incentive}. This gives
\(
  0\geq G(c\mid c)[V(c,c)-\rho]
\),
and hence \(\rho\geq V(c,c)\). Together with \(\rho\leq v^1(c)\leq V(c,c)\), we obtain
\(
  \rho=V(c,c)=v^1(c)
\).
Thus, a cutoff \(c\in(0,1)\) is possible only if \(V(c,c)=v^1(c)\). In environments where \(V(c,c)>v^1(c)\) for every \(c\in(0,1)\), as in Example~\ref{ex:affiliated}, every auction with a positive probability of sale must have \(c=0\), in which case \(0\leq\rho\leq v^1(0)\).

Consider now the dark first-price auction with reserve \(\rho\) and bidding strategy
\[
  S_{\mathrm{FP}}^d(\theta_i)
  =
  \begin{cases}
    \varnothing,&\theta_i<c,\\
    \beta(\theta_i),&\theta_i\geq c.
  \end{cases}
\]
By continuity, \eqref{eq:interdependent-report-incentive} extends to the cutoff and excluded signals. Every bid in \([\rho,\beta(1)]\) therefore corresponds to a signal report in \([c,1]\), and BIC rules out profitable deviations to such bids. Any bid above \(\beta(1)\) is weakly dominated by bidding \(\beta(1)\), since it cannot increase the probability of winning and entails a higher payment whenever the buyer wins. Abstention yields zero and cannot improve on the prescribed nonnegative utility. Thus \(S_{\mathrm{FP}}^d\) is a symmetric equilibrium in the dark first-price auction with reserve \(\rho\). By Steps~4--6, the allocation probabilities and expected payments induced by this dark first-price auction coincide almost everywhere, conditional on the full signal profile, with those induced by \((\Gamma^d,S^d)\). Hence, \((\Gamma^d,S^d)\) is essentially a dark first-price auction (with reserve \(\rho\)).\qed

\section{Extensive-Form Games}\label{app:extensive-form-games}

We formalize the games introduced in Section~\ref{sec:game}. Consider an anonymous finite-depth \(n\)-player extensive-form game \(\Gamma^n\) with perfect recall:
\[
  \Gamma^n
  =
  \left(
    H^n,Z^n,\iota^n,A^n,(\mathcal I_i^n)_{i=1}^n
  \right).
\]
Let \(H^n\) denote the set of action histories, ordered by the strict precedence relation \(\prec\), and let \(Z^n\subseteq H^n\) denote the set of terminal histories. At each nonterminal history \(h\), player \(\iota^n(h)\in[n]:=\{1,\ldots,n\}\) chooses an action from \(A^n(h)\). For each player \(i\), let \(\mathcal I_i^n\) denote their information partition of
\[
  H_i^n:=\{h\in H^n\setminus Z^n:\iota^n(h)=i\}.
\]
The feasible action set is constant within each information set \(I_i\in\mathcal I_i^n\), and we denote it by \(A^n(I_i)\).

The outcome rule \(\varphi^n:Z^n\to\Delta(X^n)\) maps each terminal history to a lottery over outcomes in \(X^n\).


\paragraph{Strategies.}

Define the space of history-contingent action plans for player \(i\) by
\[
  \mathcal C_i^n:=\prod_{h\in H_i^n}\Delta(A^n(h)),
  \qquad
  \mathcal C^n:=\prod_{i=1}^n\mathcal C_i^n.
\]
An action plan assigns an action distribution to every decision history, without imposing the player's information restrictions. Player \(i\)'s interim strategy space consists of those plans that respect their information partition:
\[
  \Sigma_i^n
  :=
  \left\{
    c_i^n\in\mathcal C_i^n:
    c_i^n(h)=c_i^n(h')
    \text{ whenever }h,h'\in I_i
    \text{ for some }I_i\in\mathcal I_i^n
  \right\}.
\]
By anonymity, these strategy spaces are identical up to relabeling, and we write \(\Sigma^n:=\Sigma_i^n\).

An action-plan profile \(c^n=(c_i^n)_{i=1}^n\in\mathcal C^n\) induces a distribution over terminal histories,
\[
  \tau^n:\mathcal C^n\to\Delta(Z^n).
\]
We extend \(\varphi^n\) to lotteries over terminal histories in the usual way and write the resulting outcome lottery as \(\varphi^n(\tau^n(c^n))\). In the main text, we abbreviate this expression as \(\varphi^n(c^n)\).

\paragraph{Joint strategies.}

Suppose a single player \(i\) controls multiple identities. As in the main text, let
\[
  K=\{i,n+1,\ldots,n+k\}\subseteq[n+k]
\]
denote the set of identities they control. The player may pool information across these identities and coordinate their actions. Let
\[
  H_K^{n+k}:=\bigcup_{j\in K}H_j^{n+k}.
\]
For each \(h\in H_K^{n+k}\), let \(r_K^{n+k}(h)\) denote the information jointly available to the controlled identities immediately before the action at \(h\) is chosen.

For \(c_K^{n+k}=(c_j^{n+k})_{j\in K}\), write
\[
  c_K^{n+k}(h):=
  c_{\iota^{n+k}(h)}^{n+k}(h).
\]
The joint interim strategy space is
\[
  \boldsymbol\Sigma_K^{n+k}
  :=
  \left\{
    c_K^{n+k}\in\prod_{j\in K}\mathcal C_j^{n+k}:
    c_K^{n+k}(h)=c_K^{n+k}(h')
    \text{ whenever }
    r_K^{n+k}(h)=r_K^{n+k}(h')
  \right\}.
\]
Thus, the action chosen through any controlled identity may depend on all information available across the identities controlled by the same player. In particular,
\[
  \prod_{j\in K}\Sigma_j^{n+k}
  \subseteq
  \boldsymbol\Sigma_K^{n+k}
  \subseteq
  \prod_{j\in K}\mathcal C_j^{n+k},
\]
and the first inclusion may be strict. Against ordinary strategies \(\sigma_{-K}^{n+k}\) of the remaining players, a joint strategy \(\boldsymbol\sigma_K^{n+k}\in\boldsymbol\Sigma_K^{n+k}\) induces the outcome lottery
\[
  \varphi^{n+k}\!\left(
    \tau^{n+k}(\boldsymbol\sigma_K^{n+k},\sigma_{-K}^{n+k})
  \right).
\]

In the main text, we use the product notation \((\Sigma^{n+k})^{k+1}\) for expositional simplicity. In a dynamic auction, however, a player controlling multiple identities may pool information across those identities and coordinate their actions. Accordingly, whenever multiple identities are controlled by the same player, we interpret this product notation as the full admissible joint-strategy space \(\boldsymbol\Sigma_K^{n+k}\) defined above, and impose all notions of identity compatibility with respect to this larger space. This interpretation does not affect our results. The impossibility results continue to hold because their proofs use particular deviations that already belong to the product strategy space and hence remain feasible in the larger joint-strategy space. The positive results for the dark first-price auction and the posted-price mechanism are unaffected because these mechanisms are static. The same conclusion holds for dynamic implementations such as the dark Dutch auction, in which pooling information across controlled identities provides no additional profitable deviations.

\section{Additional Examples}

\setcounter{example}{0}
\counterwithin{example}{section}

\begin{example}\label{ex:non-second-price}
  This example shows that, even under regularity, an optimal lit auction that is Bayesian buyer identity-compatible, Bayesian seller identity-compatible, and ex-post IR need not be a second-price auction.

  Let \(F(\theta_i)=\theta_i\), and consider the (lit) second-price auction with reserve \(\rho^*=1/2\). Modify only the payment rule in the two-buyer component. Define
  \begin{align*}
    g(s)&=\one\{0.60<s<0.65\}-\one\{0.65<s<0.70\},\\
    h(s)&=\one\{0.80<s<0.85\}-\one\{0.85<s<0.90\}.
  \end{align*}
  Retain the optimal allocation rule and set
  \[
    t_i^{\Gamma^{2},S^2}(\theta_i,\theta_j)
    =
    q_i^{\Gamma^{2},S^2}(\theta_i,\theta_j)
    \left[
      \max\{\rho^*,\theta_j\}
      +\frac{1}{20}h(\theta_i)g(\theta_j)
    \right],
    \qquad j\ne i.
  \]
  Thus, the winner pays the bracketed amount and the loser pays zero. At every other population size, use the usual second-price payment rule.

  The perturbation averages to zero both over the losing value for each fixed winning value and over the winning value for each fixed losing value. It therefore preserves interim expected payments and hence BIC. Bayesian seller identity compatibility is also preserved because the perturbation leaves expected revenue unchanged both without shills and for every fixed shill bid: the increases and decreases in payments cancel in expectation. Thus, the perturbation creates no additional incentive for the seller to shill bid. Bayesian buyer identity compatibility is preserved as well: deviations generating at least three identities face the usual second-price payment rule, while a singleton buyer using two identities still pays at least the reserve whenever they win. Finally, the perturbation is small enough that every winning payment lies between the reserve and the winner's value, so ex-post IR holds. Nevertheless, the payment rule differs from that of the second-price auction on a set of positive probability.
\end{example}

\begin{example}\label{ex:dark-family}
  This example shows that, under regularity, the optimal dark first-price and all-pay auctions are special cases of a larger family of optimal auctions satisfying ex-post seller identity compatibility and Bayesian buyer identity compatibility, but not necessarily ex-post IR.

  Retain the optimal allocation rule, where the object is allocated to the highest-value buyer whose value is above the Myersonian reserve \(\rho^*\). Write \(\Psi(z)=\sum_{n\in\N}p_n z^{n-1}\), and let \(T^*(\theta_i)\) denote the optimal interim expected payment rule:
  \[
    T^*(\theta_i)=
    \begin{cases}
      0,&\theta_i \leq \rho^*,\\
      \displaystyle\theta_iG(\theta_i)-\int_{\rho^*}^{\theta_i}G(s)\dd s,&\theta_i>\rho^*.
    \end{cases}
  \]
  For any measurable \(h:[0,1]\to[0,1]\), define
  \[
    B_h(\theta_i)
    =
    F(\theta_i)+\int_{\theta_i}^1 h(s)\dd F(s),
    \qquad
    \beta_h(\theta_i)
    =
    \frac{T^*(\theta_i)}{\Psi(B_h(\theta_i))}.
  \]
  Charge each buyer the deterministic payment
  \[
    t_i^{\Gamma^{d,n},S^d}(\theta^n)
    =
    \beta_h(\theta_i)
    \prod_{j:\,\theta_j>\theta_i}h(\theta_j),
  \]
  where an empty product equals one. Thus, \(h\equiv0\) gives essentially the optimal dark first-price auction, while \(h\equiv1\) gives essentially the optimal dark all-pay auction. As another concrete example, take \(h(s)=s\). With two reports \(\rho^*<\theta_i<\theta_j\), buyer \(i\) loses but pays \(\beta_h(\theta_i)\theta_j>0\).

  For every \(h\),
  \[
    \E_{\tilde\theta_{-i}^{\tilde n}}\!\left[
      \prod_{j:\,\tilde\theta_j>\theta_i}h(\tilde\theta_j)
    \right]
    =
    \Psi(B_h(\theta_i)),
  \]
  so the construction preserves the optimal interim expected payment rule and hence BIC and optimality. Shill bidding by the seller can only weakly reduce buyers' payments because \(h(\cdot)\leq1\). A buyer who uses multiple identities obtains the same aggregate allocation as under their highest report and makes a weakly larger total payment.\footnote{Ties with other buyers occur with probability zero.} Hence, every auction in this family is ex-post seller identity-compatible and Bayesian buyer identity-compatible.
\end{example}

\begin{example}\label{ex:inactive-identity-neutrality}
  This example shows that inactive-identity neutrality cannot be dropped from Theorem~\ref{thm:lit-irregular-bb}. We construct an optimal lit auction that is both ex-post seller identity-compatible and Bayesian buyer identity-compatible, but is not inactive-identity neutral.

  Set \(a=1/8\), \(b=1/2\), and \(c=1/4\), and let
  \[
    1-F(\theta_i)=
    \begin{cases}
      1-4\theta_i,&0\leq\theta_i\leq a,\\[2pt]
      \dfrac{1}{16\theta_i},&a\leq\theta_i\leq b,\\[5pt]
      \dfrac{1-\theta_i}{4},&b\leq\theta_i\leq1.
    \end{cases}
  \]
  The density is continuous and positive on \([0,1]\), and the virtual value is
  \[
    \phi(\theta_i)=
    \begin{cases}
      2\theta_i-\dfrac14,&0\leq\theta_i<a,\\[3pt]
      0,&a\leq\theta_i\leq b,\\[2pt]
      2\theta_i-1,&b<\theta_i\leq1.
    \end{cases}
  \]
  Thus, the virtual value is nondecreasing but not strictly increasing. A second-price auction with any reserve in \([a,b]\) is optimal. Let \(\gamma_n=15/16^n\) for every \(n\in\N\). The population distribution has full support and finite mean \(\lambda=16/15\), and \(p_1=\frac{\gamma_1}{\lambda}=\frac{225}{256}>\frac67\).

  Consider the following direct lit auction under truthful reporting. With one bidder, offer the object at price \(a\). With \(n\geq2\) bidders, allocate the object to the bidder with the highest report, provided that this report is at least \(b\), breaking ties uniformly, and retain it if no report reaches \(b\). For every \(n\in\N\), define
  \[
    B_n(\theta_i)=
    \begin{cases}
      0,&\theta_i<b,\\[2pt]
      \displaystyle
      (\theta_i-c)F(\theta_i)^{n-1}
      -\int_b^{\theta_i}F(s)^{n-1}\dd s,
      &\theta_i\geq b.
    \end{cases}
  \]
  For \(n\geq2\), charge
  \[
    t_i^{\Gamma^n,S^n}(\theta^n)
    =
    c q_i^{\Gamma^n,S^n}(\theta^n)+B_n(\theta_i).
  \]
  Thus, every bidder pays \(B_n\) evaluated at their report regardless of whether they win, and the winner pays an additional \(c\). The function \(B_1\) is used only for the revenue comparison below; the actual one-buyer auction charges the posted price \(a\).


  For every \(n\geq2\), the interim expected payment is zero below \(b\), while for \(\theta_i\geq b\),
  \[
    T_i^{\Gamma^n,S^n}(\theta_i)
    =
    cF(\theta_i)^{n-1}+B_n(\theta_i)
    =
    \theta_iF(\theta_i)^{n-1}
    -\int_b^{\theta_i}F(s)^{n-1}\dd s.
  \]
  The envelope formula establishes BIC and IR. Since the allocation maximizes virtual surplus and the lowest type obtains zero utility, the auction is optimal at every population size.

  To verify ex-post seller identity compatibility, write
  \[
    P_n:=1-F(b)^n,
    \qquad
    e_n:=\int_0^1B_n(\theta_i)\dd F(\theta_i),
    \qquad
    R_n:=\Rev(\Gamma^n,S^n).
  \]
  We have \(R_n=cP_n+ne_n\) for every \(n\in\N\), including \(n=1\), because the singleton prices \(a\) and \(b\) yield the same revenue:
  \[
    a[1-F(a)]=b[1-F(b)]=\frac1{16}.
  \]
  Note that
  \[
    e_n
    =
    \int_b^1[\phi(\theta_i)-c]F(\theta_i)^{n-1}\dd F(\theta_i).
  \]
  Since \(F(\theta_i)=(3+\theta_i)/4\) above \(b\), substituting \(u=F(\theta_i)\) yields
  \[
    e_n-e_{n+1}
    =
    \int_{7/8}^1
    \left(8u-\frac{29}{4}\right)(1-u)u^{n-1}\dd u.
  \]
  The factor \((8u-29/4)(1-u)\) is negative below \(29/32\) and positive above it. Hence,
  \[
    e_n-e_{n+1}
    \geq
    \left(\frac{29}{32}\right)^{n-1}(e_1-e_2)
    =
    \left(\frac{29}{32}\right)^{n-1}\frac1{1536}
    >0.
  \]

  Now fix \(n\) buyers and \(k\geq1\) fictitious identities. For every value profile and every profile of fictitious reports,
  \[
    \sum_{i=1}^n
    t_i^{\Gamma^{n+k},S^{n+k}}(\theta^n,\hat\theta^k)
    \leq
    c\one\left\{\max_{i\leq n}\theta_i\geq b\right\}
    +\sum_{i=1}^nB_{n+k}(\theta_i).
  \]
  Taking all fictitious reports below \(b\) attains this bound. Thus, even when the seller chooses the fictitious reports after observing the buyers' values, their maximum expected revenue is
  \[
    cP_n+ne_{n+k}
    \leq
    cP_n+ne_n
    =
    R_n.
  \]
  The auction is therefore ex-post seller identity-compatible.

  To verify Bayesian buyer identity compatibility, consider a buyer with value \(\theta_i\) who creates at least one fictitious identity. The resulting auction always has at least two bidders. Since every \(B_n\) is nonnegative, the buyer's total payment through the identities they control is at least \(c\) times their aggregate allocation. Their deviation payoff is consequently at most \((\theta_i-c)_+\), regardless of the number of fictitious identities or their reports. In contrast, the equilibrium payoff satisfies
  \[
    \sum_{n\in\N}p_nU_i^{\Gamma^n,S^n}(\theta_i)
    \geq
    p_1(\theta_i-a)_+
    \geq
    (\theta_i-c)_+,
  \]
  where the last inequality follows from \(p_1>\frac{1-c}{1-a}=\frac67\).
  Thus shill bidding is not profitable for buyer \(i\) in expectation.

  Finally, a singleton buyer with value \(\theta_i\in(a,b)\) receives the object at price \(a\). Adding an identity with a report below \(a\) changes the applicable reserve to \(b\), so the object is retained. The added identity is inactive, but it changes the original buyer's allocation. Hence the auction is not inactive-identity neutral.
\end{example}

\begin{example}\label{ex:non-random}
  This example shows that a posted-price mechanism need not maximize revenue if ``strong'' is dropped from either strong ex-post seller identity compatibility or strong ex-post IR in Theorem~\ref{thm:lit-posted-price}.

  \emph{(i) Retaining only strong ex-post IR.} Let \(F(\theta_i)=\theta_i\), and consider the following direct lit auction. Reports of at least \(1/2\) are eligible. With one buyer, offer the object at price \(1/2\). With two buyers, retain the object if neither is eligible and sell it at price \(1/2\) to the unique eligible buyer if there is one. If both are eligible and \(\theta_i>\theta_j\), use the lottery
  \[
    \begin{array}{c|cc}
      \text{Winner}&\text{Probability}&\text{Payment upon winning}\\ \hline
      i&2/3&\theta_j/2+1/4\\
      j&1/3&1/2.
    \end{array}
  \]
  At an eligible tie, allocate uniformly and charge the winner \((\theta_i+1)/3\). Losers pay zero. With three or more buyers, there is no sale and no payment.

  It is straightforward to verify that the auction is strategy-proof and hence BIC. Every payment conditional on winning is at most the winner's value, so strong ex-post IR holds. To verify ex-post seller identity compatibility, note that adding one shill to a singleton buyer yields expected revenue of at most \(1/2\) if the buyer is eligible and zero otherwise. Any deviation resulting in three or more bidders yields zero revenue. Strong ex-post seller identity compatibility, however, fails: when \(\theta_i>\theta_j>1/2\), the seller benefits from selecting the higher-revenue outcome in which \(i\) wins, even without adding a shill.

  \emph{(ii) Retaining only strong ex-post seller identity compatibility.} For the second auction, retain the allocation rule from the first auction but replace each buyer's random payment by its expectation conditional on the reports. Thus, when both buyers are eligible and \(\theta_i>\theta_j\), charge \(i\) the amount \(\theta_j/3+1/6\) and \(j\) the amount \(1/6\), regardless of who wins. At an eligible tie, charge each buyer \((\theta_i+1)/6\), again regardless of who wins. All other payments remain unchanged.

  Because allocation probabilities and expected payments are unchanged, the second auction remains strategy-proof and satisfies ex-post IR: each buyer's expected utility conditional on the value profile is nonnegative. Strong ex-post IR, however, fails because, when both buyers are eligible, the loser makes a positive payment.

  The second auction satisfies strong ex-post seller identity compatibility. Payments are deterministic conditional on reports, so the seller cannot profit by manipulating the auction's randomization. Adding one shill to a singleton buyer yields revenue of at most \(1/2\) if the buyer is eligible and zero otherwise, while any deviation resulting in three or more bidders yields zero revenue.

  \emph{Revenue comparison.} The two auctions generate the same expected revenue. With two buyers, this revenue is \(7/18\), which exceeds the maximal posted-price revenue \(2/(3\sqrt3)\). For an explicit full-support population distribution, let \(\varepsilon=1/1000\), \(\gamma_1=\varepsilon/2\), \(\gamma_2=1-\varepsilon\), and \(\gamma_n=\varepsilon2^{-(n-1)}\) for \(n\geq3\). For either auction,
  \[
    \Rev(\Gamma^{\N},S^{\N})
    =
    (1-\varepsilon)\frac7{18}+\frac{\varepsilon}{8}
    >
    (1-\varepsilon)\frac2{3\sqrt3}+\varepsilon,
  \]
  where the right-hand side is an upper bound on the revenue of any posted-price mechanism, even allowing the posted price to depend on the realized population size.
\end{example}

\begin{example}\label{ex:non-uniform}
  This example shows that an anonymous posted-price mechanism need not select uniformly among buyers willing to pay the posted price. Fix \(\rho\in(0,1)\), and partition the eligible values \([\rho,1]\) into three sets \(A,B,C\), each having conditional probability \(1/3\). Consider a direct lit auction that selects uniformly among reports of at least \(\rho\), except when there are exactly two buyers and both are eligible. In that case, use the cyclic rule: \(A\) beats \(B\), \(B\) beats \(C\), and \(C\) beats \(A\); if both reports belong to the same group, select uniformly. The winner pays \(\rho\), every other buyer pays zero, and the object is retained if no report is eligible. Thus, the mechanism is essentially a posted-price mechanism, but its rationing rule is not uniform.

  The rule is anonymous, yet when the two eligible reports belong to different groups, one buyer wins with probability one. Nevertheless, every eligible report has the same interim winning probability. Conditional on facing an eligible opponent, it beats one group, loses to another, and wins with probability \(1/2\) against its own group. Hence its winning probability is
  \[
    \frac13+\frac13\cdot\frac12=\frac12.
  \]
  Thus, all eligible reports induce the same interim expected allocation, establishing BIC. It is also straightforward to verify that the mechanism is strongly ex-post IR and strongly ex-post seller identity-compatible.
\end{example}

\begin{example}\label{ex:irregular}
  This example shows that, under irregularity in the IPV benchmark, either conditional efficiency or sale consistency can fail even among dark auctions satisfying ex-post seller identity compatibility, Bayesian buyer identity compatibility, and ex-post IR. Dropping either restriction can strictly increase revenue relative to every dark first-price auction with feasible bid space \([r,+\infty)\), regardless of the reserve \(r\).

  Let \(\gamma_n=2^{-n}\), so \(G(\theta_i)=[2-F(\theta_i)]^{-2}\). Let \(f\) be the continuous piecewise-linear density specified by
  \[
    \begin{array}{c|rrrrrr}
      \theta_i&0&0.499&0.501&0.749&0.751&1\\ \hline
      f(\theta_i)&1.60&1.60&0.08&0.08&0.72&0.72.
    \end{array}
  \]
  Set \(\rho=5/16\), \(a=0.51\), and \(b=0.74\), and define
  \[
    Q(\theta_i)=
    \begin{cases}
      0,&\theta_i \leq \rho,\\
      G(\theta_i),&\rho<\theta_i<a,\\
      G(a),&a\leq\theta_i<b,\\
      G(\theta_i),&\theta_i\geq b.
    \end{cases}
  \]
  Let the winning payment be
  \[
    \beta(\theta_i)
    =
    \theta_i-
    \frac{\int_0^{\theta_i}Q(s)\dd s}{Q(\theta_i)}
    \quad\text{for }\theta_i>\rho,
    \qquad
    \beta(\theta_i)=0
    \quad\text{otherwise}.
  \]
  The interim expected allocation rule \(Q\) is nondecreasing, and the induced interim expected payment rule satisfies the envelope formula.

  \emph{(i) Conditional efficiency without sale consistency.} Consider a direct dark auction that allocates the object to the buyer with the unique highest report whenever it is strictly above \(\rho\), except that no sale occurs when the highest report lies in \([a,b)\) and another report is at least \(a\). Only the winner pays:
  \[
    t_i^{\Gamma^{d,n},S^d}(\theta^n)
    =
    \beta(\theta_i)q_i^{\Gamma^{d,n},S^d}(\theta^n).
  \]
  This auction induces interim expected allocation \(Q\), and hence is BIC. Since \(0\leq\beta(\theta_i)\leq\theta_i\), it satisfies ex-post IR. Shill bidding by the seller can only cancel a sale or replace the winner; it cannot increase the payment collected from a buyer. For a buyer, additional lower reports can only reduce the winning probability of their highest report while leaving its winning payment unchanged. Together with BIC, this gives Bayesian buyer identity compatibility.

  The auction is conditionally efficient but violates sale consistency: a singleton report in \((a,b)\) wins, whereas adding a lower report in \([a,b)\) can cancel the sale. Let \(R_{\mathrm{FP}}(r)\) denote revenue from a dark first-price auction with feasible bid space \([r,+\infty)\). For this density, the revenue-maximizing reserve is \(r=\rho\). Since \(\phi(s)=2s-3<0\) on \([a,b]\),
  \[
    \Rev(\Gamma^d,S^d)-\max_rR_{\mathrm{FP}}(r)
    =
    2\int_a^b\phi(s)\left[G(a)-G(s)\right]f(s)\dd s
    >
    0.
  \]
  The respective revenues are approximately \(0.2514\) and \(0.2508\). Thus, canceling sales in this region raises revenue by reducing allocation to buyers with negative virtual values.

  \emph{(ii) Sale consistency without conditional efficiency.} Retain the same interim expected allocation \(Q\) and winning payment \(\beta\), but replace the cancellation rule with a lottery determined only by the highest report. For a value profile \(\theta^n\), let
  \[
    \overline{\theta^n}:=\max_j\theta_j,
    \qquad
    M:=\{j:\theta_j=\overline{\theta^n}\},
    \qquad
    \alpha(\overline{\theta^n})
    :=
    \frac{Q(\overline{\theta^n})}{G(\overline{\theta^n})}.
  \]
  With probability \(\alpha(\overline{\theta^n})\), select a winner uniformly from \(M\) and charge them \(\beta(\overline{\theta^n})\); otherwise retain the object. Thus,
  \[
    q_i^{\Gamma^{d,n},S^d}(\theta^n)
    =
    \frac{\alpha(\overline{\theta^n})}{|M|}
    \one\{i\in M\},
    \qquad
    t_i^{\Gamma^{d,n},S^d}(\theta^n)
    =
    \beta(\overline{\theta^n})q_i^{\Gamma^{d,n},S^d}(\theta^n).
  \]
  This lottery preserves the same interim expected allocations and payments, and therefore preserves BIC and the strict revenue comparison above. Ex-post IR also holds. It is straightforward to verify that the auction satisfies ex-post seller identity compatibility and Bayesian buyer identity compatibility.

  Adding lower reports leaves the allocation probability unchanged, so the auction is sale-consistent. It is not conditionally efficient: when \(\overline{\theta^n}\in(a,b)\), the object is sold only with probability \(G(a)/G(\overline{\theta^n})\in(0,1)\). Again, reducing the sale probability in this region raises revenue by reducing allocation to buyers with negative virtual values.

  Thus, dropping either conditional efficiency or sale consistency permits an auction satisfying the remaining requirements to strictly outperform every dark first-price auction.
\end{example}

\begin{example}\label{ex:affiliated}
  This example shows that the assumptions of Theorem~\ref{thm:interdependent} can hold with strictly affiliated signals and positive interdependence.

  Draw an unobserved common state \(Z\) uniformly on \([0,1]\). Conditional on \(Z=z\), draw the number of buyers from a Poisson distribution with mean
  \[
    \Lambda(z):=\int_0^1e^{zs}\dd s
    =
    \begin{cases}
      (e^z-1)/z,&z>0,\\
      1,&z=0.
    \end{cases}
  \]
  Given \(Z=z\) and \(\tilde n=n\), draw the \(n\) signals independently according to the density \(e^{z\theta}/\Lambda(z)\) on \(\Theta = [0,1]\). Equivalently, conditional on \(Z=z\), the buyers' signals form a Poisson point process on \(\Theta=[0,1]\) with intensity \(e^{z\theta}\dd\theta\): the number of buyers with signals in any interval \(I\) is Poisson with mean \(\int_Ie^{z\theta}\dd\theta\), independently across disjoint intervals. The distribution of the number of buyers has full support, and \(\lambda=\E[\Lambda(Z)]\leq e\).

  Conditional on \(\tilde n=n\), the joint signal density is symmetric, positive, and continuous, with
  \[
    f_n(\theta^n)
    \propto
    \int_0^1
    \exp\!\left(-\Lambda(z)+z\sum_{j=1}^n\theta_j\right)\dd z.
  \]
  For \(i\neq j\),
  \[
    \frac{\partial^2\log f_n(\theta^n)}
    {\partial\theta_i\,\partial\theta_j}
    =
    \operatorname{Var}(Z\mid\tilde n=n,\theta^n)>0.
  \]
  The inequality is strict because the conditional density of \(Z\) is positive throughout \([0,1]\). Hence, signals are strictly affiliated conditional on every \(n\geq2\).

  Let
  \[
    v^n(\theta_i,\theta^n_{-i})
    =
    \theta_i+\eta\max_{j\neq i}\theta_j,
    \qquad 0<\eta<1/e,
  \]
  where the maximum of the empty set is zero. These values are continuous, strictly increasing in the buyer's own signal, and nondecreasing in the other buyers' signals. They satisfy \(v^n(\theta_i,\theta^n_{-i})\geq v^1(\theta_i)=\theta_i\), with equality when all opponents' signals are zero.

  Under buyer \(i\)'s posterior \(K_{\theta_i}\), the conditional density of \(Z=z\) is
  \[
    \frac{e^{z\theta_i}}{\ell(\theta_i)},
    \qquad
    \ell(\theta_i)=\int_0^1e^{u\theta_i}\dd u.
  \]
  Conditional on \(Z=z\), the other buyers' signals form a Poisson point process with intensity \(e^{z\theta}\dd \theta\). Let \(Y\) denote the highest signal among the other buyers, with \(Y=-\infty\) when there is none, and define \(A_r(z):=\int_r^1e^{zs}\dd s\). Since \(Y<r\) is equivalent to no other buyer having a signal in \([r,1]\),
  \[
    G(r\mid\theta_i)
    =
    \frac{\int_0^1e^{z\theta_i-A_r(z)}\dd z}
    {\int_0^1e^{z\theta_i}\dd z}.
  \]
  Consequently,
  \[
    h(r\mid\theta_i)
    :=\partial_r\log G(r\mid\theta_i)
    =
    \E[e^{Zr}\mid\theta_i,Y<r],
    \qquad
    1\leq h(r\mid\theta_i)\leq e.
  \]
  The posterior density of \(Z\) conditional on \(\theta_i\) and \(Y<r\) is proportional to \(e^{z\theta_i-A_r(z)}\). Hence,
  \[
    \partial_{\theta_i}h(r\mid\theta_i)
    =
    \operatorname{Cov}(e^{Zr},Z\mid\theta_i,Y<r)
    \geq0.
  \]
  Moreover, integration by parts gives
  \[
    V(\theta_i,r)
    =
    \theta_i+\eta\left[
      r-\frac{\int_0^rG(s\mid\theta_i)\dd s}
      {G(r\mid\theta_i)}
    \right].
  \]
  These expressions establish the required joint continuity of \(G\) and \(V\), including \(V(\theta_i,0)=\theta_i\).

  Consider a dark first-price auction without a reserve. Let \(\beta:\Theta\to[0,\infty)\) denote a candidate symmetric bidding strategy, and define
  \[
    a(\theta_i):=h(\theta_i\mid\theta_i),
    \qquad
    H(\theta_i):=\int_0^{\theta_i}a(s)\dd s.
  \]
  The first-order condition for a buyer with signal \(\theta_i\) to bid \(\beta(\theta_i)\), together with the boundary condition, is
  \[
    \beta'(\theta_i)
    =
    a(\theta_i)\bigl[(1+\eta)\theta_i-\beta(\theta_i)\bigr],
    \qquad
    \beta(0)=0.
  \]
  Since \(a\) is continuous on \([0,1]\), the integrating factor \(e^{H(\theta_i)}\) gives the unique continuously differentiable solution
  \[
    \beta(\theta_i)
    =
    (1+\eta)e^{-H(\theta_i)}
    \int_0^{\theta_i}s\,a(s)e^{H(s)}\dd s.
  \]
  Moreover,
  \[
    \theta_i-\beta(\theta_i)
    =
    e^{-H(\theta_i)}
    \int_0^{\theta_i}
    \bigl[1-\eta s\,a(s)\bigr]e^{H(s)}\dd s
    >0,
    \qquad \theta_i>0,
  \]
  because \(a(s)\leq e\) and \(\eta e<1\). The first-order condition therefore implies \(\beta'(\theta_i)>0\) for \(\theta_i>0\). Every winning buyer pays less than their own signal, so ex-post IR holds.

  It remains to verify BIC. Let \(U(r\mid\theta_i)\) denote the expected utility of a buyer with signal \(\theta_i\) who bids \(\beta(r)\), while all other buyers follow \(\beta\). Then
  \[
    U(r\mid\theta_i)
    =
    \bigl[V(\theta_i,r)-\beta(r)\bigr]G(r\mid\theta_i).
  \]
  Differentiating and using the first-order condition gives
  \[
    \begin{aligned}
      \frac{\partial_rU(r\mid\theta_i)}{G(r\mid\theta_i)}
      &=
      h(r\mid\theta_i)(\theta_i-r)\\
      &\quad+
      \bigl[h(r\mid\theta_i)-h(r\mid r)\bigr]
      \bigl[(1+\eta)r-\beta(r)\bigr].
    \end{aligned}
  \]
  Since \(h(r\mid\theta_i)\) is nondecreasing in \(\theta_i\) and \((1+\eta)r-\beta(r)>0\) for \(r>0\), utility is increasing for \(r<\theta_i\) and decreasing for \(r>\theta_i\). Thus, bidding \(\beta(\theta_i)\) is optimal among bids in \([0,\beta(1)]\). Bids above \(\beta(1)\) cost more without increasing the probability of winning, while abstention cannot improve on the buyer's nonnegative equilibrium utility. Hence, \(\beta\) is a symmetric equilibrium, establishing BIC.

  The dark first-price auction, as well as its corresponding direct dark implementation, is conditionally efficient and sale-consistent. It is also straightforward to verify that both satisfy ex-post seller identity compatibility and Bayesian buyer identity compatibility.
\end{example}

\end{document}